\documentclass{amsart}
\usepackage{amsmath,amssymb,bm}
\newtheorem{theorem}{Theorem}[section]
\newtheorem{lemma}[theorem]{Lemma}
\newtheorem{proposition}[theorem]{Proposition}
\newtheorem{corollary}[theorem]{Corollary}
\newtheorem{claim}{Claim}[section]
\theoremstyle{remark}
\newtheorem{remark}{Remark}[section]
\numberwithin{equation}{section}
\newcommand{\R}{\mathbb{R}}
\newcommand{\C}{\mathbb{C}}
\newcommand{\N}{\mathbb{N}}
\newcommand{\Z}{\mathbb{Z}}
\newcommand{\T}{\mathbb{T}}
\newcommand{\la}{\langle}
\newcommand{\ra}{\rangle}
\newcommand{\pd}{\partial}
\newcommand{\eps}{\varepsilon}
\newcommand{\wP}{\widetilde{P}}
\newcommand{\wR}{\widetilde{R}}
\newcommand{\wS}{\widetilde{S}}
\newcommand{\cC}{\mathcal{C}}
\newcommand{\wC}{\widetilde{\mathcal{C}}}
\newcommand{\mF}{\mathcal{F}}
\newcommand{\mL}{\mathcal{L}}
\newcommand{\mM}{\mathcal{M}}
\newcommand{\bS}{\bar{S}}
\newcommand{\bM}{\mathbb{M}}
\newcommand{\cN}{\mathcal{N}}
\newcommand{\fN}{\mathfrak{N}}
\newcommand{\ta}{\tilde{a}}
\newcommand{\tc}{\tilde{c}}

\newcommand{\tg}{\tilde{g}}
\newcommand{\tv}{\tilde{v}}
\newcommand{\tx}{\tilde{x}}
\newcommand{\bb}{\mathbf{b}}
\newcommand{\bd}{\mathbf{d}}
\newcommand{\tpsi}{\tilde{\psi}}

\DeclareMathOperator{\sech}{sech}

\DeclareMathOperator{\supp}{supp}
\DeclareMathOperator{\diag}{diag}

\DeclareMathOperator{\spann}{span}
\DeclareMathOperator{\erf}{erf}
\begin{document}
\title[Stability of line solitons for KP-II]
{Stability of line solitons for\\ the KP-II equation in $\R^2$}
\author{Tetsu Mizumachi}
\address{Faculty of Mathematics, Kyushu University,
Fukuoka 819-0395 Japan}
\email{mizumati@math.kyushu-u.ac.jp}
\thanks{This research is supported by Grant-in-Aid for Scientific
Research (No. 21540220).}
\date{}

\begin{abstract}
We prove nonlinear stability of line soliton solutions of the KP-II
equation with respect to transverse perturbations
that are exponentially localized as $x\to\infty$. We find that the 
amplitude of the line soliton converges to that of the
line soliton at initial time whereas jumps of the local phase shift of the
crest propagate in a finite speed toward $y=\pm\infty$.
The local amplitude and the phase shift of the crest of the line solitons are
described by a system of 1D wave equations with diffraction terms.
\end{abstract}
\maketitle

\tableofcontents

\section{Introduction}
\label{sec:intro}
The KP-II equation
\begin{equation}\label{eq:KPII}
\partial_x(\pd_tu+\pd_x^3u+3\pd_x(u^2))+3\partial_y^2u=0
\quad\text{for $t>0$ and $(x,y)\in \R^2$,}
\end{equation}
is a generalization to two spatial dimensions of the KdV equation
\begin{equation}
  \label{eq:KdV}
\pd_tu+\pd_x^3u+3\pd_x(u^2)=0\,,
\end{equation}
and has been derived as a model in the study of the transverse stability
of solitary wave solutions to the KdV equation
with respect to two dimensional perturbation when the surface tension
is weak or absent. 
See \cite{KP} for the derivation of \eqref{eq:KPII}.
Note that every solution of the KdV equation \eqref{eq:KdV}
is a planar solution of the KP-II equation \eqref{eq:KPII}.
\par
The global well-posedness of \eqref{eq:KPII} in $H^s(\R^2)$ ($s\ge0$)
on the background of line solitons has been studied by Molinet,
Saut and Tzvetkov \cite{MST} 
whose  proof is based on the work of Bourgain \cite{Bourgain}.
For the other contributions on the Cauchy problem of the KP-II equation,
see e.g. \cite{GPS,Hadac,HHK,IM, Tak,TT,Tz,Ukai} and the references therein.
\par
Let $$\varphi_c(x)= c\sech^2\Big(\sqrt{\frac{c}{2}}\,x\Big)\,.$$
It is well known that the $1$-soliton solution $\varphi_c(x-2ct)$ of
the KdV equation \eqref{eq:KdV}  is orbitally stable  because $\varphi_c$
is a minimizer of the Hamiltonian of \eqref{eq:KdV}
restricted on the manifold 
$\{u\in H^1(\R)\mid \|u\|_{L^2}=\|\varphi_c\|_{L^2}\}$.  See
\cite{Benjamin,Bona} and \cite{Ca-Li,GSS,We} for stability of
solitary wave solutions of Hamiltonian systems.
We remark that the KP-II equation does not fit into those
standard argument because the 
first two terms of the Hamiltonian of the KP-II equation
$$\int\big(u^2_x(t,x,y)-3(\pd_x^{-1}\pd_yu(t,x,y))^2-2u^3(t,x,y)\big)\,dxdy\,,$$
have the opposite sign.  Recently, Mizumachi and Tzvetkov (\cite{MT})
have proved that $\varphi_c(x-2ct)$ is orbitally stable as a solution
of the KP-II equation in $L^2(\R_x\times \T_y)$.
They used the B\"acklund transformation to prove that
$L^2(\R_x\times\T_y)$-stability follows from the $L^2$-stability of
the $0$-solution, which is an immediate consequence of the
conservation law of the $L^2(\R_x\times \T_y)$-norm.

\par
Unlike the perturbations which are periodic in the transverse directions,
the perturbations in $L^2(\R^2)$ does not allow phase shifts of
line solitons that are uniform in the transverse direction.
This is because the difference of any translated line solitons and itself has
infinite $L^2(\R^2)$-mass whereas the well-posedness result 
\cite{MST} tells us the perturbation to the line soliton stay in
$L^2(\R^2)$ for all the time. 
In order to analyze modulation of line solitons, we express 
solutions around the line soliton as
\begin{equation}
  \label{eq:ansatz}
u(t,x,y)=\varphi_{c(t,y)}(z)-\psi_{c(t,y)}(x-x(t,y)+4t)+v(t,x-x(t,y),y)\,,  
\end{equation}
where $c(t,y)$ and $x(t,y)$ are the local amplitude and the local phase
shift of the modulating line soliton,
$v$ is a remainder part which is expected to behave like an
oscillating tail and $\psi_{c(t,y)}$ is an auxiliary function so that
\begin{equation}
  \label{eq:balance}
\int_\R v(t,x,y)\,dx=\int_\R v(0,x,y)\,dx\quad\text{for any $t>0$,}
\end{equation}
Eq.~\eqref{eq:balance} means that if the line soliton is locally amplified,
then small waves are emitted from the rear of the line soliton.
By introducing the auxiliary function $\psi_{c(t,y)}$, we have
 $v(t)\in L^2(\R^2)$ for every $t\ge0$ and we are able to show that
the $L^2$-norm of $v$ is almost conserved. 
We find that local modulations of the amplitude and phase shift can be
described by a system of $1$-dimensional wave equations with
diffraction (viscous damping) terms, that a modulating line soliton
converges to a line soliton with the same height as the original
soliton on any compact subset of $\R^2$
(Theorem~\ref{thm:exp-stability}) and that \lq\lq jumps\rq\rq of the phase
shift of the modulating line soliton propagate toward $y=\pm\infty$
along the crest of line solitons, which makes the set of all line soliton
solutions unstable (Theorem~\ref{thm:instability}).  
\par
Using geometric optics, Pedersen (\cite{Ped}) heuristically explained
that the amplitude and the orientation of the crest are described by a
system of the Burgers equation. Since both the KP-II equation and the
Boussinesq equation are long wave models for the $3$D shallow water
waves, it is natural to expect the same phenomena for KP-II.  We find
that the first order asymptotics of $\pd_yx(t,y)$ and $c(t,y)$ around
$y=\pm (8c_0)^{1/2}t+O(\sqrt{t})$ are given by self-similar solutions of
the Burgers equations as $t\to\infty$ (Theorem~\ref{thm:Burgers}).
\par

Now let us introduce our results.
The first result is the stability of line soliton solutions for exponentially
localized perturbations.
\begin{theorem}
  \label{thm:exp-stability}
Let $c_0>0$ and $a\in (0,\sqrt{c_0/2})$. Then there exist positive constants
$\eps_0$ and $C$ satisfying the following:
if $u(0,x)=\varphi_{c_0}(x-x_0)+v_0(x)$ and 
$\eps:=\|e^{ax}v_0\|_{L^2(\R^2)}+\|e^{ax}v_0\|_{L^1_yL^2_x}+\|v_0\|_{L^2(\R^2)}
<\eps_0$, then there exist $C^1$-functions $c(t,y)$ and $x(t,y)$ such that
for $t\ge0$,
\begin{align}
\label{OS}
& \|u(t,x,y)-\varphi_{c(t,y)}(x-x(t,y))\|_{L^2(\R^2)}\le C\eps\,,\\
\label{phase-sup}
&  \sup_{y\in \R}(|c(t,y)-c_0|+|x_y(t,y)|)\le C\eps (1+t)^{-1/2}\,,\\
\label{phase2}
&  \|x_t(t,\cdot)-2c(t,\cdot)\|_{L^2}\le C\eps(1+t)^{-3/4}\,,\\
\label{AS}
& \left\|e^{ax}(u(t,x+x(t,y),y)-\varphi_{c(t,y)}(x))\right\|_{L^2}
\le C\eps(1+t)^{-3/4}\,.
\end{align}
\end{theorem}
\begin{remark}
The KP-II equation has no localized solitary waves (see \cite{dBM}).
On the other hand, the KP-I equation has stable localized solitary waves
(see \cite{dBS,LiuW}) and line solitons of the KP-I equation are unstable
(\cite{RT1,RT2,Z}).
\end{remark}

The KP-II equation \eqref{eq:KPII} is invariant under a change of variables
\begin{equation}
  \label{eq:rot}
x\mapsto x+ky-3k^2t+\gamma\quad\text{and}\quad y\mapsto y-6kt
\quad\text{for any $k$, $\gamma\in\R$,}  
\end{equation}
and has a $3$-parameter family of
line soliton solutions
$$
\mathcal{A}=
\{\varphi_c(x+ky-(2c+3k^2)t+\gamma)\mid c>0\,,\, k\,,\gamma\in\R\}\,.$$
The set of all $1$-soliton solutions of KdV or line soliton solutions of
KP-II under the $y$-periodic boundary conditions
are known to be stable in $L^2(\R^2)$ (see \cite{MV},\cite{MT}).
However the set $\mathcal{A}$  is not large enough to be stable
for the flow generated by KP-II in $L^2(\R^2)$.
\begin{theorem}
  \label{thm:instability}
Let $c_0>0$. There exists a positive constant $C$ such that for any $\eps>0$,
there exists a solution of \eqref{eq:KPII} such that
$\|u(0,x,y)-\varphi_{c_0}(x)\|_{L^2}<\eps$ and
$$\liminf_{t\to\infty}t^{-1/4}\|u(t,x,y)-\varphi_{c_0}(x)\|_{L^2(\R^2)}\ge C\eps\,.$$
\end{theorem}
\begin{remark}
If $(c,\gamma)\ne (c_0,0)$, then
$\|u(t,x,y)-\varphi_c(x-\gamma)\|_{L^2(\R^2)}=\infty$ thanks to
the well-posedness result (\cite{MST}). Thus the
\textit{orbital instability}
$$\liminf_{t\to\infty}t^{-1/4}\inf_{v\in\mathcal{A}}
\|u(t,\cdot)-v\|_{L^2(\R^2)}\ge C\eps$$
follows immediately from Theorem~\ref{thm:instability}.  
\end{remark}

Orbital instability is a consequence of finite speed propagations of
local phase shifts along the crest of the modulating line soliton.
We find that $c(t,y)$ and $\pd_yx(t,y)$ behave like a self-similar solution of
the Burgers equation around $y=\pm \sqrt{8c_0}t$.
\begin{theorem}\label{thm:Burgers}
Let $c_0=2$ and let $v_0$ and $\eps$ be the same
as in Theorem~\ref{thm:exp-stability}. 
Then for any $R>0$,
$$\left\|\begin{pmatrix} c(t,\cdot)\\ x_y(t,\cdot) \end{pmatrix}
-\begin{pmatrix}  2 & 2 \\ 1 & -1\end{pmatrix}
\begin{pmatrix}u_B^+(t,y+4t) \\ u_B^-(t,y-4t)\end{pmatrix}
\right\|_{L^2(|y\pm4t|\le R\sqrt{t})}=o(t^{-1/4})$$
as $t\to\infty$, where $u_B^\pm$ are self similar solutions of
the Burgers equation 
$$\pd_tu=2\pd_y^2u\pm 4\pd_y(u^2)$$
 such that 
$$u_B^\pm(t,y)=\frac{\pm m_\pm H_{2t}(y)}
{2\left(1+m_\pm \int_0^{y}H_{2t}(y_1)\,dy_1\right)}\,,
\quad H_t(y)=(4\pi t)^{-1/2}e^{-y^2/4t}\,,$$
and that $m_\pm$ are constants satisfying
$$\int_\R u_B^\pm(t,y)\,dy=\frac14\int_\R c(0,y)\,dy+O(\eps^2)\,.$$
\end{theorem}

Now we recall known results on stability of planar traveling wave
solutions. Stability of planar traveling waves in $L^2(\R^n)$
($n\ge 2$) were studied for reaction
diffusion equations by Xin (\cite{Xin}), Levermore and Xin
(\cite{Le-Xin}) and Kapitula (\cite{Kapitula}).
Stability of kink
solutions of Hamiltonian systems has been studied for $3$-dimensional
$\phi^4$-model by Cuccagna (\cite{Cu}).

The difficulty of those problems is that the spectrum of the linearized
operator $\mL$ around planar traveling waves has continuous
spectrum converging to $0$ whereas in the case where $n=1$, we see that
$0$ is an isolated eigenvalue
of the linearized operator around the traveling wave solution
and all the rest of the spectrum is in the left half plane and
away from the imaginary axis.
When $n\ge2$, the paper \cite{Xin} tells us that the semigroup generated by
the linearized operator decays to zero like $t^{-(n-1)/4}$.
This corresponds to the relation between our results
and the asymptotic stability result for the KdV equation by Pego and Weinstein
(\cite{PW}) where the spectrum of the linearized
operator in $L^2(\R;e^{2ax}dx)$ consists of the isolated eigenvalue $0$ and
$\sigma_c$ satisfying
$\sigma_c\subset \{\lambda\in\C\mid \Re \lambda <-b\}$
for some $b>0$. By measuring the size of perturbations with an
exponentially  weighted norm biased in the direction of motion, 
one obtains that exponential decay of the oscillating tail of the
solution for both KdV and KP-II and that leads to exponential
stability of the KdV $1$-soliton. However, thanks to the tranverse direction,
the linearized operator around a
line soliton of the KP-II equation has two branches of continuous
spectrum all the way up to $0$ in $L^2(\R^2;e^{ax}dxdy)$ with $a>0$.
We remark that those resonant modes are exponentially growing
as $x\to-\infty$ (see Lemma~\ref{lem:kpresonances}) and that the
corresponding continuous spectrum does not show up when we 
consider $L^2(\R^2)$-linear stability of the line soliton.
We refer the readers  e.g. \cite{APS,Haragus} for linear stability
of solitary waves and cnoidal waves to transverse perturbations.
\par
Since the transverse direction is $1$-dimensional,
the rate of decay of $\|\pd_y^kc(t,\cdot)\|_{L^2}$ and
$\|\pd_y^{k+1}x(t,\cdot)\|_{L^2}$ is at most $t^{-(2k+1)/4}$
and the nonlinearity of the modulation equations is quadratic,
it was fortunate that they have the similar
structure as the Burgers equations. Indeed, there are $1$D-heat equations
with quadratic nonlinearity whose solutions may not exist global in time
(\cite{Fujita}).
\par

Our plan of the present paper is as follows.  In
Section~\ref{sec:onmkp2}, we obtain explicit formula of resonant modes
of $\mL$ and $\mL^*$, where $\mL$ is a linearized operator of the
KP-II equation around the line soliton $\varphi(x-4t)$ by using the
linearized Miura transformations.  As is well known, the Miura
transformations connect line solitons and the null solution of the
KP-II equation with kink solutions of the modified KP-II equation
and all the slowly decaying eigenmodes of the linearized equation
$\pd_tu=\mL u$ can be found by investigating the kernel and the
cokernel of the linearized Miura transformation. We find two branches
of (resonant) eigenmodes $\{g(x,\pm\eta)e^{iy\eta}\}$ of $\mL$ such
that $g(x,\eta)\in L^2(\R;e^{2ax}dx)$ for a $a>0$ and
$\eta\in(-\eta_*,\eta_*)$, where $\eta_*$ is a positive number
depending on $a$.  In Section~\ref{sec:semigroup}, we prove that
solutions which are orthogonal to resonant modes of $\mL^*$ decay
exponentially in $L^2(\R^2;e^{2ax}dxdy)$ like solutions of the
linearized KP-II equation around the null solution by using a
bijection composed of the linearized Miura transformations by using an
idea of \cite{MP}.  In Section~\ref{sec:preliminaries}, we collect
linear estimates of 1D damped wave equations which shall be used to
analyze modulation equations of line solitons.  In
Section~\ref{sec:decomp}, we fix the decomposition \eqref{eq:ansatz}
by imposing that $v(t)$ is orthogonal to secular resonant modes of
$\mL^*$.  In Section~\ref{sec:modulation}, we derive modulation
equations on $c(t,y)$ and $x(t,y)$ from the non-secular conditions
introduced in Section~\ref{sec:decomp}.  Since the KP equations are
anisotropic in $x$ and $y$, the resonant eigenfunctions cannot be
written in the form $\{g(x)e^{iy\eta}\mid\eta\in\R\}$ as in the case
for reaction diffusion equations (\cite{Kapitula,Xin}) or the $\phi^4$
model (\cite{Cu}).
Moreover, the resonant eigenfunctions grow like
$g(x,\eta)\sim e^{\eta^2|x|/2}$ as $x\to-\infty$. For this reason, we
work on exponentially weighted space $X$ and impose the non-secular
conditions only for small $\eta$.  To rewrite
modulation equations of $c(t,y)$ and $x(t,y)$ in a PDE form, we
compute the time derivative of the non-secular condition, take the inverse
Fourier transform of the resulting equation. 
Although the modulation equations are non-local due to the $\eta$-dependence
of the resonant modes $g(x,\eta)$, the dominant part of the
modulation equations are damped wave equations.
Indeed, the modulation equations for the line soliton
$\varphi_{c_0}(x-2c_0t)$ with $c_0=2$ are approximately
\begin{equation}
\label{eq:approxeq}
\begin{pmatrix}  b_t \\ \tx_t  \end{pmatrix}
\simeq \begin{pmatrix}  3\pd_y^2 & 8\pd_y^2 \\ 2-\mu_3\pd_y^2 & \pd_y^2
\end{pmatrix}
\begin{pmatrix} b \\ \tx \end{pmatrix}
+\begin{pmatrix} 6(bx_y)_y \\ 3(\tx_y)^2-\frac{1}{4}b^2\end{pmatrix}\,,
\end{equation}
where $\mu_3=1/2+\pi^2/24$ and
$b(t,y)=4/3\{(c(t,y)/2)^{3/2}-1\}$ (see \eqref{eq:bdef}
for the precise definition) and $\tx(t,y)=x(t,y)-4t$.
We remark that $\pd_tx(t,y)\simeq 2c(t,y)$ and 
$b(t,y)\simeq c(t,y)-2\to 0$ as $t\to\infty$.
If we translate \eqref{eq:approxeq} into a system of $b(t,y)$ and $\pd_yx(t,y)$
and diagonalize the resulting equation, then we obtain a coupled 
Burgers equations.

In Section~\ref{sec:apriori}, we obtain $\mF^{-1}L^\infty$-$L^2$ decay
estimates on $b(t,y)$ and $\pd_yx(t,y)$ presuming a decay estimate on
$v(t)$ in $X:=L^2(\R^2;e^{2ax}dxdy)$ and the $L^2$-bound of $v(t)$.
 In Section~\ref{sec:L2norm}, we prove the $L^2$-estimate of $v$ assuming
the decay estimate on $v(t)$ in $X$.  In Section~\ref{sec:lowbound},
we estimate the low-frequency part of $v$
in $L^2(\R^2;e^{2ax}dxdy)$ by using the semigroup estimates obtained
in Section~\ref{sec:semigroup}. We estimate the high frequency part
separately in Section~\ref{sec:virial} by using the virial type
estimate to avoid the derivative loss. We remark that the potential
term produced by the linearization around the line soliton is
negligible to obtain time-global virial type estimates for the high
frequency part. In Sections~\ref{sec:thm1} and \ref{sec:thm2}, we
prove Theorems~\ref{thm:exp-stability} and \ref{thm:instability}.  In
Section~\ref{sec:thm3}, we prove Theorem~\ref{thm:Burgers} by using a
rescaling argument by Karch (\cite{Karch}).
\par
Using the inverse scattering method, Villarroel and Ablowitz (\cite{VA})
studied the Cauchy problem and stability of line solitons of the KP-II
equation. However, it is not clear from their
result that how modulations to line solitons evolve because they
did not explain in which sense line solitons are stable.
Moreover, our method does not rely on integrability of the equation
except for the linear estimate and can possibly be applied to
bidirectional models such as the Benney-Luke equation (\cite{BL,PQ}).
Our result is a first step toward $L^2(\R^2)$-stability 
of planar solitary waves for non-integrable equations
such as the generalized KP equations
(see e.g. Martel-Merle \cite{MM} for $H^1$-asymptotic stability of
1D solitary waves for gKdV).
\par
Finally, let us introduce several notations. 
For Banach spaces $V$ and $W$, let $B(V,W)$ be the space of all
linear continuous operators from $V$ to $W$ and let
$\|T\|_{B(V,W)}=\sup_{\|x\|_V=1}\|Tu\|_W$ for $T\in B(V,W)$.
We abbreviate $B(V,V)$ as $B(V)$.
For $f\in \mathcal{S}(\R^n)$ and $m\in \mathcal{S}'(\R^n)$, let 
\begin{gather*}
(\mathcal{F}f)(\xi)=\hat{f}(\xi)
=(2\pi)^{-n/2}\int_{\R^n}f(x)e^{-ix\xi}\,dx\,,\\
(\mathcal{F}^{-1}f)(x)=\check{f}(x)=\hat{f}(-x)\,,\quad
(m(D_x)f)(x)=(2\pi)^{-n/2}(\check{m}*f)(x)\,.
\end{gather*}
We use $a\lesssim b$ and $a=O(b)$ to mean that there exists a
positive constant such that $a\le Cb$. 
Various constants will be simply denoted
by $C$ and $C_i$ ($i\in\mathbb{N}$) in the course of the
calculations. We denote $\la x\ra=\sqrt{1+x^2}$ for $x\in\R$.

\begin{center}
\textbf{Acknowledgments}
\end{center}
The author would like to express his gratitude to Professor Yuji Kodama
for fruitful discussion based on \cite{KaKo}
and to R.~L.~Pego who pointed out \cite{Ped}.
A part of this work was carried out when the author stayed at
Carnegie Mellon University and University of Cergy Pontoise.
The author would like to express Professor Robert L.~ Pego and 
Nikolay Tzvetkov for their hospitality.
\bigskip

\section{The Miura transformation and resonant modes
of the linearized operator}
\label{sec:onmkp2}
In this section, we will find resonant eigenmodes of the linearized operator
around line solitons and prove exponential stability of
non-resonant modes in an exponentially weighted space
by using the linearized Miura transformations.
\par
For $p\in[1,\infty]$ and $k\in\N$, let $L^p_a(\R)=\{v\mid e^{ax}v\in L^p(\R)\}$
and $H^k_a(\R)=\{v\mid  e^{ax}v\in H^k(\R)\}$ whose norms are given by
$$\|v\|_{L^p_a(\R)}=\|e^{ax}v\|_{L^p(\R)}\,,\quad
\|v\|_{H^k_a(\R)}=\left(\sum_{j=0}^k\|\pd_x^jv\|_{L^2_a(\R)}^2\right)^{1/2}\,.$$
For any $a>0$,  we define the anti-derivative operator
$\pd_x^{-1}$ on $L^2_{\pm a}(\R)$ by
\begin{gather*}
(\pd_x^{-1}u)(x)=-\int_x^\infty u(x_1)\,dx_1 \quad\text{for $u\in L^2_a(\R)$,}
\\
(\pd_x^{-1}u)(x)=\int^x_{-\infty} u(x_1)\,dx_1 \quad\text{for $u\in L^2_{-a}(\R)$.}
\end{gather*}
The operator $\pd_x^{-1}$ is bounded on $L^2_a(\R)$. Indeed,
it follows from Young's inequality that
$\|\pd_x^{-1}\|_{B(L^2_a(\R))}=\|\pd_x^{-1}\|_{B(L^2_{-a})}=1/a$ for $a>0$.
\par
We interpret \eqref{eq:KPII} in the ``integrated'' form
\begin{equation}\label{KPII_integrated}
\pd_tu+\pd_x^3u+3\pd_x^{-1}\pd_y^2u+3\pd_x(u^2)=0\,,\quad
\end{equation}
where $\pd_x^{-1}\pd_y^2u(x,y)=- \int_x^\infty \pd_y^2u(x_1,y)\,dx_1$
in the sense of distribution, that is,
\begin{equation}
  \label{eq:anti-x-yy}
\la \pd_x^{-1}\pd_y^2u\,,\,\psi\ra =
-\int_{\R^2}\left(\int_x^\infty u(x_1,y)\,dx_1\right)\pd_y^2\psi(x,y)\,dxdy
\quad\text{for $\psi\in C_0^\infty(\R^2)$.}
\end{equation}
If $u$ is smooth and $u$, $\pd_yu\in X:=L^2(\R^2;e^{2ax}dxdy)$, then
$$\la \pd_x^{-1}\pd_y^2u\,,\,\psi\ra =
\int_{\R^2} \left(\int_x^\infty \pd_yu(x_1,y)\,dx_1\right)
\pd_y\psi(x,y)\,dxdy\,.$$
Eq.~\eqref{eq:anti-x-yy} follows from the standard definition
$\pd_x^{-1}\pd_yu=\mF^{-1}(\frac{\eta}{\xi}\hat{u}(\xi,\eta))$
when a solution is exponentially localized in the $x$-direction.
Indeed, we have $(\pd_x^{-1}u)(x,y)=(\mF^{-1}(i\xi)^{-1}\hat{u})(x,y)
=-\int_x^\infty u(x_1,y)\,dx_1$
for $u\in \{f\in C_0^\infty (\R^2)\mid \int_\R f(x,y)\,dx=0\;{}^\forall
y\in\R\}=:\mathcal{B}$. Since $\mathcal{B}$ is a dense subset of
$X\cap L^2(\R^2)$, we have \eqref{eq:anti-x-yy} for $u\in X\cap L^2(\R^2)$
by taking a limit.

We remark that \eqref{KPII_integrated} has a solution in the class
$$u(t,x,y)-\varphi_c(x-2ct)\in L^\infty_{loc}([0,\infty);X)
\cap C(\R; L^2(\R^2))$$
for initial data $u(0)\in X\cap L^2(\R^2)$ (see Appendix~\ref{ap:LWP}).

\subsection{Resonant modes}
\label{subsec:resonance}
Let $\varphi=\varphi_2$, $u(t,x,y)=\varphi(x-4t)+U(t,x-4t,y)$.
Linearizing \eqref{KPII_integrated} around $U=0$, we have
\begin{equation}
  \label{eq:LKP}
\pd_tU=\mL U,\quad 
\mL U=-\pd_x^3U+4\pd_xU-3\pd_x^{-1}\pd_y^2U-6\pd_x(\varphi U).
\end{equation}
Let
$\mL(\eta)u:=-\pd_x^3u+4\pd_xu+3\eta^2\pd_x^{-1}u-6\pd_x(\varphi u)$
be the operator on $L^2_a(\R)$ with its domain $D(\mL(\eta))=H^3_a(\R)$.
Since the potential of $\mL$ does not depend on $y$,
we have $\mL(u(x)e^{\pm i\eta y})=e^{\pm i\eta y} \mL(\eta)u(x)$.
We will look for resonant modes $\{g(x,\eta)e^{i\eta y}\}$
such that $g(\cdot,\eta)\in L^2_a(\R)$ is a solution of $\mL(\eta)u=\lambda u$.
\begin{lemma}
\label{lem:kpresonances}
Let $\eta\in\R\setminus\{0\}$, $\beta(\eta)=\sqrt{1+i\eta}$,
$\lambda(\eta)=4i\eta\beta(\eta)$  and
$$g(x,\eta)=\frac{-i}{2\eta\beta(\eta)}
\pd_x^2(e^{-\beta(\eta)x}\sech x),\quad
g^*(x,\eta)=\pd_x(e^{\beta(-\eta)x}\sech x).$$
Then
\begin{gather}
\label{eq:lem-kp1}
\mL(\eta)g(x,\pm\eta)=\lambda(\pm\eta)g(x,\pm\eta)\,,\\
\label{eq:lem-kp2}
\mL(\eta)^*g^*(x,\pm\eta)=\lambda(\mp\eta)g^*(x,\pm\eta)\,,\\
\label{eq:lem-kp3}
\int_\R g(x,\eta)\overline{g^*(x,\eta)}\,dx=1,\quad
\int_\R g(x,\eta)\overline{g^*(x,-\eta)}\,dx=0\,.  
\end{gather}
\end{lemma}
Lemma~\ref{lem:kpresonances} will be proved in Subsection~\ref{subsec:Miura}.
\par

To  resolve the singularity of $g(x,\eta)$ and the degeneracy of $g_*(x,\eta)$
at $\eta=0$, we decompose resonant modes and adjoint resonant modes into
their real parts and imaginary parts. Let
\begin{gather*}
g_1(x,\eta)=g(x,\eta)+g(x,-\eta)\,,\quad
g_2(x,\eta)=i\eta\{g(x,\eta)-g(x,-\eta)\}\,,\\
g_1^*(x,\eta)=\frac12\{g^*(x,\eta)+g^*(x,-\eta)\}\,,\quad
g_2^*(x,\eta)=\frac{i}{2\eta}\{g^*(x,\eta)-g^*(x,-\eta)\}\,.
\end{gather*}
Then we have the following.
\begin{lemma}
  \label{cl:L*g}
\begin{equation*}
  \int_\R g_j(x,\eta)\overline{g_k^*(x,\eta)}\,dx=\delta_{jk}
\quad\text{for $j$, $k=1$, $2$.}
\end{equation*}
\begin{gather*}
  \mL(\eta)g_1(x,\eta)
=\Re\lambda(\eta)g_1(x,\eta)+\frac{\Im\lambda(\eta)}{\eta}g_2(x,\eta)\,,\\
\mL(\eta)g_2(x,\eta)
=-\eta\Im\lambda(\eta)g_1(x,\eta)+\Re\lambda(\eta)g_2(x,\eta)\,,\\
\mL(\eta)^*g_1^*(x,\eta)
=\Re\lambda(\eta)g_1^*(x,\eta)-\eta\Im\lambda(\eta)g_2^*(x,\eta)\,,\\
\mL(\eta)^*g_2^*(x,\eta)
=\frac{\Im\lambda(\eta)}{\eta}g_1^*(x,\eta)+\Re\lambda(\eta)g_2^*(x,\eta)\,.
  \end{gather*}
\end{lemma}
\begin{proof}
Lemma~\ref{cl:L*g} follows immediately from Lemma~\ref{lem:kpresonances}
since $\overline{\lambda(\eta)}=\lambda(-\eta)$,\newline
$\overline{g(x,\eta)}=g(x,-\eta)$  and $\overline{g^*(x,\eta)}=g^*(x,-\eta)$.
\end{proof}
We remark that $\mL(0)$ coincides with the linearized operator
of the KdV equation around the $1$-soliton $\varphi(x-4t)$
and that $g_k(x,0)\in \ker_g(\mL(0))$, $g_k^*(x,0)\in \ker_g(\mL(0)^*)$ 
for $k=1$ and $2$.
\begin{claim}
  \label{cl:gk-approx}
Let $a\in(0,2)$ and $\nu(\eta)=\Re\beta(\eta)-1$.
Let $\eta_0$ be a positive number  such that $\nu_0:=\nu(\eta_0)<a$.
Then for $\eta\in[-\eta_0,\eta_0]$,
\begin{align*}
& g_1(x,\eta)=\frac14\varphi'+\frac{x}{4}\varphi'+\frac{1}{2}\varphi
+O(\eta^2)\,,\quad
g_2(x,\eta)=-\frac{1}{2}\varphi'+O(\eta^2) 
\quad\text{in $L^2_a(\R)$,}\\
& g_1^*(x,\eta)=\frac12\varphi+O(\eta^2)\,,\quad
g_2^*(x,\eta)=\int_{-\infty}^x\pd_c\varphi dx+O(\eta^2)
\quad\text{in $L^2_{-a}(\R)$,}
\end{align*}
where $\pd_c\varphi=\pd_c\varphi_c|_{c=2}$.
\end{claim}
\begin{proof}
Since $g_1(x,\eta)$ and $g_2(x,\eta)$ are even in $\eta$,
  \begin{align*}
g_1(x,\eta)=&\frac{1}{2i\eta}\pd_x^2\left\{
\left(\frac{e^{-\beta(\eta)x}}{\beta(\eta)}-\frac{e^{-\beta(-\eta)x}}{\beta(-\eta)}
\right)\sech x\right\}
\\=& \pd_s\left(\frac{e^{-\sqrt{s}x}}{\sqrt{s}}\sech x\right)_{xx}\bigg|_{s=1}
+O(\eta^2)
\\=& \frac{x+1}4\varphi'+\frac{1}{2}\varphi+O(\eta^2)\,,
  \end{align*}
and 
$$g_2(x,\eta)=(e^{-x}\sech x)_{xx}+O(\eta^2)=-\frac12\varphi'+O(\eta^2).$$
We can compute $g_1^*(x,\eta)$ and $g_2^*(x,\eta)$ in the same way.
\end{proof}

\subsection{Linearized Miura transformation}
\label{subsec:Miura}

Now we recall the Miura transformation of the KP-II equation. Let
$$M^c_{\pm}(v)=\pm \partial_x v+\partial_x^{-1}\partial_{y}v-v^2+\frac{c}{2}\,.$$
The transformations $M^c_{\pm}$ relate the KP-II equation to the mKP-II
equation (mKP-II) which reads
\begin{equation}\label{eq:MKP}
\pd_tv+\pd_x^3v+3\pd_x^{-1}\pd_y^2v-6v^2\pd_xv+6\pd_xv\pd_x^{-1}\pd_yv=0.
\end{equation}
Formally, if $v(t,x,y)$ is a solution of \eqref{eq:MKP} and $c>0$,
then $M^c_{\pm}(v)(t,x-3ct,y)$ 
are solutions of the KP-II equation \eqref{eq:KPII}.
A line soliton solution $\varphi_c(x-2ct)$ of the KP-II equation
is related to a kink solution $Q_c(x+ct)$ of \eqref{eq:MKP}, where
$Q_c(x)=\sqrt{\frac{c}{2}}\tanh\Big(\sqrt{\frac{c}{2}}\,x\Big)$.
Indeed, we have
\begin{equation}\label{algebra}
M_{+}^c(Q_c)=\varphi_c,\quad M_{-}^c(Q_c)=0\,.
\end{equation}
\par
From now on, let $c=2$, $Q=Q_2$ and $M_\pm=M_\pm^2$.
Let $v(t,x,y)=Q(x+2t)+V(t,x+2t,y)$ and linearize \eqref{eq:MKP} around
$V=0$. Then
\begin{gather}
  \label{eq:LMKP}
\pd_tV=\mL_MV,\\
\begin{split}
\mL_MV:=&-\pd_x^3V-2\pd_xV-3\pd_x^{-1}\pd_y^2V
+6\pd_x(Q^2 V)-6Q'\pd_x^{-1}\pd_yV 
\\=&-\pd_x^3V+4\pd_xV-3\pd_x^{-1}\pd_y^2V
-6\pd_x(Q'V)-6Q'\pd_x^{-1}\pd_yV\,. 
\end{split}
\notag
\end{gather}
In the last line, we use $Q'=1-Q^2$.
Let $X_M$ be the Banach space equipped with the norm
$\|v\|_{X_M}:=(\|v\|_X^2+\|\pd_xv\|_X^2+\|\pd_x^{-1}\pd_yv\|_X^2)^{1/2}$.
Thanks to the smoothing effect of $\mL_0$ in $X$
(see Lemma~\ref{lem:free-semigroup} in Section~\ref{sec:semigroup}),
the initial value problem 
$$\pd_tv=\mL_Mv\,,\quad v(0)=v_0$$
has a unique solution in the class $C([0,\infty);X_M)$.
\par

Solutions of \eqref{eq:LKP} are related to those of \eqref{eq:LMKP}
by the linearized Miura transformation
\begin{equation}
  \label{eq:LMiura+}
u=\nabla M_+(Q)v=\pd_xv+\pd_x^{-1}\pd_yv-2Qv.
\end{equation} 
Another linearized Miura transformation
\begin{equation}
  \label{eq:LMiura-}
u=\nabla M_-(Q)v=-\pd_xv+\pd_x^{-1}\pd_yv-2Qv
\end{equation}
maps solutions of \eqref{eq:LMKP} to those of the linearized KP-II
around $0$
\begin{equation}
  \label{eq:LKP0}
\pd_tu+\pd_x^3u-4\pd_xu+3\pd_x^{-1}\pd_y^2u=0\,.
\end{equation}

\begin{lemma}
\label{lem:relation}
  Suppose that $v$ is a solution to \eqref{eq:LMKP}.
Then $u=\nabla M_+(Q)v$ satisfies \eqref{eq:LKP}
and  $u=\nabla M_-(Q)v$ satisfies \eqref{eq:LKP0}.
\end{lemma}

\begin{proof}[Proof of Lemma~\ref{lem:relation}]
By a straightforward computation, we find that
\begin{equation}
  \label{eq:LML0}
\mL\nabla M_+(Q)=\nabla M_+(Q)\mL_M\,,\quad
\mL_0\nabla M_-(Q)=\nabla M_-(Q)\mL_M\,.
\end{equation}
Let $u_\pm=\nabla M_\pm(Q)v$. Then it follows from \eqref{eq:LML0} that
\begin{align*}
\pd_tu_+-\mL u_+=& \nabla M_+(Q)(\pd_tv-\mL_Mv)\,,\\
\pd_tu_--\mL_0u_-=& \nabla M_-(Q)(\pd_tv-\mL_Mv)\,.
\end{align*}
Therefore $u_+$ and $u_-$ are solutions of \eqref{eq:LKP}
and \eqref{eq:LKP0}, respectively, if $v$ is a solution to \eqref{eq:LMKP}.
Thus we complete the proof.
\end{proof}
\begin{lemma}
\label{lem:relation2}
Let $a>0$ and $v(t)\in C([0,\infty);X_M)$ be a solution to \eqref{eq:LMKP}.
\begin{enumerate}
\item
Suppose that $u(t)\in C([0,\infty);X)$ is a solution to
\eqref{eq:LKP} satisfying $u(0)=\allowbreak\nabla M_+(Q)v(0)$. Then
$u(t)=\nabla M_+(Q)v(t)$ holds for every $t\ge0$.
\item
Suppose that $u(t)\in C([0,\infty);X)$ is a solution to
\eqref{eq:LKP0} satisfying $u(0)=\allowbreak\nabla M_-(Q)v(0)$. Then
$u(t)=\nabla M_-(Q)v(t)$ holds for every $t\ge0$.
\end{enumerate}
\end{lemma}
\begin{proof}
Let $\tilde{u}(t)=\nabla M_+(Q)v(t)$.
Then $\tilde{u}(t)\in C([0,\infty);X)$.
Moreover, Lemma~\ref{lem:relation} implies that $\tilde{u}(t)$
is a solution of \eqref{eq:LKP}. Since $u(0)=\tilde{u}(0)$ and
both $u(t)$ and $\tilde{u}(t)$ are solutions of \eqref{eq:LKP}
in the class $ C([0,\infty);X)$, we have $u(t)=\tilde{u}(t)$.
Thus we prove (i).
We can prove (ii) in exactly the same way. This completes the proof of
Lemma~\ref{lem:relation2}.
\end{proof}

Let 
$$\mL_M(\eta)v:=-\pd_x^3v+4\pd_xv+3\eta^2\pd_x^{-1}v-6\pd_x(Q' v)
-6i\eta Q'\pd_x^{-1}v\,.$$
Then $\mL_M(v(x)e^{i\eta y})=e^{i\eta y} \mL_M(\eta)v(x)$ and $\mL_M(\eta)$
has the following resonant modes.
\begin{lemma}
\label{lem:mkpresonance}
Let $\eta\in\R$ and
$$
g_M(x,\eta)=\frac{-1}{2\beta(\eta)}\pd_x(e^{-\beta(\eta)x}\sech x)\,,\quad
g_M^*(x,\eta)=e^{\beta(-\eta)x}\sech x\,.$$
Then
\begin{gather}
\label{eq:lem-mkp1}
\mL_M(\eta)g_M(x,\eta)=\lambda(\eta)g_M(x,\eta)\,,\\
\label{eq:lem-mkp2}
\mL_M(\eta)^*g_M^*(x,\eta)=\lambda(-\eta)g_M^*(x,\eta)\,,\\
\label{eq:lem-mkp3}
\int_\R g_M(x,\eta)\overline{g_M^*(x,\eta)}\,dx=1\,.
\end{gather}
\end{lemma}
The eigenvalue problem $\mL u=\lambda u$ is related to the eigenvalue problem
$\mL_Mv=\lambda v$ via \eqref{eq:LML0}.
Before we prove Lemmas~\ref{lem:kpresonances} and \ref{lem:mkpresonance},
we will investigate the kernel and the cokernel of 
bounded operators $\mM_\pm(\eta):H^1_a(\R)\to L^2_a(\R)$ defined by
\begin{equation*}
\mM_\pm(\eta)g(x):=\pm g'(x)-i\eta\int_x^\infty g(t)\,dt-2Q(x)g(x) \,.
\end{equation*}
\begin{lemma}
\label{lem:RangeM-}
Let $a\in(0,2)$ and $\eta_0$ be a positive number satisfying $a>\nu_0$.
Then $\ker(\mM_-(\eta))=\spann\{g_M(\cdot,\eta)\}$ and
$\mathrm{Range}(\mM_-(\eta))=L^2_a(\R)$. Moreover, 
for any $\eta\in[-\eta_0,\eta_0]$
and $f\in L^2_a(\R)$, there exists a unique solution $v\in H^1_a(\R)$ of
\begin{equation}
\label{eq:m-}
\mM_-(\eta)v=f\,,
\end{equation}
that satisfies $\int_\R v(x)\overline{g_M^*(x,\eta)}\,dx=0$.
Moreover,
\begin{equation*}
\|v\|_{H^1_a(\R)}+|\eta|\|\pd_x^{-1}v\|_{L^2_a(\R)}
\le \frac{C}{a-\nu(\eta_0)}\|f\|_{L^2_a(\R)}\,,  
\end{equation*}
where $C$ is a constant depending only on $a$.
\end{lemma}

\begin{lemma}
\label{lem:RangeM+}
Let $a\in(0,2)$ and $\eta_0$ be a positive number satisfying $a>\nu_0$.
If $\eta\in [-\eta_0,\eta_0]$, then
$\ker(\mM_+(\eta)=\{0\}$ and $\mathrm{Range}(\mM_+(\eta))=
{}^\perp\spann\{g^*(x,-\eta)\}$. Moreover, for any $f\in L^2_a(\R)$ satisfying
$\int_\R f(x)\overline{g^*(x,-\eta)}\,dx=0$,
there exists a unique solution $v\in H^1_a(\R)$ of
\begin{equation}
  \label{eq:m+v=f}
\mM_+(\eta)v=f\,,
\end{equation}
satisfying
\begin{equation}
  \label{eq:est-m+v=f}
\|v\|_{H^1_a(\R)}+|\eta|\|\pd_x^{-1}v\|_{L^2_a}\le C\|f\|_{L^2_a}\,,  
\end{equation}
where $C$ is a constant depending only on $a$.
If $f$ satisfies $\int_\R f(x)\overline{g^*(x,\eta)}\,dx=0$
in addition, then $\int_\R v(x)\overline{g^*_M(x,\eta)}\,dx=0$.
\end{lemma}

\begin{proof}[Proof of Lemma~\ref{lem:RangeM-}]
Suppose $v\in \ker(\mM_-(\eta))$. Then $v\in H^1_a(\R)$ and
\begin{equation}
  \label{eq:kermM-}
  -v''+i\eta v -2(Qv)'=0\,.
\end{equation}
Eq.~\eqref{eq:kermM-} has a fundamental system
$\{\tg_1\,,\, \tg_2\}$, where
$$\tg_1(x)=\left(e^{-\beta(\eta)x}\sech x\right)_x\,,\quad 
\tg_2(x)=\left(e^{\beta(\eta)x}\sech x\right)_x\,.$$
Since $1\le\Re \beta(\nu)\le \Re \beta(\eta_0)$ and
\begin{equation}
  \label{eq:tgasymp}
\tg_1(x)\sim e^{-(\beta(\eta)\pm1)x}\quad\text{and}\quad
\tg_2(x)\sim e^{(\beta(\eta)\mp1)x}\quad\text{as $x\to\pm\infty$,}  
\end{equation}
it follows that
$\tg_1\in H^1_a(\R)$ and $\tg_2\not\in H^1_a(\R)$ and that
$v(x)=\alpha\tg_1(x)$ for an $\alpha$.
Thus we prove $\ker(\mM_-(\eta))=\spann\{g_M(\cdot,\eta)\}$.
\par

Suppose $v\in H^1_a(\R)$ is a solution of \eqref{eq:m-}.
Then  $v$ satisfies an ODE
\begin{equation}
  \label{eq:ode2}
-v''+i\eta v -2(Qv)'=f'\,.
\end{equation}
By the variation of the constants formula,
\begin{align*}
v(x)= & \tg_1(x)\int^x\frac{\tg_2(t)f'(t)}{W(t)}dt
-\tg_2(x)\int^x\frac{\tg_1(t)f'(t)}{W(t)}dt
\\=& \tg_1(x)\int^x k_1'(t)f(t)\,dt+\tg_2(x)\int^x k_2'(t)f(t)\,dt\,,  
\end{align*}
where $W(t)=\tg_1(t)\tg_2'(t)-\tg_1'(t)\tg_2(t)=-2i\eta\beta(\eta)\sech^2t$,
\begin{align*}
k_1(t)=& -\frac{\tg_2(t)}{W(t)}=
\frac{e^{\beta(\eta)t}(\beta(\eta)\cosh t-\sinh t)}
{2i\eta\beta(\eta)}\,,\\
k_2(t)=& \frac{\tg_1(t)}{W(t)}=
\frac{e^{-\beta(\eta)t}(\beta(\eta)\cosh t+\sinh t)}
{2i\eta\beta(\eta)}\,,
\end{align*}
$k_1'(t)=(2\beta(\eta))^{-1}e^{\beta(\eta)t}\cosh t$ and
$k_2'(t)=-(2\beta(\eta))^{-1}e^{-\beta(\eta)t}\cosh t$.
Now let
\begin{gather}
\label{eq:u1}
v(x)=\alpha \tg_1(x)+T_1(f)+T_2(f)\,,\\
T_1(f)=-\tg_1(x)\int_x^\infty k_1'(t)f(t)\,dt\,,\quad
T_2(f)= -\tg_2(x)\int_x^\infty k_2'(t)f(t)\,dt\,, \notag 
\end{gather}
where $\alpha$ is a constant to be chosen later.
Since $\sech x\cosh t \le e^{t-x}$ for $t\in [x,\infty)$ and
$\nu(\eta)\le \nu_0$ for $\eta\in[-\eta_0,\eta_0]$, 
$$|\tg_1(x) k_1'(t)|\lesssim e^{\nu_0(t-x)}\quad \text{if $t\ge x$.}$$
Using Young's inequality and the above, we have
\begin{align*}
\|T_1(f)\|_{L^2_a(\R)}\lesssim & 
\left\| \int^\infty_x e^{\nu(\eta)(t-x)}|f(t)|\,dt\right\|_{L^2_a(\R)}
\\ \lesssim &  \|e^{-(a-\nu_0)t}\|_{L^1(0,\infty)} \|f\|_{L^2_a(\R)}
\le \frac{C_0}{a-\nu_0}\|f\|_{L^2_a(\R)},
\end{align*}
where $C_0$ is a constant independent of $\eta_0$ and $f\in L^2_a(\R)$.
Using the fact that $0\le \cosh t\sech x\le e^{t-x}$ if $x \le t$
and that $\nu(\eta)\ge0$, we have
\begin{align*}
\|T_2(f)\|_{L^2_a(\R)}\lesssim &
\left\| \int_x^\infty e^{\nu(\eta)(x-t)}|f(t)|\,dt \right\|_{L^2_a(\R)}
\\ \lesssim & \|e^{-(a+\nu(\eta))t}\|_{L^1(0,\infty)} \|f\|_{L^2_a(\R)}
\le C_1\|f\|_{L^2_a(\R)},
\end{align*}
where $C_1$ is a constant independent of $\eta_0$ and $f\in L^2_a(\R)$.
Since
\begin{align*}
\int_\R \tg_1(x)\overline{g_M^*(x,\eta)}\,dx=&
-\int_\R \sech^2x(\beta(\eta)-\tanh x)\,dx= -2\beta(\eta)\ne0\,,
\end{align*}  
there exists a unique $\alpha$ such that $\int v(x)\overline{g^*(x,\eta)}\,dx=0$.
Since $L^2_a(\R)\ni f\mapsto T_1(f)$, $T_2(f)\in L^2_a(\R)$ are continuous,
$\alpha=\alpha(f)$ is also continuous in $f$.
Thus we prove that there exists a constant $C_2$  such that 
\begin{equation}
  \label{eq:solm-uf}
\|v\|_{L^2_a(\R)}\le C_2\|f\|_{L^2_a(\R)}  
\end{equation}
for every  $\eta\in[-\eta_0,\eta_0]\setminus\{0\}$ and $f\in L^2_a(\R)$.
\par
Differentiating \eqref{eq:u1} with respect to $x$,
we have
$$v'(x)=\alpha \tg_1'(x) -f(x)
-\tg_1'(x)\int_x^\infty k_1'(t)f(t)\,dt-\tg_2'(x)\int_x^\infty k_2'(t)f(t)\,dt\,.$$
We can prove 
\begin{equation}
  \label{eq:solm-uf'}
\|v'(x)\|_{L^2_a(\R)} \le \frac{C_3}{a-\nu_0} \|f\|_{L^2_a(\R)}\,,
\end{equation}
in the same way as \eqref{eq:solm-uf},
where  $C_3$ is a positive constant independent of $\eta_0$ and 
$f\in L^2_a(\R)$.
Combining \eqref{eq:solm-uf} and \eqref{eq:solm-uf'} with \eqref{eq:m-},
we have
\begin{align*}
 |\eta| \|\pd_x^{-1}v\|_{L^2_a(\R)}\le & \|v'\|_{L^2_a(\R)}+2\|Qv\|_{L^2_a(\R)}+\|f\|_{L^2_a(\R)}
\le \frac{C_4}{a-\nu_0} \|f\|_{L^2_a(\R)}\,,
\end{align*}
where $C_4$ is a positive constant independent of
$\eta_0$ and  $f\in L^2_a(\R)$.
Thus we complete the proof.
\end{proof}

\begin{proof}[Proof of Lemma~\ref{lem:RangeM+}]
First, we will show that $\ker\left(\mM_+(\eta)^*\right)
=\spann\{\tg_2(x)\}$.
Since $\mM_-(\eta)$ is formally an adjoint of $\mM_+(\eta)$,
we easily see that $h\in \ker(\mM_+^*(\eta))\subset  L^2_{-a}(\R)$
is a solution of \eqref{eq:ode2}
and that $h(x)=\alpha\tg_2(x)=\alpha g^*(x,-\eta)$
for an $\alpha\in \C$.
Since $\ker\left(\mM_+(\eta)^*\right)=\spann\{\tg_2(x)\}$,
we have $\mathrm{Range}(\mM_+(\eta))\subset {}^\perp\spann\{g^*(x,-\eta)\}$.
\par

Next we will show that $\ker(\left(\mM_+(\eta)\right)=\{0\}$.
Suppose $\mM_+(\eta)h=0$. Then
\begin{equation}
\label{eq:kerM+}
 h''-2(Qh)'+i\eta h=0\,.
\end{equation}
Eq.~\eqref{eq:kerM+} has a fundamental system
$\{h_1(x),h_2(x)\}$, where
$$h_1(x)=e^{\beta(-\eta)x}\cosh x\,,\quad h_2(x)=e^{-\beta(-\eta)x}\cosh x\,.$$
Since
\begin{equation}
  \label{eq:thasymp}
h_1(x)\sim e^{(\beta(-\eta)\pm1)x}\,,\quad h_2(x)\sim e^{(-\beta(-\eta)\pm1)x}
\quad\text{as $x\to\pm\infty$,}  
\end{equation}
it is clear that $h\in H^1_a(\R)$ if and only if $h=0$.
Thus we prove $\ker\left(\mM_+(\eta)\right)=\{0\}$.
\par

Secondly, we will show that
$\mathrm{Range}\left(\mM_+(\eta)\right)={}^\perp\spann\{g^*(x,-\eta)\}$.
Suppose that $v\in H^1_a(\R)$ is a solution of \eqref{eq:m+v=f}.
Then
\begin{equation}
  \label{eq:m+v=f2}
v''-2(Qvg)'+i\eta v=f'\,.  
\end{equation}
By the variation of constants formula,
we can find the following solution of \eqref{eq:m+v=f2}.
\begin{gather}
  \label{eq:solm+v}
v(x)=T_3(f)+T_4(f),\\
\notag
T_3(f):=\frac{e^{\beta(-\eta) x}\cosh x}{2\beta(-\eta)}\int_x^\infty
\left(e^{-\beta(-\eta) t}\sech t\right)_t f(t)\,dt,\\
 \label{eq:T41}
T_4(f):=\frac{e^{-\beta(-\eta) x}\cosh x}{2\beta(-\eta)}
\int^x_{-\infty}\left(e^{\beta(-\eta)t}\sech t\right)_tf(t)\,dt\,.
\end{gather}
Since $\cosh x\sech t \le e^{|x-t|}$, we have
\begin{align*}
  \|T_3(f)\|_{L^2_a(\R)}\lesssim &
\left\|\int_x^\infty e^{\nu(-\eta)(x-t)}|f(t)|dt\right\|_{L^2_a(\R)}
\\ \lesssim & \|e^{-(a+\nu(-\eta))t}\|_{L^1(0,\infty)}\|f\|_{L^2_a(\R)}
\le C_1\|f\|_{L^2_a(\R)}\,,
\end{align*}
where $C_1$ is a constant depending only on $a$.
If $\int_\R f(x)\overline{g^*(x,-\eta)}\,dx=0$,
then $T_4(f)$ can be rewritten as
\begin{equation}
  \label{eq:T42}
T_4(f)=-\frac{e^{-\beta(-\eta) x}\cosh x}{2\beta(-\eta)}
\int_x^\infty \left(e^{\beta(-\eta)t}\sech t\right)_tf(t)\,dt\,.
\end{equation}
Using \eqref{eq:T41} for $x\ge0$ and \eqref{eq:T42} for $x\le 0$
and the fact that $\cosh x\sech t\le 2e^{-|x-t|}$ for $t$ satisfying
$|t|\ge |x|$, we have
\begin{align*}
\|T_4(f)\|_{L^2_a(\R)} \lesssim &
\left\| \int_x^\infty e^{(a-\nu(-\eta))(x-t)}e^{at}|f(t)|\,dt\right\|_{L^1(0,\infty)}
\\ &  +\left\|
\int^x_{-\infty} e^{(a-\nu(-\eta)-2)(x-t)}e^{at}|f(t)|\,dt\right\|_{L^1(-\infty,0)}
\\ \le & C_2\|f\|_{L^2_a(\R)}\,,
\end{align*}
where $C_2$ is a constant depending only on $a$.
Thus we prove that \eqref{eq:m+v=f} has a unique solution $v\in L^2_a(\R)$.
We can prove \eqref{eq:est-m+v=f} in the same way as Lemma~\ref{lem:RangeM-}.
\par
Suppose $f$ satisfies $\int_\R f(x)\overline{g^*(x,\pm\eta)}\,dx=0$,
then it follows from  \eqref{eq:M*gM*} and \eqref{eq:m+v=f} that
\begin{align*}
  2i\eta\int_\R v(x)\overline{g_M^*(x,\eta)}\,dx=& 
-\int_\R \mM_+(\eta)v(x)\overline{g^*(x,\eta)}\,dx
\\= -\int_\R f(x)\overline{g^*(x,\eta)}=0\,.
\end{align*}
Thus we have $\int_\R v(x)\overline{g^*_M(x,\eta)}\,dx=0$
for $\eta\in[-\eta_0,\eta_0]\setminus\{0\}$.
This completes the proof of Lemma~\ref{lem:RangeM+}.
\end{proof}

Now we are in position to prove Lemmas~\ref{lem:kpresonances} and \ref{lem:mkpresonance}.
\begin{proof}[Proof of Lemmas~\ref{lem:kpresonances} and \ref{lem:mkpresonance}]
First, we will show that $\nabla M_+(Q)g_M(x,\eta)e^{i\eta y}$
are the resonant eigenmodes of $\mL$ and that
$\nabla M_-(Q)(g_M(x,\eta)e^{i\eta y})=0$ by using \eqref{eq:LML0}, 
\par
Let $\mL_0(\eta)u:=-\pd_x^3u+4\pd_xu+3\eta^2\pd_x^{-1}u$ be the operator
on $L^2_a(\R)$ with its domain $D(\mL_0(\eta))=H^3_a(\R)$.
By the definition of $\mL(\eta)$ and $\mM_\pm(\eta)$, we have
$\mL_0(u(x)e^{\pm i\eta y})=e^{\pm i\eta y} \mL_0(\eta)u(x)$ and
$$\nabla M_\pm(Q)(g(x)e^{i\eta y})=(\mM_\pm(\eta)g)(x)e^{i\eta y}\,.$$
In view of \eqref{eq:LML0},
\begin{gather}
  \label{eq:LM+}
\mL(\eta)\mM_+(\eta)=\mM_+(\eta)\mL_M(\eta)\,,
\\ \label{eq:LM-}
\mL_0(\eta)\mM_-(\eta)=\mM_-(\eta)\mL_M(\eta)\,.
\end{gather}
By a simple computation, we find
\begin{equation}
  \label{eq:LMgM}
\mM_+(\eta)g_M(x,\eta)=-2i\eta g(x,\eta)\,,\quad
\mM_-(\eta)g_M(x,\eta)=0\,.  
\end{equation}
Combining \eqref{eq:LM-} and \eqref{eq:LMgM}, we have
$$\mM_-(\eta)\mL_M(\eta)g_M(x,\eta)=\mL_0(\eta)\mM_-(\eta)g_M(x,\eta)=0\,.$$
Since $g_M(x,\eta)\in H^4_a(\R)$ for an $a\in (\nu(\eta),2)$,
we have $\mL_M(\eta)g_M(x,\eta)\in H^1_a(\R)$ and
$\mL_M(\eta)g_M(x,\eta)\in \ker\mM_-(\eta)$.
Lemma~\ref{lem:RangeM-} implies that there exists
a $\lambda(\eta)\in \C$ such that
$\mL_M(\eta)g_M(x,\eta)=\lambda(\eta) g_M(x,\eta)$.
Since $g_M(x,\eta)\sim e^{-(1+\beta(\eta))x}$ as $x\to\infty$,
we see that
\begin{align*}
  \lambda(\eta)=&(1+\beta(\eta))^3-4(1+\beta(\eta))
-\frac{3\eta^2}{1+\beta(\eta)}=4i\eta\beta(\eta)\,.
\end{align*}
Thus we prove \eqref{eq:lem-mkp1}.
It follows from \eqref{eq:lem-mkp1}, \eqref{eq:LM+} and \eqref{eq:LMgM}
that
\begin{align*}
\mL(\eta)g(x,\eta)=
& \frac{i}{2\eta}\mM_+(\eta)\mL_M(\eta)g_M(x,\eta)
\\=& \frac{i\lambda(\eta)}{2\eta}\mM_+(\eta)g_M(x,\eta)=
\lambda(\eta)g(x,\eta)\,,
\end{align*}
and $\mL(\eta)g(x,-\eta)=\overline{\mL(\eta)g(x,\eta)}=
\lambda(-\eta)g(x,-\eta)$.
\par
Using the fact that $\pd_x\mL(\eta)^*=-\mL(-\eta)\pd_x$ (formally) and
$\varphi$ is even, we can easily confirm \eqref{eq:lem-kp2}.
Since $g_M(x,\eta)$ is a solution of \eqref{eq:kermM-} and
$$g^*(x,\eta)=-2\beta(-\eta)g_M(-x,-\eta)\,,\quad Q(-x)=-Q(x)\,,$$
we have $\pd_x^2g^*(x,\eta)+2\pd_x(Q(x)g^*(x,\eta))+i\eta g^*(x,\eta)=0$.
Combining the above with $g^*(x,\eta)=\pd_xg_M^*(x,\eta)$, we have
\begin{equation}
  \label{eq:M*gM*}
  \begin{split}
\mM_+(\eta)^*g^*(x,\eta)=&
-\pd_xg^*(x,\eta)+i\eta\int^x_{-\infty} g^*(t,\eta)\,dt-2Q(x)g^*(x,\eta)
\\=&2i\eta g_M^*(x,\eta)\,.    
  \end{split}
\end{equation}
\begin{align*}
\end{align*}
By \eqref{eq:LM+},
\begin{equation}
  \label{eq:LM+*}
M_+(\eta)^*\mL(\eta)^*=\mL_M(\eta)^*\mM_+(\eta)^*\,.  
\end{equation}
Eq.~\eqref{eq:lem-mkp2} follows immediately from \eqref{eq:lem-kp2},
\eqref{eq:M*gM*} and \eqref{eq:LM+*}.
\par
Next, we will prove \eqref{eq:lem-mkp3}.
By integration by parts, 
 \begin{align*}
\int g_M(x,\eta)\overline{g_M^*(x,\eta)}\,dx=&
\frac{1}{2\beta(\eta)}\int (\beta(\eta)+\tanh x)\sech^2xdx=1\,.
  \end{align*}

Finally, we will prove \eqref{eq:lem-kp3}.
By \eqref{eq:LMgM} and \eqref{eq:M*gM*},
  \begin{align*}
2i\eta \int_\R g(x,\eta)\overline{g^*(x,\eta)}\,dx=&
-\int_\R \mM_+(\eta)g_M(x,\eta)\overline{g^*(x,\eta)}\,dx
\\=& -\int_\R g_M(x,\eta)\overline{\mM_+^*g^*(x,\eta)}\,dx
\\=& 2i\eta\int_\R g_M(x,\eta)\overline{g_M^*(x,\eta)}\,dx=2i\eta\,.
  \end{align*}
Thus we prove \eqref{eq:lem-kp3} for $\eta\ne0$.
\end{proof}

If $\eta$ is large, the operators $\mM_\pm(\eta): H^1_a(\R)\to L^2_a(\R)$
have bounded inverse.
\begin{lemma}
\label{lem:Mpmlarge}
Suppose $a\in(0,2)$, $\eta>0$ and $\nu(\eta)>a$.
\begin{enumerate}
\item \label{Mpml1}
For every $f\in L^2_a(\R)$, there exists a unique solution $v_+$ of
\eqref{eq:m-} satisfying
\begin{equation}
  \label{eq:Mpml1}
\|v_+\|_{H^1_a(\R)}+|\eta|\|\pd_x^{-1}v_+\|_{L^2_a(\R)} \le 
\frac{C}{\nu(\eta)-a}\|f\|_{L^2_a(\R)}\,,  
\end{equation}
where $C$ is a constant depending only on $a$.
\item \label{Mpml2}
For every $f\in L^2_a(\R)$, there exists a unique solution $v_-$ of
\eqref{eq:m+v=f} satisfying
\begin{equation}
  \label{eq:Mpml2}
\|v_-\|_{H^1_a(\R)}+|\eta|\|\pd_x^{-1}v_-\|_{L^2_a(\R)} \le 
\frac{C}{\nu(\eta)-a}\|f\|_{L^2_a(\R)}\,,  
\end{equation}
where $C$ is a constant depending only on $a$.
\end{enumerate}
\end{lemma}

\begin{proof}[Proof of Lemma~\ref{lem:Mpmlarge}]
If $\nu(\eta)>a>0$, then \eqref{eq:tgasymp} and \eqref{eq:thasymp} imply
$\ker(\mM_\pm(\eta))=\{0\}$ and that \eqref{eq:m-} and \eqref{eq:m+v=f} have at most one solution.
\par
First we prove \eqref{Mpml1}. Let 
$$v_+(x)=\tg_1(x)\int_{-\infty}^x k_1'(t)f(t)\,dt
+\tg_2(x)\int_x^\infty k_2'(t)\,dt\,.$$ Then $v(x)$ is a solution of
\eqref{eq:m-}.  Since $|\tg_1(x)k_1'(t)|+|\tg_1'(x)k_1'(t)|\lesssim
e^{-\nu(\eta)(x-t)}$ if $x>t$ and
$|\tg_2(x)k_2'(t)|+|\tg_2'(x)k_2'(t)|\lesssim e^{\nu(\eta)(x-t)}$ if
$x<t$, we have
$$|v_+(x)|+|\pd_xv_+(x)|\lesssim \int_\R e^{-\nu(\eta)|x-t|}|f(t)|\,dt\,.$$
Using Young's inequality, we have
\begin{align*}
\|v_+\|_{L^2_a(\R)}+\|\pd_xv\|_{L^2_a(\R)}
\lesssim \|e^{-(\nu(\eta)-a)|x|}\|_{L^1(\R)}\|f\|_{L^2_a(\R)}
\lesssim (\nu(\eta)-a)^{-1}\|f\|_{L^2_a(\R)}\,.
\end{align*}
Thus we can prove \eqref{eq:Mpml1}
in the same way as the proof of Lemma~\ref{lem:RangeM-}.
\par
Now we prove \eqref{Mpml2}.
Let $v_-=T_3(f)+T_4(f)$.  Obviously, $v_-$ is a solution of \eqref{eq:m+v=f}
satisfying 
$$|v_-(x)|+|\pd_xv_-(x)|\lesssim \int_\R e^{-\nu(\eta)|x-t|}|f(t)|\,dt\,.$$
Thus we can prove \eqref{eq:Mpml2} in the same way as \eqref{eq:Mpml1}.
This completes the proof of Lemma~\ref{lem:Mpmlarge}.
\end{proof}

Using Lemmas~\ref{lem:RangeM-}, \ref{lem:RangeM+} and \ref{lem:Mpmlarge},
we will investigate the spectrum $\sigma(\mL(\eta))$ of $\mL(\eta)$.
\begin{lemma}
\label{lem:Llarge-eta}
Let $a\in(0,2)$ and  $\eta_*$ be a positive number satisfying $\nu(\eta_*)=a$.
\begin{enumerate}
\item If $\eta\in (-\eta_*,\eta_*)$, then $\mL(\eta)$
has no eigenvalue other than $\lambda(\pm\eta)$ and
$$\sigma(\mL(\eta))=\{\lambda(\pm\eta_*)\}
\cup\{ip(\xi+ia,\eta)\mid \xi\in\R\}\,.$$
\item If $\eta\in \R\setminus[-\eta_*,\eta_*]$, then
 $\sigma(\mL(\eta))=\{ip(\xi+ia,\eta)\mid \xi\in\R\}$.
\end{enumerate}
\end{lemma}
\begin{proof}[Proof of Lemma~\ref{lem:Llarge-eta}]
The equation $\nu(\eta)=a$ has a unique positive root $\eta_*$
because $\nu(\eta)$ is monotone increasing for $\eta\ge0$,
$\nu(0)=0$ and $\nu(\infty)=\infty$.
\par

Since $\lambda-\mL(\eta)$ and $\lambda-\mL_0(\eta)$ are
invertible for large $\lambda>0$ and 
$(\lambda-\mL(\eta))^{-1}-(\lambda-\mL_0(\eta))^{-1}$
is compact, it follows from the Weyl essential
spectrum theorem that $\sigma(\mL(\eta))\setminus \sigma_p(\mL(\eta))
=\{ip(\xi+ia)\mid \xi\in\R\}$.
\par
Suppose that $\eta\in(-\eta_*,\eta_*)$ and  that $\mathcal{L}(\eta)u=\lambda u$
for some $u\in H^3_a(\R)$ and $\lambda\in \C\setminus\{\lambda(\pm\eta)\}$.
Then
\begin{equation}
  \label{eq:uperpg}
\int_\R u(x)\overline{g^*(x,\pm\eta)}\,dx=0\,.
\end{equation}
Indeed, it follows from Lemma~\ref{lem:kpresonances} that
\begin{align*}
(\lambda-\lambda(\pm\eta))\int_\R u(x)\overline{g^*(x,\pm\eta)}\,dx
=& \int_\R \{\lambda u(x) -(\mL(\eta)u)(x)\}\overline{g^*(x,\pm\eta)}\,dx
=0\,.
\end{align*}
Lemma~\ref{lem:RangeM+} implies that there exists a solution $v\in H^4_a(\R)$ of
$u=\mM_+(\eta)v$ satisfying $\int_\R v(x)\overline{g_M^*(x,\eta)}\,dx=0$.
By \eqref{eq:LM+},
\begin{align*}
\mM_+(\eta)(\mL_M(\eta)v-\lambda v)=& (\mL(\eta)-\lambda)\mM_+(\eta)v
\\ =& \mL(\eta)u-\lambda u=0\,.
\end{align*}
Since $\ker(\mM_+(\eta))=\{0\}$, it follows that $\mL_M(\eta)v=\lambda v$.
Using \eqref{eq:LM-}, we have
\begin{equation}
\label{eq:supp}
(\mL_0(\eta)-\lambda)\mM_-(\eta)v
= \mM_-(\eta)(\mL_M(\eta)v-\lambda v)=0\,,
\end{equation}
whence $\mM_-(\eta)v=0$ because \eqref{eq:supp} implies that
the support of  $\mF_x(\mM_-(\eta)v)(\xi)$ 
is contained in $\{\xi\in \R \mid \xi^4+4\xi^2+i\lambda \xi-\eta^2=0\}$.
Lemma~\ref{lem:RangeM-} implies there exists an $\alpha \in \C$ such that
$v(x)=\alpha g_M(x,\eta)$ and hence it follows from \eqref{eq:LMgM} that
$$u(x)=\mM_+(\eta)v=-2i\alpha\eta g(x,\eta)\,.$$
By Lemma~\ref{lem:kpresonances} and \eqref{eq:uperpg},
$$\int_\R u(x)\overline{g^*(x,\eta)}\,dx=-2i\alpha\eta=0\,,$$
whence $u=0$. Thus we prove (1).
\par
Suppose $\eta\in\R\setminus[-eta_*,\eta_*]$ and that
$\mathcal{L}u=\lambda u$ for some $u\in H^3_a(\R)$ and $\lambda\in \C$.
Lemma~\ref{lem:Mpmlarge} implies that there exists $v\in H^4_a(\R)$
satisfying $u=\mM_+(\eta)v$ and we can prove that $\mM_-(\eta)v=0$
in the same way as the proof of (1). Since $\mM_-(\eta)$ has
the bounded inverse, it follows that $v=0$ and $u=\mM_+(\eta)v=0$.
Thus we complete the proof.
\end{proof}
\bigskip

\section{Semigroup estimates for the linearized KP-II equation}
\label{sec:semigroup}
In this section, we will prove exponential decay estimates
of solutions to \eqref{eq:LKP}.
To begin with, we define a spectral projection to low frequency resonant
modes. Let $P_0(\eta_0)$ be an operator defined by
\begin{equation*}
P_0(\eta_0)f(x,y)=\frac{1}{2\pi}\sum_{k=1,\,2}
\int_{-\eta_0}^{\eta_0}a_k(\eta)g_k(x,\eta)e^{iy\eta}\,d\eta\,,  
\end{equation*}
\begin{align*}
 a_k(\eta)=&\int_\R \lim_{M\to\infty}\left(\int_{-M}^M f(x_1,y_1)e^{-iy_1\eta}
\,dy_1\right)\overline{g_k^*(x_1,\eta)}\,dx_1
\\=& \sqrt{2\pi}\int_\R (\mF_yf)(x,\eta)\overline{g_k^*(x,\eta)}\,dx\,.
  \end{align*}
We will show that $P_0(\eta_0)$ is a spectral projection on
$X=L^2(\R^2;e^{2ax}dxdy)$.
\begin{lemma}
  \label{lem:p0}
Let $a\in(0,2)$ and $\eta_1$ be a positive constant satisfying
$\nu(\eta_1)<a$. If $\eta_0\in[-\eta_1,\eta_1]$, then
\begin{enumerate}
\item \label{it:p01}
$\|P_0(\eta_0)f\|_X+\|P_0(\eta_0)\pd_xf\|_X
\le C\|f\|_X$ for any $f\in X$,
where $C$ is a positive constant depending only on $a$ and $\eta_1$,
\item \label{it:p02}
$\|P_0(\eta_0)f\|_X+\|P_0(\eta_0)\pd_xf\|_X
\le C\|e^{ax}f\|_{L^1_xL^2_y}$ for any $e^{ax}f\in L^1_xL^2_y$,
where $C$ is a positive constant depending only on $a$ and $\eta_1$,
\item \label{it:p03}
$\mL P_0(\eta_0)f=P(\eta_0)\mL f$ for any $f\in D(\mL)
=\{u\mid u\,,\; \pd_x^3u\,,\;\pd_x^{-1}\pd_y^2u\in X\}$,
\item \label{it:p04}
$P_0(\eta_0)^2=P_0(\eta_0)$ on $X$,
\item \label{it:p05}
$e^{t\mL}P_0(\eta_0)=P_0(\eta_0)e^{t\mL}$ on $X$.
\end{enumerate}
\end{lemma}
\begin{proof}
First, we will show \eqref{it:p01}.
Since $C_0^\infty(\R)\otimes C_0^\infty(\R)$ is dense in $X$,
we may assume $f\in C_0^\infty(\R)\otimes C_0^\infty(\R)$. Let
$$ f_k(x,y)=\frac{1}{2\pi}
\int_{-\eta_0}^{\eta_0}a_k(\eta)g_k(x,\eta)e^{iy\eta}d\eta
\quad\text{for $k=1$, $2$.}$$
By Plancherel's theorem,
\begin{equation}
  \label{eq:fk-bound}
\begin{split}
\|f_k(x,y)\|_{L^2_y}=& \frac{1}{\sqrt{2\pi}}
\left(\int_{-\eta_0}^{\eta_0}\left|a_k(\eta)g_k(x,\eta)\right|^2d\eta\right)^{1/2}
\\ \le & \frac{1}{\sqrt{2\pi}}\sup_{\eta\in[-\eta_0,\eta_0]}|g_k(x,\eta)|
\left(\int_{-\eta_0}^{\eta_0}\left|a_k(\eta)\right|^2d\eta\right)^{1/2}\,.
\end{split}  
\end{equation}
If $\nu(\eta_1)<a$, then it follows from the definition of $g_k$ and $g_k^*$
that there exists a positive constant $C'$ such that
for $\eta\in[-\eta_1,\eta_1]$ and $x\in\R$,
\begin{equation}
  \label{eq:gk-bound2}
\begin{split}
& |g_1(x,\eta)|\le C'\la x\ra e^{-2x_+}e^{\nu(\eta_1)x_-}\,,\quad
|g_2(x,\eta)|\le C' e^{-2x_+}e^{\nu(\eta_1)x_-}\,,\\
& |g_1^*(x,\eta)|\le C' e^{\nu(\eta_1)x_+}e^{-2x_-}\,,\quad
|g_2^*(x,\eta)|\le C' \la x\ra e^{\nu(\eta_1)x_+}e^{-2x_-}\,,
\end{split}  
\end{equation}
where $x_\pm=\max(\pm x\,,\,0)$ and $C'$ is a constant depending only $\eta_1$.
Hence it follows from \eqref{eq:fk-bound} and \eqref{eq:gk-bound} that
\begin{equation}
  \label{eq:f-bound2}
\begin{split}
\|P_0(\eta_0)f\|_X \le
& \sum_{k=1,2} \left\|\|f_k\|_{L^2(\R_y)}\right\|_{L^2_a(\R_x)}
\\ \le & C_1\left(\int_{-\eta_0}^{\eta_0}(|a_1(\eta)|^2+|a_2(\eta)|^2)
\,d\eta\right)^{1/2}\,,
\end{split}  
\end{equation}
where $C_1$ is a constant depending only on $a$ and $\eta_1$.
Using the Schwarz inequality and \eqref{eq:gk-bound2}, we have
for $\eta\in[-\eta_0,\eta_0]$,
\begin{align*}
|a_k(\eta)|\le & \sqrt{2\pi}\|(\mF_yf)(x,\eta)\|_{L^2_a(\R_x)}
\|g_k^*(x,\eta)\|_{L^2_{-a}(\R_x)}
\le  C_2\|(\mF_yf)(x,\eta)\|_{L^2_a(\R_x)}\,,
\end{align*}
where $C_2$ is a constant depending only on $a$ and $\eta_1$.
Hence it follows that for any $f\in C_0^\infty(\R)\otimes C_0^\infty(\R)$, 
\begin{align*}
\|P_0(\eta)f\|_X
\le & C_1C_2 \left(
\int_{-\eta_0}^{\eta_0}\|\mF_yf(x,\eta)\|_{L^2_a(\R_x)}^2\,d\eta\right)^{1/2}
\\ =& C_1C_2\|f\|_X\,.
\end{align*}
We can prove $\|P_0(\eta)\pd_x f\|_X \lesssim\|f\|_X$ in exactly the same
way.

Next we will prove \eqref{it:p02}.
Using Minkowski's inequality and applying 
\eqref{eq:gk-bound2} and Plancherel's theorem to the resulting equation,
we have
\begin{align*}
\|a_k\|_{L^2(-\eta_0,\eta_0)}\le &\sqrt{2\pi}\int_\R
 \|(\mF f)(x,\cdot)\overline{g_k^*(x,\cdot)}\|_{L^2(-\eta_0,\eta_0)}\, dx
\\ \le & \sqrt{2\pi}
\sup_{x\in\R\,, \eta\in[-\eta_0,\eta_0]}|e^{-ax}g_k^*(x,\eta)|
\int_\R e^{ax}\|(\mF f)(x,\cdot)\|_{L^2(-\eta_0,\eta_0)}\,dx
\\ \lesssim & \|e^{ax}f\|_{L^1_xL^2_y}\,.
\end{align*}
Substituting the above into \eqref{eq:f-bound2}, we have
$\|P_0(\eta)f\|_X \lesssim \|e^{ax}f\|_{L^1_xL^2_y}$.
We can prove $\|P_0(\eta)\pd_xf\|_X \lesssim \|e^{ax}f\|_{L^1_xL^2_y}$
in exactly the same way.
\par
Since the potential of $\mL$ is independent of $y$,
it suffices to show \eqref{it:p03} for $f\in D(\mL)\cap\widetilde{X}$,
where $\widetilde{X}=\{f\in X\mid (\mF_yf)(\cdot,\eta)=0\;
a.e.\;\eta\not\in[-\eta_0,\eta_0]\}$.
Since $\lambda(\pm\eta)$ are isolated eigenvalue of $\mL(\eta)$
by Lemma~\ref{lem:Llarge-eta}, it follows from Lemmas~\ref{lem:kpresonances}
and \ref{cl:L*g} that
\begin{align*}
P_0(\eta_0)f= & \frac{1}{(2\pi)^{3/2}i}
\int_{-\eta_0}^{\eta_0}\int_\Gamma
(\lambda-\mL(\eta))^{-1}(\mF_yf)(\cdot,\eta)e^{iy\eta}\,d\lambda\, d\eta
\\ =& \frac{1}{2\pi i}\int_\Gamma (\lambda-\mL)^{-1}f\,d\lambda\,,
\end{align*}
where $\Gamma$ is the boundary of a domain
$D\supset \{\lambda(\pm\eta)\mid \eta\in[-\eta_0,\eta_0]\}$
satisfying $D \cap \{p(\eta+ia)\mid \eta\in\R\} =\emptyset$.
Thus $P_0(\eta_0)$ equals to a spectral projection of $\mL|_{\widetilde{X}}$
defined by the Dunford integral and 
\eqref{it:p03}--\eqref{it:p05} can be obtained by a standard argument.
We remark that $e^{t\mL}$ is a $C^0$-semigroup on $X$ 
because $\mL_0:=-\pd_x^3+4\pd_x-3\pd_x^{-1}\pd_y^2$ is $m$-dissipative on $X$ and
$\mL-\mL_0$ is infinitesimally small with respect to $\mL_0$.
Thus we complete the proof of Lemma~\ref{lem:p0}.
\end{proof}

Let $0<\eta_1\le \eta_2\le \infty$ and
$P_1(\eta_1,\eta_2)$ and $P_2(\eta_1,\eta_2)$ be projections defined by
\begin{gather*}
P_1(\eta_1, \eta_2)u(x,y)=\frac{1}{2\pi}\int_{\eta_1\le |\eta|\le \eta_2}
\int_\R  u(x,y_1)e^{i\eta(y-y_1)}dy_1d\eta\,,
\\ P_2(\eta_1,\eta_2)= P_1(0,\eta_2)-P_0(\eta_1)\,.
\end{gather*}
We remark that $P_2(\eta_1,\eta_2)$ is a projection onto 
non-resonant low frequency modes and
that $\|P_2(\eta_1,\eta_2)e^{t\mL}\|_{B(X)}$ decays exponentially as $t\to\infty$.
\begin{proposition}
  \label{prop:semigroup-p2}
Let $a\in (0,2)$ and $\eta_1$ be a positive number satisfying
$\nu(\eta_1)<a$.  Then there exist positive constants
$K$ and  $b$ such that for any $\eta_0\in(0,\eta_1]$, $f\in X$ and $t\ge 0$,
$$ \|e^{t\mL}P_2(\eta_0,\infty)f\|_X\le K
(\eta_0^{-1}e^{\Re\lambda(\eta_0)t}+e^{-bt})\|f\|_X\,.$$
\end{proposition}
\begin{corollary}
  \label{cor:semigroup-p2}
Let $a$ and $\eta_0$ be as in Lemma~\ref{prop:semigroup-p2}.
Then there exist positive constants $K_1$ and $b$ such that
for every $M\ge \eta_0$ and $f\in X$ and $t>0$,
\begin{gather}
\label{eq:corsp21}
\|e^{t\mL}P_2(\eta_0,M)\pd_xf\|_X
\le K_1(1+\eta_0^{-1}+t^{-1/2})e^{-b t}\|f\|_X\,,\\
\label{eq:corsp22}
\|e^{t\mL}P_2(\eta_0,M)\pd_xf\|_X \le
K_1(1+\eta_0^{-1}+t^{-3/4})e^{-b t}\|e^{ax}f\|_{L^1_xL^2_y}\,.
\end{gather}
\end{corollary}

To prove Proposition~\ref{prop:semigroup-p2}, we need decay estimates
for the free semigroup $e^{t\mL_0}$.
\begin{lemma}
  \label{lem:free-semigroup}
Let $a\in(0,2)$. Then there exists a positive constant
$C$ such that for every $f\in C_0^\infty(\R^2)$ and $t>0$,
  \begin{gather}
\label{eq:lemfs1}
\|e^{t\mL_0}f\|_X \le Ce^{-a(4-a^2)t}\|f\|_X\,,\\
\label{eq:lemfs2}
\|e^{t\mL_0}\pd_x f\|_X \le C(1+t^{-1/2})e^{-a(4-a^2)t}\|f\|_X\,,\\
\|e^{t\mL_0}\pd_x^{-1}\pd_y f\|_X \le Ct^{-1/2}e^{-a(4-a^2)t}\|f\|_X\,,\\
\label{eq:lemfs3}
\|e^{t\mL_0}\pd_xf\|_X\le C(1+t^{-3/4})e^{-a(4-a^2)t}\|e^{ax}f\|_{L^2_yL^1_x}\,,\\
\label{eq:lemfs4}
\|e^{t\mL_0}f\|_X\le C(t^{-1/2}+t^{-3/4})e^{-a(4-a^2)t}\|e^{ax}f\|_{L^1(\R^2)}\,.
  \end{gather}
\end{lemma}
\begin{proof}
Let $u(t)$ be a solution to \eqref{eq:LKP0} satisfying the initial condition
$u(0)=f$. Then
$$\hat{u}(t,\xi,\eta)=e^{itp(\xi,\eta)}\hat{f}(\xi,\eta)\,,\quad
p(\xi,\eta)=\xi^3+4\xi-\frac{3\eta^2}{\xi}\,.$$
It follows from Plancherel's theorem that for every $g\in X$,
\begin{equation}
  \label{eq:planchrel}
\|g\|_X^2=\int_{\R^2}e^{2ax}g(x,y)^2dxdy=
\int_{\R^2}|\hat{g}(\xi+ia,\eta)|^2d\xi d\eta\,.
\end{equation}
Making use of \eqref{eq:planchrel} and the fact that
\begin{equation}
  \label{eq:Imp}
\Im p(\xi+ia,\eta)=a(4-a^2)+3a\xi^2+3a\eta^2/(\xi^2+a^2)\,,  
\end{equation}
we have for $j\ge0$,
\begin{align*}
\|\pd_x^je^{t\mL_0}f\|_X\lesssim &
\left\||\xi+ia|^j e^{-t\Im p(\xi+ia,\eta)}|\hat{f}(\xi+ia,\eta)|\right\|_{L^2}
\\ \lesssim & e^{-a(4-a^2)t}\left(\sup_{\xi}(|\xi|+a)^j e^{-3at|\xi|^2}\right)
\|\hat{f}(\cdot+ia,\cdot)\|_{L^2(\R^2)}
\\ \lesssim & e^{-a(4-a^2)t}(1+t^{-j/2})\|f\|_X\,,
\end{align*}
and
\begin{multline*}
\|(\pd_x^{-1}\pd_y)^je^{t\mL_0}f\|_X\lesssim 
\left\|\frac{\eta^j}{(\xi^2+a^2)^{j/2}} 
e^{-t\Im p(\xi+ia,\eta)}|\hat{f}(\xi+ia,\eta)|\right\|_{L^2}
\\ \lesssim  e^{-a(4-a^2)t}\|\hat{f}(\cdot+ia,\cdot)\|_{L^2(\R^2)}
\sup_{\xi\,,\,\eta}\left(\frac{\eta^2}{\xi^2+a^2}\right)^{j/2}e^{-3a\eta^2t/(\xi^2+a^2)}
\\ \lesssim   e^{-a(4-a^2)t}t^{-j/2}\|f\|_X\,.
\end{multline*}
\par

Similarly,
\begin{align*}
\|\pd_x^je^{t\mL_0}f\|_X\lesssim &
\left\|\left\||\xi+ia|^j
e^{-t\Im p(\xi+ia,\eta)}|\hat{f}(\xi+ia,\eta)|\right\|_{L^2_\xi}
\right\|_{L^2_\eta}
\\ \lesssim & e^{-a(4-a^2)t}\left\|(|\xi|+a)^j e^{-3at|\xi|^2}\right\|_{L^2}
\|\hat{f}(\cdot+ia,\cdot)\|_{L^2_\eta L^\infty_\xi}
\\ \lesssim & e^{-a(4-a^2)t}(1+t^{-(2j+1)/4})\|e^{ax}f\|_{L^2_yL^1_x}\,,
\end{align*}
and
\begin{align*}
\|e^{t\mL_0}f\|_X\lesssim &
\left\|e^{-t\Im p(\xi+ia,\eta)}\right\|_{L^2_{\xi,\eta}}
\left\|\hat{f}(\xi+ia,\eta)\right\|_{L^\infty_{\xi,\eta}}
\\ \lesssim & e^{-a(4-a^2)t}
\left\|e^{-3at\xi^2}\|e^{-3at\eta^2/(\xi^2+a^2)}\|_{L^2_\eta}\right\|_{L^2_\xi}
\|\hat{f}(\cdot+ia,\cdot)\|_{L^2_\eta L^\infty_\xi}
\\ \lesssim & e^{-a(4-a^2)t}(t^{-1/2}+t^{-3/4})\|e^{ax}f\|_{L^1(\R^2)}\,.
\end{align*}
This completes the proof of Lemma~\ref{lem:free-semigroup}.
\end{proof}

Combining properties of the linearized Miura transformation
and Lemma~\ref{lem:free-semigroup}, 
we will prove linear decay estimates for non-resonant modes.
\begin{lemma}
\label{cl:nonresonant-low}
Let $a$ and $\eta_*$ be as in Lemma~\ref{lem:Llarge-eta} and let
$\eta_1\in(0,\eta_*)$. Then there exists a positive constant $K$
such that for every $t\ge 0$, $\eta_0\in[-\eta_1,\eta_1]$
and $f\in C_0^\infty(\R^2)$, 
\begin{equation}
\label{eq:nr-low1}
\|e^{t\mL}P_2(\eta_0,\eta_0)f\|_X\le K e^{-a(4-a^2)t}\|f\|_X\,.
\end{equation}
\end{lemma}
\begin{proof}
Since $C_0^\infty(\R)$ is dense in $X$,
it suffices to prove \eqref{eq:nr-low1} for $f\in C_0^\infty(\R)$.
\par
Let $u(t)=e^{t\mL}P_2(\eta_0,\eta_0)f$.
Since $P_0(\eta)$ is a spectral projection associated with $\mL$
(Lemma~\ref{lem:p0} (\ref{it:p05}), we have $P_0(\eta_0)u(t)=0$
for every $t\ge0$.
Let $u_\eta(t,x)=(\mF_yu)(t,x,\eta)$. Then
$u_\eta(t,\cdot)\in L^2_a(\R)$ and $\int_\R u_\eta(t,x)\overline{g(x,\pm\eta)}\,dx=0$
for \textit{a.e.} $\eta\in[-\eta_0,\eta_0]$.
Hence it follows from Lemma~\ref{lem:RangeM+} that
there exists $v_\eta(t,x)\in H^1_a(\R)$ 
such that for $t\ge0$ and \textit{a.e.} $\eta\in[-\eta_0,\eta_0]$, 
\begin{gather}
\notag u_\eta(t,\cdot)=\mM_+(\eta)v_\eta(t,\cdot)\,,
\\  \notag
\int_\R v_\eta(t,x)\overline{g_M^*(x,\eta)}\,dx=0\,,
\\ \label{eq:ueta-veta}
(C_1\|u_\eta(t)\|_{L^2_a(\R)})^2
\le \|v_\eta(t)\|_{H^1_a(\R)}^2+(|\eta|\|\pd_x^{-1}v_\eta(t)\|_{L^2_a(\R)})^2
\le (C_2\|u_\eta(t)\|_{L^2_a(\R)})^2\,,
\end{gather}
where $C_1$ and $C_2$ are positive constants depending only on $a$ and $\eta_1$.
Moreover, 
$$v(t)=\frac{1}{\sqrt{2\pi}}\int_{-\eta_0}^{\eta_0}
v_\eta(t,x,\eta)e^{iy\eta}\,d\eta$$
satisfies $u(t)=\nabla M_+(Q)v(t)$ for every $t\ge0$.
Hence it follows from Lemma~\ref{lem:relation2} that $v(t)$
is a solution of \eqref{eq:LMKP} satisfying.
Moreover, we have for $t\ge0$,
\begin{equation}
\label{eq:vorth}
\int_\R (\mF_yv)(t,x,\eta)\overline{g_M^*(x,\eta)}e^{-iy\eta}\,dx=0\,.
\end{equation}
Integrating \eqref{eq:ueta-veta} over $[-\eta_0,\eta_0]$ and using
Plancherel's theorem, we obtain
\begin{equation}
 \label{eq:equiv-u-v}
C_1\|u(t)\|_X\le \|v(t)\|_{X_M}\le C_2\|u(t)\|_X\,.
\end{equation}
\par
Let $\tilde{u}(t)=\nabla M_-(Q)v(t)$ and
$\tilde{u}_\eta(t,x)=(\mF_y\tilde{u})(t,x)$.
Then $\tilde{u}_\eta(t)=\mM_-(\eta)v_\eta(t)$ and it follows from 
Lemma~\ref{lem:relation} that $\tilde{u}(t)$ is a solution to \eqref{eq:LKP0}.
Using Lemma~\ref{lem:RangeM-}, we can prove that for $t\ge0$,
\begin{equation}
  \label{eq:equiv-tu-v}
C_1'\|\tilde{u}(t)\|_X\le \|v(t)\|_{X_M}\le C_2'\|\tilde{u}(t)\|_X\,,
\end{equation}
in the same way as \eqref{eq:equiv-u-v}.
Here  $C_1'$ and $C_2'$ are positive constants depending only on $a$ and
$\eta_1$.
By Lemma~\ref{lem:free-semigroup},
\begin{equation}
  \label{eq:u0-est}
  \|\tilde{u}(t)\|_X \le C\|\tilde{u}(0)\|_Xe^{-a(4-a^2)t}\,.
\end{equation}
Combining \eqref{eq:equiv-u-v}, \eqref{eq:equiv-tu-v} and \eqref{eq:u0-est},
we obtain \eqref{eq:nr-low1}.
Thus we complete the proof.
\end{proof}

\begin{lemma}
\label{cl:high}
Let $a$ and $\eta_*$ be as in Lemma~\ref{lem:Llarge-eta} and 
let $\eta_2>\eta_*$. Then there exists a positive constant $K$
such that for every $t\ge 0$ and $f\in C_0^\infty(\R^2)$, 
\begin{equation}
\label{eq:nr-high}
\|e^{t\mL}P_1(\eta_2,\infty)f\|_X\le K e^{-a(4-a^2)t}\|f\|_X\,.
\end{equation}
\end{lemma}
Using Lemma~\ref{lem:Mpmlarge} instead of Lemmas~\ref{lem:RangeM-}
and \ref{lem:RangeM+}, we can prove Lemma~\ref{cl:high} in exactly
the same way as Lemma~\ref{cl:nonresonant-low}. Thus we omit the proof.

\par
Middle frequency resonant modes are exponentially stable.
We can obtain decay estimates of these modes by a direct computation.
\begin{lemma}
  \label{cl:decay-resonant}
Let $a$ and $\eta_*$ be as in Lemma~\ref{lem:Llarge-eta}.
Let $\eta_0$ and $\eta_1$ be positive numbers satisfying
$0<\eta_0<\eta_1<\eta_*$. Then for every $f\in X$,
$$\|e^{t\mL}(P_0(\eta_1)-P_0(\eta_0))f\|_{X} \le  C(1+\eta_0^{-1})e^{\Re \lambda(\eta_0)t}\|f\|_X\,,$$
where $C$ is a constant depending only on $a$ and $\eta_1$.
\end{lemma}
\begin{proof}
Let $a_k(t,\eta)=\int_\R(\mF_yu)(t,x,\eta)\overline{g_k^*(x,\eta)}e^{-iy\eta}\,dx$ for
$k=1$, $2$ and let
$$E_a(t,\eta_0,\eta_1)=\int_{\eta_0\le |\eta|\le \eta_1}
(|a_1(t,\eta)|^2+\eta^2|a_2(t,\eta)|^2)d\eta\,.$$
Since $u(t)$ is a solution of \eqref{eq:LKP}, it follows from Lemma~\ref{cl:L*g} that
\begin{equation}
  \label{eq:dif-a1}
\begin{split}
\pd_ta_1(t,\eta)=&\int_\R  \mL(\eta)(\mF_yu)(t,x,\eta)\overline{g_1^*(x,\eta)}\,dx
\\=& \Re\lambda(\eta) a_1-\eta\Im\lambda(\eta)a_2\,,
\end{split}  
\end{equation}
\begin{equation}
  \label{eq:dif-a2}
\begin{split}
\pd_ta_2(t,\eta)=&\int_\R  \mL(\eta)(\mF_yu)(t,x,\eta)\overline{g_2^*(x,\eta)}\,dx
\\=& \eta^{-1}\Im\lambda(\eta) a_1+\Re\lambda(\eta)a_2\,,
\end{split}  
\end{equation}
Using \eqref{eq:dif-a1}, \eqref{eq:dif-a2} and the fact that $\Re\lambda(\eta)$ is even
and monotone decreasing for $\eta\ge0$,
\begin{align*}
\pd_tE_a(t,\eta_0,\eta_1)=& 2\int_{\eta_0\le |\eta|\le \eta_1}
\Re\lambda(\eta)(|a_1(t,\eta)|^2+\eta^2|a_2(t,\eta)|^2)d\eta
\\ \le & 2\Re\lambda(\eta_0)E_a(t,\eta_0,\eta_1)\,.
\end{align*}
Thus we have for $t\ge0$,
\begin{equation}
  \label{eq:Ea-bound}
E_a(t,\eta_0,\eta_1)\le E_a(0,\eta_0,\eta_1)e^{2\Re\lambda(\eta_0)t}\,.
\end{equation}
\par

As in the proof of Lemma~\ref{lem:p0}, we have
\begin{equation}
  \label{eq:f-bound3}
 \|e^{t\mL}(P_0(\eta_1)-P_0(\eta_0))f\|_{X}
\le C_1(1+\eta_0^{-1}) E_a(t,\eta_0,\eta_1)^{1/2}\,,
\end{equation}
\begin{align}
\label{eq:f-bound4}
E_a(0,\eta_0,\eta_1)^{1/2} \le C_2 \|f\|_X\,,
\quad E_a(0,\eta_0,\eta_1)^{1/2} \le C_2 \|e^{ax}f\|_{L^1_xL^2_y}\,,
\end{align}
where $C_1$ and $C_2$ are constants depending only on $a$ and $\eta_1$.
By \eqref{eq:Ea-bound},
\begin{equation}
  \label{eq:f-bound6}
E_a(t,\eta_0,\eta_1) \le  e^{2\Re\lambda(\eta_0)t} E_a(0,\eta_0,\eta_1)\,.
\end{equation}
Combining \eqref{eq:f-bound3}--\eqref{eq:f-bound6}, we obtain
Lemma~\ref{cl:decay-resonant}.
Thus we complete the proof.
\end{proof}

Now we are in position to prove Proposition~\ref{prop:semigroup-p2}.
\begin{proof}[Proof of Proposition~\ref{prop:semigroup-p2}]
Let $a_1$, $a_2$, $\eta_1$ and $\eta_2$ be positive numbers satisfying
$a_1<\nu(\eta_1)<a<\nu(\eta_2)<a_2$. 
Note that $\eta_0<\eta_1<\eta_*<\eta_2$, where $\eta_*$ is a root of
$\nu(\eta)=a$. Since
$$P_2(\eta_0,\infty)=P_2(\eta_1,\eta_1)+P_0(\eta_1)-P_0(\eta_0)
+P_1(\eta_1,\eta_2)+P_1(\eta_2,\infty)\,,$$
it follows from Lemmas~\ref{cl:nonresonant-low}, \ref{cl:high}
and \ref{cl:decay-resonant} that
\begin{equation}
  \label{eq:sem-pf1}
\|e^{t\mL}(P_2(\eta_0,\infty)-P_1(\eta_1,\eta_2))f\|_X \le 
(e^{-a(4-a^2)t}+(1+\eta_0^{-1})e^{\Re\lambda(\eta_0)t})\|f\|_X\,.
\end{equation}
\par
In order to estimate $\|e^{t\mL}P_1(\eta_1,\eta_2)f\|_X$,
we will interpolate the decay estimate of $e^{t\mL}P_1(\eta_1,\eta_2)$
in $L^2(\R^2;e^{2a_jx}dxdy)$ ($j=1$, $2$).
Since $\nu(\eta_2)<a_2$, it follows from Lemmas~\ref{cl:nonresonant-low}
and \ref{cl:decay-resonant} that for $t\ge0$,
\begin{gather*}
\|e^{a_2x}e^{t\mL}P_2(\eta_2,\eta_2)f\|_{L^2(\R^2)}\lesssim
e^{-a_2(4-a_2^2)t}\|e^{a_2x}f\|_{L^2(\R^2)}\,,
\\
\|e^{a_2x}e^{t\mL}(P_0(\eta_2)-P_0(\eta_1))f\|_{L^2(\R^2)}\lesssim
(e^{-a_2(4-a_2^2)t}+\eta_1^{-1}e^{\Re\lambda(\eta_1)t})\|e^{a_2x}f\|_{L^2(\R^2)}\,.
\end{gather*}
Since $P_1(\eta_1,\eta_2)=P_2(\eta_1,\eta_2)+P_0(\eta_2)-P_1(\eta)$,
\begin{equation*}
\|e^{a_2x}e^{t\mL}P_1(\eta_1,\eta_2)f\|_{L^2(\R^2)}
 \lesssim (e^{-a_2(4-a_2^2)t}+e^{\Re\lambda(\eta_1)t})\|e^{a_2x}f\|_{L^2(\R^2)}\,.
\end{equation*}
On the other hand, Lemma~\ref{cl:high} implies that
$$\|e^{a_1x}e^{t\mL}P_1(\eta_1,\eta_2)f\|_{L^2(\R^2)}
\lesssim \|e^{a_1x}f\|_{L^2(\R^2)}\,.$$
Hence it follows from  the complex interpolation theorem that
\begin{equation}
  \label{eq:sem-pf2}
\|e^{t\mL}P_1(\eta_1,\eta_2)f\|_X\lesssim 
\left\{e^{-a_1(4-a_1^2)t}+e^{-a_2(4-a_2^2)t}+(1+\eta_1^{-1})e^{\Re\lambda(\eta_1)t}
\right\}\|f\|_X\,.
\end{equation}
\par
By \eqref{eq:sem-pf1} and \eqref{eq:sem-pf2}, we obtain
\begin{align*}
\|e^{t\mL}(P_2(\eta_0,\infty)f\|_X  \lesssim &
\{e^{-a(4-a^2)t}+(1+\eta_0^{-1})e^{\Re(\lambda(\eta_0)t}\}\|f\|_X
\\ & +\{e^{-a_1(4-a_1^2)t}+e^{-a_2(4-a_2^2)t}
+(1+\eta_1^{-1})e^{\Re\lambda(\eta_1)t}\}\|f\|_X\,.
\end{align*}
Thus we complete the proof of Proposition~\ref{prop:semigroup-p2}.
\end{proof}

\begin{proof}[Proof of Corollary~\ref{cor:semigroup-p2}]
Without loss of generality, we may assume that $M=\infty$.
By the variation of constants formula, we have for any $f\in X$,
\begin{equation}
\label{eq:var1}
\begin{split}
e^{t\mL}P_2(\eta_0,\infty)\pd_xf=& e^{t\mL_0}P_2(\eta_0,\infty)\pd_xf\\
&-6\int_0^t \pd_xe^{(t-s)\mL_0}\left(\varphi e^{s\mL}P_2(\eta_0,\infty)\pd_xf\right)
ds\,.  
\end{split}
\end{equation}
Let $t\in (0,2]$. 
Applying Proposition~\ref{prop:semigroup-p2} and
Lemmas~\ref{lem:p0}, \ref{lem:free-semigroup} to \eqref{eq:var1},
\begin{align*}
\|e^{t\mL}P_2(\eta_0,\infty)\pd_xf\|_X \le & \|e^{t\mL_0}P_2(\eta_0,\infty)\pd_xf\|_X
\\ &   + 6\int_0^t \left\|\pd_x e^{(t-s)\mL_0}
\left(\varphi e^{s\mL}P_2(\eta_0,\infty)\pd_xf\right)\right\|_X
\\ \lesssim &  (1+t^{-1/2})\|f\|_X
+\int_0^t(t-s)^{-1/2}\|e^{s\mL}P_2(\eta_0,\infty)\pd_xf\|_{L^2}\,.
\end{align*}
By Gronwall's inequality, we have
\begin{equation}
  \label{eq:p21}
\|e^{t\mL}P_2(\eta_0,\infty)\pd_xf\|_X \le Ct^{-1/2}\|f\|_X\quad\text{for $t\in(0,2]$,}
\end{equation}
where $C$ is a constant independent of $t\in(0,2]$ and $f\in X$.
\par 
Let $t\ge 2$. Eq.~\eqref{eq:p21} implies that $e^{\mL}P_2(\eta_0,\infty)\pd_x$ is bounded
on $X$. Applying Proposition~\ref{prop:semigroup-p2} to
$e^{t\mL}P_2(\eta_0,\infty)\pd_x=e^{(t-1)\mL}P_2(\eta_0,\infty)e^{\mL}
P_2(\eta_0,\infty)\pd_x$, we have for $t\ge 2$,
\begin{align*}
\|e^{t\mL}P_2(\eta_0,\infty)\pd_x f\|_X  \lesssim e^{-bt}\|f\|_X\,.
\end{align*}
Combining the above with \eqref{eq:p21}, we obtain \eqref{eq:corsp21}.
\par
Using \eqref{eq:lemfs2} and \eqref{eq:lemfs3} and Lemma~\ref{prop:semigroup-p2},
we can prove \eqref{eq:corsp22} in the same way as \eqref{eq:corsp21}.
\end{proof}
\bigskip

\section{Preliminaries}
\label{sec:preliminaries}
To begin with, we will introduce notation of Banach spaces which shall be used
to analyze modulation equations. For an $\eta_0>0$, let
$Y$ and $Z$ be closed subspaces of $L^2(\R)$ defined by
$$Y=\mF^{-1}_\eta Z\quad\text{and}\quad
Z=\{f\in L^2(\R)\mid\supp f \subset[-\eta_0,\eta_0]\}\,,$$
and let $Y_1=\mF^{-1}_\eta Z_1$ and
$Z_1=\{f\in Z \mid\|f\|_{Z_1}:=\|f\|_{L^\infty}<\infty\}$.
\begin{remark}
  \label{rem:smoothness}
We have
\begin{equation}
  \label{eq:H^s-Y}
\|f\|_{H^s}\le (1+\eta_0^2)^{s/2}\|f\|_{L^2}
\quad\text{for any $s\ge0$ and $f\in Y$,}  
\end{equation}
since $\hat{f}$ is $0$ outside of $[-\eta_0,\eta_0]$.
Especially, we have $\|f\|_{L^\infty}\lesssim \|f\|_{L^2}$ for any $f\in Y$.
\par
Let $\wP_1$ be a projection defined by
$\wP_1 f=\mF_\eta^{-1}\mathbf{1}_{[-\eta_0,\eta_0]}\mF_yf$,
where $\mathbf{1}_{[-\eta_0,\eta_0]}(\eta)=1$ for $\eta\in[-\eta_0,\eta_0]$
and $\mathbf{1}_{[-\eta_0,\eta_0]}(\eta)=0$ for $\eta\not\in[-\eta_0,\eta_0]$.
Then $\|\wP_1f\|_{Y_1}\le (2\pi)^{-1/2}\|f\|_{L^1}(\R)$ for any $f\in L^1(\R)$.
In particular, for any $f$, $g\in Y$,
\begin{equation}
  \label{eq:Y1-L1}
\|\wP_1(fg)\|_{Y_1}\le (2\pi)^{-1/2}\|fg\|_{L^1}\le (2\pi)^{-1/2}\|f\|_Y\|g\|_Y\,.
\end{equation}
\end{remark}

In order to estimate modulation parameters $c(t,y)$ and $x(t,y)$,
we will use a linear estimate for solutions to
\begin{equation}
\label{eq:cx-linear}
\frac{\pd u}{\pd t}=A(t)u\,,
\end{equation}
where $A(t)=A_0(D_y)+A_1(t,D_y)$,
$u(t,y)={}^t(u_1(t,y)\,,\, u_2(t,y))$,
$$A_0(D_y)=\begin{pmatrix}a_{11}(D_y) & a_{12}(D_y)\pd_y
\\ a_{21}(D_y)\pd_y & a_{22}(D_y)\end{pmatrix}\,,\quad
A_1(t,D_y)=\begin{pmatrix}b_{11}(t,D_y) & b_{12}(t,D_y)\\
b_{21}(t,D_y) & b_{22}(t,D_y)\end{pmatrix}\,,$$
and $a_{ij}(\eta)$ and $b_{ij}(t,\eta)$
are continuous in $\eta\in[-\eta_0,\eta_0]$ and $t\ge0$.
We denote by $U(t,s)f$  a solution of \eqref{eq:cx-linear} 
satisfying $u(s,y)=f(y)$. Then we have the following.
\begin{lemma}
\label{lem:decay-BB}
Let $k\in\Z_{\ge0}$, $\mu>1/8$. Let $\delta_1$, $\delta_2$, $\kappa$
be positive constants.  Suppose that $a_{ij}(\eta)$, $b_{ij}(t,\eta)$
$(i$, $j=1$, $2)$ satisfy
\begin{equation}
  \label{eq:H}
  \begin{split}
& |a_{11}(\eta)+3\eta^2|\le \delta_1\eta^2\,,
\quad |a_{12}(\eta)-8|\le \delta_1|\eta|\,,\\
& |a_{21}(\eta)-(2+\mu\eta^2)| \le \delta_1\eta^2\,,
\quad |a_{22}(\eta)+\eta^2|\le \delta_1\eta^2\,,\\
& |b_{ij}(t,\eta)|\le \delta_2 e^{-\kappa t}\quad\text{for $j=1$, $2$.}
  \end{split}  \tag{H}
\end{equation}
If $\delta_1$ is sufficiently small, then for every $t\ge s\ge0$ and $f\in Y$,
\begin{align}
\label{eq:BB1}
\|\pd_y^kU(t,s)f\|_Y\le C(1+t-s)^{-k/2}\|f\|_Y\,,\\
\label{eq:BB2}
\|\pd_y^kU(t,s)f\|_Y\le C(1+t-s)^{-(2k+1)/4}\|f\|_{Y_1}\,,
\end{align}
where $C=C(\eta_0)$ is a constant satisfying
$\limsup_{\eta_0\downarrow0}C(\eta_0)<\infty$.
\end{lemma}
\begin{proof}
We will prove Lemma~\ref{lem:decay-BB} by the energy method. Let
\begin{equation}
\begin{split}
  \label{eq:A*}
& \omega(\eta)=\sqrt{16+(8\mu-1)\eta^2}\,,\\
& A_*(\eta)=\begin{pmatrix}
-3\eta^2 & 8i\eta \\ i\eta(2+\mu\eta^2) & -\eta^2\end{pmatrix}\,,
\quad \Pi_*(\eta)=
\begin{pmatrix} 8i & 8i \\ \eta+i\omega(\eta) & \eta-i\omega(\eta)
\end{pmatrix}\,.
\end{split}
\end{equation}
The matrix $A_*(\eta)$ has eigenvalues $\lambda_*^\pm=-2\eta^2\pm i\eta\omega$
and  $\Pi_*(\eta)^{-1}A_*(\eta)\Pi_*(\eta)=
\diag(\lambda^+_*(\eta)\,,\,\lambda^-_*(\eta))$.
By the assumption, there exist
eigenvalues $\lambda^\pm(\eta)$ and an eigensystem $\Pi(\eta)$ of $A_0(\eta)$
satisfying for $\eta\in[-\eta_0,\eta_0]$,
\begin{gather*}
|\lambda^\pm(\eta)-\lambda_*^\pm(\eta)|\lesssim \delta_1\eta^2\,,\quad
|\Pi(\eta)-\Pi_*(\eta)|\lesssim \delta_1\,.
\end{gather*}
Let $\Lambda(\eta)=\diag(\lambda^+(\eta), \lambda^-(\eta))$,
$B(t,\eta)=\Pi(\eta)^{-1}A_1(t,\eta)\Pi(\eta)$ and
\begin{equation*}
\mathbf{e}(t,\eta)=\begin{pmatrix}e_+(t,\eta)\\ e_-(t,\eta)\end{pmatrix}
=\Pi(\eta)^{-1}(\mF_yu)(t,\eta)\,.
\end{equation*}
Then \eqref{eq:cx-linear} can be rewritten as
\begin{equation*}
\pd_t  \mathbf{e}(t,\eta)=(\Lambda(\eta)+B(t,\eta))\mathbf{e}(t,\eta)\,.
\end{equation*}
Differentiating the energy function $e(t,\eta):=|e_+(t,\eta)|^2+|e_-(t,\eta)|^2$ with respect to $t$,
we have
\begin{equation}
  \label{eq:Et}
\begin{split}
\pd_t e(t,\eta)=& 2\sum_\pm\Re \lambda_\pm(\eta)|e_\pm(t,\eta)|^2
+2\Re\la B(t,\eta)\mathbf{e}(t,\eta)\,,\, \mathbf{e}(t,\eta)\ra_{\C^2}
\\ \le & (-4+O(\delta_1))\eta^2e(t,\eta)
+C\delta_2 e^{-\kappa t}e(t,\eta)\,,
\end{split}  
\end{equation}
where $C$ is a positive constant and 
$\la \cdot,\cdot\ra_{\C^2}$ is the standard inner product on $\C^2$.
By Gronwall's inequality, there exists a positive constant $c_3$ such that
\begin{equation}
  \label{eq:Et-bound}
  e(t,\eta)\le c_3e(s,\eta)e^{(-4+O(\delta_1))\eta^2(t-s)}
\quad\text{for $t\ge s\ge0$.}
\end{equation}
Since $\|\pd_y^ku(t)\|_Y^2\simeq \int_{|\eta|\le \eta_0}
\eta^{2k}e(t,\eta)d\eta$,
\begin{align*}
\|\pd_y^ku(t)\|_Y^2
\lesssim & \left(\sup_{|\eta|\le \eta_0}\eta^{2k}e^{-4\eta^2(t-s)}\right)
\int_{|\eta|\le \eta_0} e(s,\eta)d\eta
\\ \lesssim & \la t-s\ra^{-k}\|u(s)\|_Y^2
\quad\text{for $t\ge s\ge 0$.}
\end{align*}
Similarly, we have
\begin{align*}
\|\pd_y^ku(t)\|_Y^2\lesssim &
 \sup_{|\eta|\le \eta_0}e(s,\eta)
\int_{|\eta|\le \eta_0} \eta^{2k}e^{(-4+O(\delta_1))\eta^2(t-s)} d\eta
\\ \lesssim & \la t-s\ra^{-(2k+1)/2}\|u(s)\|_{Y_1}^2
\quad\text{for $t\ge s\ge 0$.}
\end{align*}
Thus we complete the proof
\end{proof}

Let $A_*=\mF^{-1}A_*(\eta)\mF$, where $A_*(\eta)$ is a matrix
defined by \eqref{eq:A*}.
For a specific choice of $\mu$, we can express the semigroup
$e^{tA_*}$ by using the kernel $H_t(y)=(4\pi t)^{-1/2}e^{-y^2/4t}$.
\par
\begin{lemma}
  \label{lem:mu=1/8}
Let $\mu=1/8$ and $(f_1,f_2)\in Y\times Y$. Then
\begin{equation}
  \label{eq:expA*}
e^{tA_*}\begin{pmatrix}f_1\\ f_2\end{pmatrix}
=\begin{pmatrix}k_{11}(t,\cdot)*f_1  + k_{12}(t,\cdot)*f_2 \\
k_{21}(t,\cdot)*f_1+ k_{22}(t,\cdot)*f_2 \end{pmatrix}\,,  
\end{equation}
where
\begin{align*}
& k_{11}(t,y)=\left(\frac12+\frac18\pd_y\right)H_{2t}(y+4t)
+\left(\frac12-\frac18\pd_y\right)H_{2t}(y-4t)\,,\\
& k_{12}(t,y)=H_{2t}(y+4t)-H_{2t}(y-4t)\,,\\
& k_{21}(t,y)=\left(\frac{1}{4}-\frac{1}{64}\pd_y^2\right)
\left(H_{2t}(y+4t)-H_{2t}(y-4t)\right)\,,\\
& k_{22}(t,y)=\left(\frac12-\frac18\pd_y\right)H_{2t}(y+4t)
+\left(\frac12+\frac18\pd_y\right)H_{2t}(y-4t)\,.
\end{align*}
Moreover, for every $k\in \Z_{\ge0}$, there exists a positive constant $C$
such that
$$\|\pd_y^ke^{tA_*}\|_{B(Y,Y)}\le C\la t\ra^{-k/2}\,,\quad
\|\pd_y^ke^{tA_*}\|_{B(Y_1,Y)}\le C\la t\ra^{-(2k+1)/4}\,.$$
\end{lemma}
\begin{proof}
In view of the proof of Lemma~\ref{lem:decay-BB},
$$e^{tA_*}(\eta)=e^{-2t\eta^2}\begin{pmatrix}
\cos4t\eta-\frac{\eta}{4}\sin4t\eta & 2i\sin4t\eta\\
\left(\frac{i\eta^2}{32}+\frac{i}{2}\right)\sin4t\eta
& \cos4t\eta+\frac{\eta}{4}\sin4t\eta\end{pmatrix}\,,$$
provided $\mu=1/8$. 
Taking the inverse Fourier transform of the above, we obtain
\eqref{eq:expA*}. The decay estimates follow immediately
from \eqref{eq:expA*} and the fact that
$\|\pd_y^kH_t\|_{B(Y,Y)}\lesssim \la t\ra^{-k/2}$ and
$\|\pd_y^kH_t\|_{B(Y_1,Y)}\lesssim \la t\ra^{-(2k+1)/4}$.
Thus we complete the proof.
\end{proof}

Using \eqref{eq:expA*} and the fact that 
$\|\pd_y^kH_{2t}*f\|_Y\lesssim \la t\ra^{-(2k+1)/4}\|f\|_{Y_1}$
for $t\ge0$ and $k\in \Z_{\ge0}$, we can obtain the first order
asymptotics of $e^{tA_*}(f_1,f_2)$ as $t\to\infty$.
\begin{corollary}
  \label{cor:mu=1/8}
Let $\mu$ and $A_*$ be as in Lemma~\ref{lem:mu=1/8}. Then there
exists a positive constant $C$ such that for every
$(f_1,f_2)\in Y_1\times Y_1$,
$$\left\|e^{tA_*}\begin{pmatrix}f_1\\ f_2\end{pmatrix}
-e^{4t\pd_y}H_{2t}* \begin{pmatrix} 2f_+ \\ f_+ \end{pmatrix}
-e^{-4t\pd_y}H_{2t}*\begin{pmatrix} 2f_- \\ -f_- \end{pmatrix}\right\|_Y
\le C\la t\ra^{-3/4}\sum_{i=1,2}\|f_i\|_{Y_1}\,,$$
where $f_+=\frac{1}{4}f_1+\frac12f_2$ and $f_-=\frac{1}{4}f_1-\frac12f_2$.
\end{corollary}

To estimate inhomogeneous terms of modulation equations,
we will use the following.
\begin{claim}
\label{cl:abc}
Let $\alpha$ and $\beta$ be positive constants and
$\gamma=\min\{\alpha,\beta,\alpha+\beta-1\}$. 
If $(\alpha,\beta)$ satisfies
\textrm{i)} $\alpha\ne1$ and $\beta\ne1$
or \textrm{ii)} $\alpha>\beta=1$ or \textrm{(iii)} $\beta>\alpha=1$,
then there exists a positive constant $C$ such that 
$$\int_0^t \la t-s\ra^{-\alpha}\la s\ra^{-\beta}ds \le C\la t\ra^{-\gamma}\,.$$
\par
If $\alpha <1$ and $\beta=1$ or $\alpha=1$ and $\beta<1$, then
there exists a $C>0$ such that 
$$\int_0^t \la t-s\ra^{-\alpha}\la s\ra^{-\beta}ds \le C\la t\ra^{-\min(\alpha,\beta)}
\log\la t\ra\,.$$
\end{claim}
\bigskip

\section{Decomposition of the perturbed line soliton}
\label{sec:decomp}
In this section, we will decompose a solution around a line soliton solution
$\varphi(x-4t)$ into a sum of a modulating line soliton
and a non-resonant dispersive
part plus a small wave which is caused by amplitude changes of the line
soliton:
\begin{equation}
  \label{eq:decomp}
u(t,x,y)=\varphi_{c(t,y)}(z)-\psi_{c(t,y),L}(z+4t)+v(t,z,y)\,,\quad
z=x-x(t,y)\,.
\end{equation}
The modulation parameters $c(t_0,y_0)$ and $x(t_0,y_0)$ denote
the maximum height and the phase shift of the modulating line soliton
$\varphi_{c(t,y)}(x-x(t,y))$ along the line $y=y_0$ at the time $t=t_0$,
and $\psi_{c,L}$ is an auxiliary function such that
\begin{equation}
  \label{eq:0mean}
\int_\R \psi_{c,L}(x)\,dx=\int_\R(\varphi_c(x)-\varphi(x))\,dx\,.
\end{equation}
More precisely, 
$$\psi_{c,L}(x)=2(\sqrt{2c}-2)\psi(x+L)\,,$$
where $L>0$ is a large constant to be fixed later and 
$\psi(x)$ is a nonnegative function such that
$\psi(x)=0$ if  $|x|\ge1$ and that $\int_\R \psi(x)\,dx=1$.
Since a localized solution to KP-type equations
satisfies $\int_\R u(t,x,y)\,dx=0$ for any $y\in\R$ and $t>0$
(see \cite{MST_contr}), it is natural to expect small perturbations
appear in the rear of the solitary wave if the solitary wave is amplified.
\par
To fix the decomposition \eqref{eq:decomp}, we impose that 
$v(t,z,y)$ is symplectically orthogonal to low frequency resonant modes.
More precisely, we impose the constraint that for $k=1$, $2$,
\begin{equation}
  \label{eq:orth}
\lim_{M\to\infty}\int_{-M}^M \int_\R
v(t,z,y)\overline{g_k^*(z,\eta,c(t,y))}e^{-iy\eta}\,dzdy=0
\quad\text{in $L^2(-\eta_0,\eta_0)$,}
\end{equation}
where 
$g_1^*(x,\eta,c)=cg_1^*(\sqrt{c/2}x,\eta)$ and
$g_2^*(x,\eta,c)=\frac{c}{2}g_2^*(\sqrt{c/2}x,\eta)$.
\par
We will show that the decomposition \eqref{eq:decomp} with
\eqref{eq:orth} is well defined if $u$ is close to a modulating line soliton
in the exponentially weighted space $X$.
It is expected that $\|c(t,\cdot)-2\|_{L^\infty}$ remains small as long as
\eqref{eq:decomp} persists.
\par

Now let us introduce functional to prove the existence of
the representation \eqref{eq:decomp} that satisfies 
the orthogonality condition \eqref{eq:orth}.
For $v\in X$ and $\gamma$, $\tc\in Y$ and $L\ge 0$, let 
$c(y)=2+\tc(y)$ and 
\begin{align*}
F_k[u,\tc,\gamma,L](\eta):=\mathbf{1}_{[-\eta_0,\eta_0]}&(\eta)
\lim_{M\to\infty}\int_{-M}^M\int_\R
\bigl\{u(x,y)+\varphi(x)- \varphi_{c(y)}(x-\gamma(y))
\\ & +\psi_{c(y),L}(x-\gamma(y))\bigr\}
\overline{g_k^*(x-\gamma(y),\eta,c(y))}e^{-iy\eta}\,dxdy\,.
\end{align*}
To begin with, we will show that $F=(F_1,F_2)$ is a mapping from
$X\times Y\times Y\times \R$ into $Z\times Z$.
\begin{lemma}
\label{lem:F_k}
Let $a\in(0,2)$,  $u\in X=L^2(\R^2;e^{2ax}dxdy)$, $\tc$, $\gamma\in Y$
and $L\ge 0$.
Then there exists a $\delta>0$ such that if 
$\|\tc\|_Y+\|\gamma\|_Y\le \delta$, then
$F_k[u,\tc,\gamma,L]\in Z$ for $k=1$, $2$.
Moreover, if $u\in X_1:=L^1(\R_y;L^2_a(\R_x))$ and $\tc$, $\gamma\in Y_1$, then
$F_k[u,\tc,\gamma,L]\in Z_1$ for $k=1$, $2$.
\end{lemma}
\begin{proof}
Let $u\in C_0^\infty(\R^2)$ and
\begin{align*}
& \Phi_1(x,y)=\varphi_{c(y)}(x-\gamma(y))-\varphi(x)
-\psi_{c(y),L}(x-\gamma(y))\,,\\
& \Phi_{1,0}(x,y)=\pd_c\varphi(x)\tc(y)-\varphi'(x)\gamma(y)
-\psi_{c(y),L}(x)\,,\\
& \Phi_2(x,y)=\Phi_1(x,y)-\Phi_{1,0}(x,y)\,,\quad
\Psi(x,y)=\overline{g_k^*(x-\gamma(y),\eta,c(y))}-\overline{g_k^*(x,\eta)}
\,.  
\end{align*}
Then
\begin{align*}
\int_{\R^2}\left\{u(x,y)-\Phi_1(x,y)\right\}
\overline{g_k^*(x-\gamma(y),\eta,c(y))}e^{-iy\eta}\,dxdy=\sum_{j=1}^4 I_j(\eta)\,.
\end{align*}
where
\begin{align*}
I_1=& \int_{\R^2}u(x,y)\overline{g_k^*(x,\eta)}e^{-iy\eta}\,dxdy\,,\\
I_2=& -\int_{\R^2}\Phi_{1,0}(x,y)\overline{g_k^*(x,\eta)}e^{-iy\eta}\,dxdy\,,\\
I_3=& -\int_{\R^2}\Phi_2(x,y)\overline{g_k^*(x,\eta)}e^{-iy\eta}\,dxdy\,,\\
I_4=& \int_{\R^2}\{u(x,y)-\Phi_1(x,y)\}\Psi(x,y)e^{-iy\eta}\,dxdy\,.\\
\end{align*}
By Claim~\ref{cl:gk-approx}, 
\begin{equation}
  \label{eq:gk-bound}
\sup_{c\in[2-\delta,2+\delta]}\sup_{\eta\in[-\eta_0, \eta_0]}
\left\|\pd_c^j\pd_x^kg_k^*(\cdot,\eta,c)\right\|_{L^2_{-a}(\R)}<\infty
\enskip\text{ for $j$, $k\ge0$,}
\end{equation}
and it follows from Plancherel's theorem and \eqref{eq:gk-bound} that
\begin{align*}
 \int_{-\eta_0}^{\eta_0}|I_1(\eta)|^2d\eta \lesssim \|u\|_X^2\,,
\enskip
 \int_{-\eta_0}^{\eta_0}|I_2(\eta)|^2d\eta \lesssim 
\|\tc\|_Y^2+\|\gamma\|_Y^2\,.
\end{align*}
Since $\sup_y(|\tc(y)|+|\gamma(y)|)\lesssim \|\tc\|_Y+\|\gamma\|_Y$, we have
\begin{gather*}
\|\Phi_1\|_X+\|\Phi_{1,0}\|_X\le C(\|\tc\|_Y+\|\gamma\|_Y)\,,\\  
\|e^{ax}\Phi_2(x,y)\|_{L^1(\R^2)}\le C(\|\tc\|_Y+\|\gamma\|_Y)^2\,,\\
\|e^{-ax}\Psi(x,y)\|_{L^2(\R^2)}\le C(\|\tc\|_Y+\|\gamma\|_Y)\,,
\end{gather*}
where $C$ is a positive constant depending only on $\delta$.
\par
Combining the above, we obtain
$$\sup_{-\eta_0\le \eta\le \eta_0}(|I_3(\eta)|+|I_4(\eta)|)
\lesssim \|u\|_X(\|\tc\|_Y+\|\gamma\|_Y)
+(\|\tc\|_Y+\|\gamma\|_Y)^2\,.$$
Since $C_0^\infty(\R^2)$ is dense in $X$, it follows that for any $u\in X$,
$$\mathbf{1}_{[-\eta_0,\eta_0]}(I_1+I_2)\in Z\,,\quad
\mathbf{1}_{[-\eta_0,\eta_0]}(I_3+I_4)\in Z_1\subset Z\,.$$
\par
Suppose $u\in X_1$ and $\tc$, $\gamma\in Y_1$.
Noting that $\sqrt{2c}-2=\tc/2+O(\tc^2)\in Y_1$, we have
\begin{align*}
& \sup_{[-\eta_0,\eta_0]}|I_1(\eta)|\lesssim \|u\|_{X_1}\,,\quad
\sup_{[-\eta_0,\eta_0]}|I_2(\eta)|\lesssim 
\|\tc\|_{Y_1}+\|\gamma\|_{Y_1}\,,
\end{align*}
and $\mathbf{1}_{[-\eta_0,\eta_0]}\sum_{1\le i\le 4}I_i\in Z_1$.
Thus we complete the proof.
\end{proof}

\begin{lemma}
  \label{lem:decomp}
Let $a\in(0,2)$.
There exist positive constants $\delta_0$, $\delta_1$, $L_0$ and $C$ such that
if $\|u\|_X<\delta_0$ and $L\ge L_0$, then there exists a unique
$(\tc,\gamma)$ with $c=2+\tc$ satisfying
\begin{gather}
\label{eq:imp1}
  \|\tc\|_{Y}+\|\gamma\|_{Y} <\delta_1\,,\\
\label{eq:imp2}
F_1[u,\tc,\gamma,L]=F_2[u,\tc,\gamma,L]=0\,.
\end{gather}
Moreover, the mapping $\{u\in X\mid \|u\|_X<\delta_0\}
\ni u\mapsto (\tc,\gamma)=:\Phi(u)$ is $C^1$.
\end{lemma}

\begin{proof}
Clearly, we have $(F_1,F_2)\in C^1(X\times Y\times Y\times\R;Z\times Z)$ and 
for $\tc$, $\tilde{\gamma}\in Y$, 
$$D_{(c,\gamma)}(F_1,F_2)(0,0,0,L)
\begin{pmatrix}\tc \\ \tilde{\gamma}\end{pmatrix}
=\sqrt{2\pi}\begin{pmatrix}f_{11} & f_{12} \\ f_{21} & f_{22}\end{pmatrix}
\begin{pmatrix}
\mF_y\tc \\ \mF_y\tilde{\gamma}\end{pmatrix}\,,
$$
where 
\begin{align*}
& f_{11}=-\int_\R(\pd_c\varphi(x)-\psi(x+L))g_1^*(x,\eta)\,dx\,,\quad
 f_{12}=\int_\R\varphi'(x)g_1^*(x,\eta)\,dx\,,\\
& f_{21}=-\int_\R(\pd_c\varphi(x)-\psi(x+L))g_2^*(x,\eta)\,dx\,,\quad
 f_{22}=\int_\R\varphi'(x)g_2^*(x,\eta)\,dx\,.
\end{align*}
By Claims~\ref{cl:gk-approx} and \ref{cl:phi-int} in Appendix~\ref{ap:G},
\begin{align*}
& f_{11}=-1+O(\eta_0^2)+O(e^{-aL})\,,\quad f_{12}=O(\eta_0^2)\,,\\
& f_{21}=-\frac12+O(\eta_0^2)+O(e^{-aL})\,,\quad f_{22}=-2+O(\eta_0^2)\,,
\end{align*}
and $D_{(c,\gamma)}(F_1,F_2)(0,0,0,L)\in B(Y\times Y,Z\times Z)$ has a bounded
inverse if $\delta_0$ and $e^{-aL}$ are sufficiently small.
Hence it follows from the implicit function theorem that for any
$u$ satisfying $\|u\|_X<\delta_0$,
there exists a unique $(c,\gamma)\in Y\times Y$ 
satisfying \eqref{eq:imp1} and \eqref{eq:imp2}.
Moreover, the mapping $(\tc,\gamma)=\Phi(u)$ is $C^1$.
\end{proof}
\begin{remark}
\label{rem:decomp}
In Lemma \ref{lem:decomp}, we can replace $X$ by a Banach space $X_2$ whose
norm is
$$\|u\|_{X_2}=\left(\int_{-\eta_0}^{\eta_0}\int_{\R}
\frac{|\hat{u}(\xi+ia,\eta)|^2}{(1+\xi^2)^3}\,d\xi d\eta\right)^{1/2}\,.
$$
Suppose $u(t,x,y)$ is a solution of \eqref{KPII_integrated} satisfying
$u(0,x,y)=\varphi(x)+v_0(x,y)$ and $v_0\in X\cap H^1(\R)$. Then for any $T>0$,
\begin{equation}
  \label{eq:lbd}
\tv(t,x,y):=u(t,x+4t,y)-\varphi(x)\in C([0,\infty);X)\,,  
\end{equation}
(see Proposition~\ref{prop:LWPX} in Appendix~\ref{ap:LWP}).
Combining \eqref{eq:lbd} and the fact that
$\tv(t)$ is a solution of \eqref{eq:v-fix}, we have
$\pd_tP_1(0,\eta_0)u\in C([0,\infty);X_2)$ and 
\begin{equation}
  \label{eq:tuC1}
P_1(0,\eta_0)\tv(t)\in C^1([0,\infty);X_2)\,.
\end{equation}
If $\sup_{t\in[0,T)}\|\tv(t)\|_{X_2}$ is sufficiently small for a $T>0$,
then there exists $(\tc(t),\tx(t)):=\Phi(\tv(t))$ satisfying
\eqref{eq:orth} for  $t\in[0,T)$,
where $c(t,y)=\tc(t,y)+2$ and $x(t,y)=4t+\tx(t,y)$ and $v$ and $z$ are
defined by \eqref{eq:decomp}.  That is, the decomposition
\eqref{eq:decomp} satisfying \eqref{eq:orth} 
exists on $[0,T]$ if $\|v_0\|_{X_2}\lesssim \|v_0\|_X$
is sufficiently small.
Since $X_2\ni u\mapsto \Phi(u)\in Y\times Y$ is $C^1$, it follows from
\eqref{eq:tuC1} that
$$(\tc(t),\tx(t))=\Phi(\tv(t))\in C([0,T];Y\times Y)
\cap C^1((0,T);Y\times Y)\,.$$
\end{remark}
We use the following lemma to decompose initial data around the line soliton.
\begin{lemma}
  \label{lem:decomp2} Let $a\in(0,2)$.
There exist positive constants $\delta_2$, $\delta_3$ and $L_0'$ such that
if $\|u\|_{X_1}<\delta_2$ and $L\ge L_0'$, then there exists a unique
$(\tc,\gamma)$ with $c=2+\tc$ satisfying
$\|\tc\|_{Y_1}+\|\gamma\|_{Y_1}<\delta_3$ and \eqref{eq:imp2}.
\end{lemma}
Lemma~\ref{lem:decomp2} can be proved in exactly the same as Lemma
\ref{lem:decomp}.
\par
We provide a continuation principle that ensures the existence of
\eqref{eq:decomp} as long as $\|v(t)\|_X$ and $\|\tc(t)\|_Y$ remain small.
\begin{proposition}
\label{prop:continuation}
Let $a$, $\delta_0$ and $L$ be the same as in Lemma~\ref{lem:decomp}
and let $u(t)$ be a solution of \eqref{KPII_integrated} such that
$u(t,x,y)-\varphi(x-4t)\in C([0,\infty);X\cap L^2(\R^2))$.
Then there exists a constant $\delta_2>0$
such that if \eqref{eq:decomp} and \eqref{eq:orth} hold for $t\in[0,T)$
and $v(t,z,y)$, $\tc(t,y):=c(t,y)-2$ and $\tx(t,y):=x(t,y)-4t$ satisfy
\begin{gather}
\label{eq:C1cx}
(\tc,\tx)\in C([0,T);Y\times Y)\cap C^1((0,T);Y\times Y)\,,\\
\label{eq:bd-vc}
\sup_{t\in[0,T)}\|v(t)\|_X\le \frac{\delta_0}{2}\,,
\quad \sup_{t\in[0,T)}\|\tc(t)\|_Y<\delta_2\,,\quad
\sup_{t\in[0,T)}\|\tx(t)\|_Y<\infty\,,
\end{gather}
then either $T=\infty$ or $T$ is not the maximal time of
the decomposition \eqref{eq:decomp} satisfying \eqref{eq:orth},
\eqref{eq:C1cx} and \eqref{eq:bd-vc}.
\end{proposition}
\begin{proof}
Suppose $T<\infty$.
Let $\tau\in(0,T)$, $\bar{x}(t,y)=x(t,y)$ for $t\in[0,T-\tau]$
and $\bar{x}(t,y)=x(T-\tau,y)+4(t+\tau-T)$ for $t\ge T-\tau$.
Let $u_1(t,x,y)=u(t,x+\bar{x}(t,y),y)-\varphi(x)$.
Then 
\begin{align*}
\sup_{t\in[0,T-\tau]}\|u_1(t)\|_X \le &   \sup_{t\in[0,T)}
(\|v(t)\|_X+\|\varphi_{c(t,y)}-\varphi\|_X+\|\tpsi_{c(t,y)}\|_X)
\\ \le & \frac{\delta_0}{2}+C_1\sup_{t\in[0,T)}\|\tc(t)\|_Y
\le \frac{\delta_0}{2}+C_1\delta_2\,,
\end{align*}
where $\tpsi_c(x)=\psi_{c,L}(x+4t)$ and 
$C_1$ is a constant that does not depend on $\tau$.
Since $Y\subset L^\infty(\R)$, it follows from the assumption that
$C_2:=\sup_{\tau\in(0,T)}\sup_ye^{-a\tx(\tau,y)}<\infty$. 
Thus for $t\in [T-\tau, T)$,
\begin{align*}
\|u_1(t)\|_X\le & \|u_1(T-\tau)\|_X
+\left\|e^{-a\tx(T-\tau,y)}\{\tv(t)-\tv(T-\tau)\}\right\|_X
\\ \le & \frac{\delta_0}{2}+C_1\delta_2
+C_2\|\tv(t)-\tv(T-\tau)\|_X\,.
\end{align*}
Now we choose $\delta_2$ and $\tau$ so that 
$$\delta_2<\min\left\{\delta_1,\delta_0/(4C_1)\right\}\,,\quad
\sup_{t_1\,,\,t_2\in[T-\tau,T+\tau]}
\|\tv(t_2)-\tv(t_1)\|_X<\delta_0/(4C_2)\,.$$
Then we have $\sup_{t\in[0,T+\tau]}\|u_1(t)\|_X<\delta_0$ and it follow from
Lemma~\ref{lem:decomp} and Remark~\ref{rem:decomp} that there exists
a unique 
$$(\tc_1(t),\tx_1(t))\in C([T-\tau,T+\tau];Y\times Y)\cap
C^1((T-\tau,T+\tau);Y\times Y)$$
satisfying $\sup_{t\in(T-\tau,T+\tau)}(\|\tc_1(t)\|_Y+\|\tx_1(t)\|_Y)<\delta_1$ and
\begin{gather}
\label{eq:decompu1}
u(t,x+\bar{x}(t,y),y)=\varphi_{c_1(t,y)}(z_1)-\tpsi_{c_1(t,y)}(z_1)+v_1(t,z_1,y)\,,\\
\label{eq:orthv1}
\lim_{M\to\infty}\int_{-M}^M\int_\R v_1(t,z_1,y)
\overline{g_k^*(z_1,\eta,c_1(t,y))}e^{-iy\eta}\,dz_1dy=0
\quad\text{in $L^2(-\eta_0,\eta_0)$}
\end{gather}
for $k=1$ and $2$, where $c_1(t,y)=2+\tc_1(t,y)$ and $z_1=x-x_1(t,y)$.
By the local uniqueness of the decomposition, we have for $t\in[T-\tau,T)$,
\begin{equation}
  \label{eq:tctx-1}
\tc(t)=\tc_1(t)\,,\quad\tx(t)=\tx(T-\tau)+\tx_1(t)\,,\quad
v(t,x,y)=v_1(t,x,y)\,.
\end{equation}
Let us define $\tc(t)$ and $\tx(t)$ by \eqref{eq:tctx-1} 
and $v(t)$ by \eqref{eq:decomp} for $t\in[T,T+\tau]$.
Then $(\tc,\tx)\in C([0,T+\tau];Y\times Y)\cap
C^1((0,T+\tau),Y\times Y)$ and  \eqref{eq:decompu1} and \eqref{eq:orthv1}
imply that $v(t)$ satisfy \eqref{eq:orth} for $t\in[0,T+\tau]$.
Thus we prove that $T$ is not maximal.
This completes the proof of Proposition~\ref{prop:continuation}.
\end{proof}
\bigskip

\section{Modulation equations}
\label{sec:modulation}
In this section, we will derive a system of PDEs which describe the motion
of  $c(t,y)$ and $x(t,y)$.
Substituting the ansatz \eqref{eq:decomp} into \eqref{KPII_integrated},
we obtain
\begin{equation}
  \label{eq:v}
\pd_tv=\mL_{c(t,y)}v+\ell+\pd_z(N_1+N_2)+N_3\,,  
\end{equation}
where  
$\mL_cv=-\pd_z(\pd_z^2-2c+6\varphi_c)v-3\pd_z^{-1}\pd_y^2$,
$\ell=\ell_1+\ell_2$, $\ell_k=\ell_{k1}+\ell_{k2}+\ell_{k3}$
$(k=1, 2)$ and
\begin{align*}
\ell_{11}=&(x_t-2c-3(x_y)^2)\varphi'_c-(c_t-6c_yx_y)\pd_c\varphi_c\,,\quad
\ell_{12}=3x_{yy}\varphi_c\,,\\
\ell_{13}=& 3c_{yy}\int_z^\infty\pd_c\varphi_c(z_1)dz_1
+3(c_y)^2\int_z^\infty \pd_c^2\varphi_{c}(z_1)dz_1\,,\\
\ell_{21}=& (c_t-6c_yx_y)\pd_c\tpsi_c-(x_t-4-3(x_y)^2)\tpsi_c'\,,\\
\ell_{22}=&\pd_z^3\tpsi_c-3\pd_z(\tpsi_c^2)
+6\pd_z(\varphi_c\tpsi_c)-3x_{yy}\tpsi_c\,,\\
\ell_{23}=&-3c_{yy}\int_z^\infty\pd_c\tpsi_c(z_1)dz_1
-3(c_y)^2\int_z^\infty \pd_c^2\tpsi_c(z_1)dz_1\,,
\end{align*}
\begin{align*}
& N_1=-3v^2\,,\quad N_2=\{x_t-2c-3(x_y)^2\}v+6\tpsi_cv\,,\\
& N_3=6x_y\pd_yv+3x_{yy}v=6\pd_y(x_yv)-3x_{yy}v\,.
\end{align*}
Here we abbreviate $c(t,y)$ as $c$ and $x(t,y)$ as $x$.
\par

First, we will derive modulation equations of $c(t,y)$ and $x(t,y)$
from the orthogonality condition \eqref{eq:orth}
assuming that $v_0\in X\cap H^3(\R^2)$ and $\pd_x^{-1}v_0\in H^2(\R^2)$.
If $v_0\in H^3(\R^2)$ and $\pd_x^{-1}v_0\in H^2(\R^2)$, then
it follows from \cite{MST} that $\tv(t)\in H^3(\R^2)\in C(\R;H^3(\R^2))$
and $\pd_x^{-1}\tv(t)\in C(\R;H^2(\R^2))$. Moreover, Proposition~\ref{prop:LWPX}
implies that $\tv(t)\in C([0,\infty);X)$.
If $\bM_1(T)$ and $\bM_2(T)$ are sufficiently small,
then the decomposition \eqref{eq:decomp} satisfying \eqref{eq:orth}
and \eqref{eq:C1cx} exists for $t\in[0,T]$
by Lemma~\ref{lem:decomp}, Remark~\ref{rem:decomp}
and Proposition~\ref{prop:continuation}.
Since $Y\subset H^4(\R)$,
$$v(t,z,y)-\tv(t,z+\tx(t,y),y)=\varphi(z+x(t,y))
-\varphi_{c(t,y)}(z)+\tpsi_{c(t,y)}(z)\in H^3(\R^2)\,,$$
and we easily see that $v(t)\in C([0,T];X\cap H^3(\R^2))$.
Using
$$\int_\R(v(t,z,y)-\tv(t,z+\tx(t,y),y))\,dz=0\,,$$
by \eqref{eq:0mean} and its integrand decays exponentially as $z\to\pm\infty$,
we have 
$$(\pd_z^{-1}v)(t,z,y)=-\int_z^\infty v(t,z_1,y)\,dz_1\in X\cap H^2(\R^2)\,.$$
\par
By Proposition~\ref{prop:continuation} and Remark~\ref{rem:decomp},
the mapping
$$ t\mapsto \int_{\R^2}v(t,z,y)\overline{g_k^*(z,\eta,c(t,y))}
e^{-iy\eta}dzdy\in Z$$
is $C^1$ for $t\in[0,T]$ if we have \eqref{eq:C1cx} and \eqref{eq:bd-vc}.
Differentiating \eqref{eq:orth} with respect to $t$ and substituting
\eqref{eq:v} into the resulting equation, we have in $L^2(-\eta_0,\eta_0)$
\begin{equation}
\label{eq:orth_t}
\begin{split}
& \frac{d}{dt}\int_{\R^2}v(t,z,y)\overline{g_k^*(z,\eta,c(t,y))}
e^{-iy\eta}dzdy
\\=& \int_{\R^2} \ell\overline{g_k^*(z,\eta,c(t,y))}e^{-iy\eta}dzdy
+\sum_{j=1}^5II^j_k(t,\eta)=0\,,
\end{split}  
\end{equation}
where
\begin{align*}
II^1_k=& \int_{\R^2} v(t,z,y)\mL_{c(t,y)}^*(\overline{g_k^*(t,z,c(t,y))
e^{iy\eta}})\, dzdy\,,\\
II^2_k=& -\int_{\R^2} N_1\overline{\pd_zg_k^*(z,\eta,c(t,y))}
e^{-iy\eta}\, dzdy\,,\\
II^3_k=& \int_{\R^2} N_3\overline{g_k^*(z,\eta,c(t,y))}e^{-iy\eta}dzdy
\\ & +6\int_{\R^2} v(t,z,y)c_y(t,y)x_y(t,y)
\overline{\pd_cg_k^*(z,\eta,c(t,y))}e^{-iy\eta}\, dzdy\,,\\
II^4_k=& \int_{\R^2} v(t,z,y)\left(c_t-6c_yx_y\right)(t,y)
\overline{\pd_cg_k^*(z,\eta,c(t,y))}e^{-iy\eta}\, dzdy\,,\\
II^5_k=& -\int_{\R^2} N_2\overline{\pd_zg_k^*(z,\eta,c(t,y))}e^{-iy\eta}\, dzdy\,.
\end{align*}
\par

Next, we will show the second equation of \eqref{eq:orth_t} 
for $t\in[0,T]$ and 
$v(t)\in C([0,T]; L^2(\R^2))\cap L^\infty([0,T];X)$ assuming that
$\bM_1(T)$ and $\bM_2(T)$ are sufficiently small.
Let $\{v_{0n}\}_{n=1}^\infty$ be a sequence such that
$$v_{0n}\in H^3(\R^2)\cap X\,, \quad\pd_x^{-1}v_{0n}\in H^2(\R^2)\,,\quad
\lim_{n\to\infty}(\|v_{0n}-v_0\|_X+\|v_{0n}-v_0\|_{L^2(\R^2)})=0\,,$$
and let $u_n(t)$ be a solution of \eqref{KPII_integrated} satisfying
$u_n(0,x,y)=\varphi(x)+v_{0n}(x,y)$ and $\tv_n(t,x,y)=u_n(t,x,y)-\varphi(x)$.
Since $\sup_{t\in[0,T]}\|\tv_n(t)-\tv(t)\|_{L^2(\R^2)}\to0$ as $n\to\infty$
by \cite{MST} 
and $\sup_n\sup_{t\in[0,T]}\|\tv_n(t)\|_X<\infty$ by \eqref{eq:virial-fx} in
Appendix~\ref{ap:LWP}, we have 
$\lim_{n\to\infty}\sup_{t\in[0,T]}\|\tv_n(t)-\tv(t)\|_{L^2(\R^2;e^{az}dz)}=0$
for any $t\ge0$.  If $\eta_0$ is so small that $a/2>\nu_0$, we can
replace the weight function $e^{2az}$  by $e^{az}$ in Lemma~\ref{lem:decomp},
Remark~\ref{rem:decomp} and Proposition~\ref{prop:continuation} and see that
there exist $v_n(t)$, $c_n(t)$ and $x_n(t)$ satisfying for $t\in[0,T]$, 
\begin{gather*}
u_n(t,x,y)=\varphi_{c_n(t,y)}(z)-\psi_{c_n(t,y),L}(z+4t)+v_n(t,z,y)\,,\quad
z=x-x_n(t,y)\,.\\
\lim_{M\to\infty}\int_{-M}^M \int_\R
v_n(t,z,y)\overline{g_k^*(z,\eta,c_n(t,y))}e^{-iy\eta}\,dzdy=0
\quad\text{in $L^2(-\eta_0,\eta_0)$,}\\
\lim_{n\to\infty}\sup_{t\in[0,T]}\|v_n(t)-v(t)\|_{L^2(\R^2;e^{ax}dzdy)}=0\,,\\
\lim_{n\to\infty}\left\|\left(c_n(t)-c(t),x_n(t)-x(t)\right)
\right\|_{C^1([0,T];Y\times Y)}=0\,.
\end{gather*}
Thus we can obtain the second equation of \eqref{eq:orth_t} on
$[0,T]$ for $v_0\in X\cap L^2(\R^2)$ by passing to the limit $n\to\infty$.
\par

The modulation PDEs of $c(t,y)$ and $x(t,y)$ can be obtained by
computing the inverse Fourier transform of \eqref{eq:orth_t}
in the $\eta$-variable. The leading term of
$$\int_{\R^2}\ell_1\overline{g_k^*(z,\eta,c(t,y))}e^{-iy\eta}\,dzdy$$
is $\sqrt{2\pi}\mF_yG_k(t,\eta)$, where 
\begin{equation}
  \label{eq:Gk}
G_k(t,y)=\int_\R \ell_1 \overline{g_k^*(z,0,c(t,y))}dz\,.  
\end{equation}
Using the asymptotic formula of $g_k^*(z,\eta)$, we can see
that $G_k$ has the following expression.
\begin{lemma}
  \label{lem:G}
Let $\mu_1=\frac{1}{2}-\frac{\pi^2}{12}$ and
$\mu_2=\frac{\pi^2}{32}-\frac{3}{16}$. Then
  \begin{align*}
G_1=&16x_{yy}\left(\frac{c}{2}\right)^{3/2}
-2(c_t-6c_yx_y)\left(\frac{c}{2}\right)^{1/2}+6c_{yy}-\frac{3}{c}(c_y)^2\,,\\
G_2=& -2(x_t-2c-3(x_y)^2)\left(\frac{c}{2}\right)^{2}
+6x_{yy}\left(\frac{c}{2}\right)^{3/2}
-\frac12(c_t-6c_yx_y)\left(\frac{c}{2}\right)^{1/2}
\\ & +\mu_1c_{yy}+\mu_2(c_y)^2\left(\frac{c}{2}\right)^{-1}\,.
  \end{align*}
\end{lemma}
The proof is given in Appendix~\ref{ap:G}.
We remark that $(G_1,G_2)$ are the dominant part of diffusion wave
equations for $c$ and $x$.
\par

Next, we will expand 
$$\int_\R \ell_1\left(\overline{g_k^*(z,\eta,c(t,y))-g_k^*(z,0,c(t,y))}\right)
e^{-iy\eta}dzdy$$ in $c(t,y)$ and $x(t,y)$ up to the second order.
In order to express the coefficients of $c_t$, $x_t$, $c_{yy}$ and $x_{yy}$,
let us introduce the operators $S_k^j$ ($j\,,\,k=1$, $2$).
For $q_c(z)=\varphi_c(z)$, $\varphi_c'(z)$, $\pd_c\varphi_c(z)$
and $\pd_z^{-1}\pd_c^m\varphi_c(z)=-\int_z^\infty \pd_c^m\varphi_c(z_1)dz_1$
($m\ge1$), let 
\begin{align*}
& S_k^1[q_c](f)(t,y)=\frac{1}{2\pi} \int_{-\eta_0}^{\eta_0}\int_{\R^2}
f(y_1)q_2(z)\overline{g_{k1}^*(z,\eta,2)}e^{i(y-y_1)\eta}dy_1dzd\eta\,,\\
& S_k^2[q_c](f)(t,y)=\frac{1}{2\pi} \int_{-\eta_0}^{\eta_0}\int_{\R^2}
f(y_1)\tc(t,y_1)\overline{g_{k2}^*(z,\eta,c(t,y_1))}
e^{i(y-y_1)\eta}dy_1dzd\eta\,,
\end{align*}
where 
\begin{align*}
& g_{k1}^*(z,\eta,c)=\frac{g_k^*(z,\eta,c)-g_k^*(z,0,c)}{\eta^2}\,,\quad
\delta q_c(z)=\frac{q_c(z)-q_2(z)}{c-2}\,,\\
& g_{k2}^*(z,\eta,c)=g_{k1}^*(z,\eta,2)\delta q_c(z)+
\frac{g_{k1}^*(z,\eta,c)-g_{k1}^*(z,\eta,2)}{c-2}q_c(z)\,.
\end{align*}
We have $S^1_k\in B(Y)$ and $S^1_k$ are independent of $c(t,y)$
whereas  $\|S^2_k\|_{B(Y,Y_1)}\lesssim \|\tc\|_Y$.
See Claims~\ref{cl:S1} and \ref{cl:S2} in Appendix~\ref{ap:opS}.
Using $S^j_k$ ($j$, $k=1$, $2$), we have
\begin{equation}
  \label{eq:der2}
\begin{split}
& \frac{1}{\sqrt{2\pi}}\wP_1\mF^{-1}_\eta\left(
\int_\R \ell_1\left(\overline{g_k^*(z,\eta,c(t,y))-g_k^*(z,0,c(t,y))}\right)
e^{-iy\eta}dzdy\right)
\\=& -\sum_{j=1,2}\pd_y^2\left(S^j_k[\varphi_c'](x_t-2c-3(x_y)^2)
 -S^j_k[\pd_c\varphi_c](c_t-6c_yx_y)\right)-\pd_y^2(R^1_k+R^2_k)\,,
\end{split}  
\end{equation}
\begin{align*}
& R^1_k=3S^1_k[\varphi_c](x_{yy})-3S^1_k[\pd_z^{-1}\pd_c\varphi_c](c_{yy})\,,
\\ & 
R^2_k=3S^2_k[\varphi_c](x_{yy})-3S^2_k[\pd_z^{-1}\pd_c\varphi_c](c_{yy})
-3\sum_{j=1,2}S^j_k[\pd_z^{-1}\pd_c^2\varphi_c](c_y^2)\,.
\end{align*}
We rewrite the linear term $R^1_k$ as
$$\begin{pmatrix}R^1_1 \\ R^1_2\end{pmatrix}
=\wS_0\begin{pmatrix}c_{yy} \\ x_{yy} \end{pmatrix}\,,\quad
\wS_0=3\begin{pmatrix}-S^1_1[\pd_z^{-1}\pd_c\varphi_c] & S^1_1[\varphi_c]
\\ -S^1_2[\pd_z^{-1}\pd_c\varphi_c] & S^1_2[\varphi_c]\end{pmatrix}\,.$$
\par

Next we deal with
$\int_\R \ell_2\overline{g_k^*(z,\eta,c(t,y))}e^{-iy\eta}dzdy$.
Let $S^3_k[p]$ and $S^4_k[p]$ be operators defined by
\begin{equation*}
S^3_k[p](f)(t,y)=\frac{1}{2\pi}\int_{-\eta_0}^{\eta_0}\int_{\R^2}
f(y_1)p(z+4t+L)\overline{g_k^*(z,\eta)}e^{i(y-y_1)\eta}dy_1dzd\eta\,,  
\end{equation*}
\begin{multline*}
S^4_k[p](f)(t,y)=\frac{1}{2\pi}\int_{-\eta_0}^{\eta_0}\int_{\R^2}
f(y_1)\tc(t,y_1)p(z+4t+L)\\ \times\overline{g_{k3}^*(z,\eta,c(t,y_1))}
e^{i(y-y_1)\eta}dy_1dzd\eta\,,
\end{multline*}
where
$$g_{k3}^*(z,\eta,c)=\frac{g_k^*(z,\eta,c)-g_k^*(z,\eta)}{c-2}\,.$$
By the definition of $\tpsi_c$,
\begin{equation}
  \label{eq:l21F-1}
  \begin{split}
& \frac{\mathbf{1}_{[-\eta_0,\eta_0]}(\eta)}{\sqrt{2\pi}}
\int_\R \ell_{21}\overline{g_k^*(z,\eta,c(t,y))}e^{-iy\eta}dzdy
\\=&\mF_y\left\{(S^3_k[\psi]+S^4_k[\psi])
\bigl(\sqrt{2/c}(c_t-6c_yx_y)\bigr)\right\}(t,\eta)
\\ &-2\sqrt{2}\mF_y\left\{(S^3_k[\psi']+S^4_k[\psi'])
\bigl((\sqrt{c}-\sqrt{2})(x_t-4-3(x_y)^2\bigr)\right\}(t,\eta)\,.    
  \end{split}
\end{equation}
The operator norms of $S^j_k[\psi]$, $S^j_k[\psi']$
($j=3\,,\,4$, $k=1\,,\,2$) decay exponentially because $g_k^*(z,\eta)$
and $g_k^*(z,\eta,c)$ are exponentially localized as $z\to-\infty$
and  $\psi\in C_0^\infty(\R)$. See Claims~\ref{cl:S3} and \ref{cl:S4}
in Appendix~\ref{ap:opS}.

\par

Next we decompose
$(2\pi)^{-1}\int_{\R^2}(\ell_{22}+\ell_{23})
\overline{g_k^*(z,\eta,c(t,y))}e^{-iy\eta}dzdy$ into a linear part
and a nonlinear part with respect to $\tc$ and $\tx$. 
The linear part can be written as
\begin{equation}
  \label{eq:der4}
\frac{1}{2\pi}\int_{-\eta_0}^{\eta_0}\int_{\R^2}
\ell_{2,lin}(t,z,y_1)\overline{g_k^*(z,\eta)}e^{i(y-y_1)\eta}dy_1dzd\eta
=:\ta_k(t,D_y)\tc\,,  
\end{equation}
where 
$$\ell_{2,lin}(t,z,y)=\tc(t,y)\pd_z\left\{\pd_z^2+6\varphi(z)\right\}
\psi(z+4t+L)-3c_{yy}(t,y)\int_z^\infty\psi(z_1+4t+L)dz_1\,,$$
\begin{equation}
\label{eq:def-ak}
\begin{split}
\ta_k(t,\eta)=& \Bigl[
\int_\R\left\{\psi'''(z+4t+L)+6(\varphi(z)\psi(z+4t+L))_z\right\}
\overline{g_k^*(z,\eta)}dz\\
& +3\eta^2\int_\R\left(\int_z^\infty\psi(z_1+4t+L)dz_1\right)
\overline{g_k^*(z,\eta)}dz\Bigr]\mathbf{1}_{[-\eta_0,\eta_0]}(\eta)\,,
\end{split}
\end{equation}
and the nonlinear part is 
\begin{equation}
  \label{eq:der5}
\begin{split}
R^3_k(t,y):=& \frac{1}{2\pi}\int_{-\eta_0}^{\eta_0}
\int_\R (\ell_{22}+\ell_{23})\overline{g_k^*(z,\eta,c(t,y_1))}
e^{i(y-y_1)\eta}dzdy_1d\eta \\
& -\frac{1}{2\pi}\int_{-\eta_0}^{\eta_0}
\int_\R \ell_{2,lin}\overline{g_k^*(z,\eta)}e^{i(y-y_1)\eta}dzdy_1d\eta,.
\end{split}
\end{equation}
\par

Next,  we deal with  $II^j_k$ $(j=1\,,\,\cdots\,,\,5)$ in \eqref{eq:orth_t}.
Let 
\begin{align*}
II^3_{k1}=& -3\int_{\R^2}v(t,z,y)x_{yy}(t,y)
\overline{g_k^*(z,\eta,c(t,y))}e^{-iy\eta}\,dzdy\,,\\
II^3_{k2}=& 6\int_{\R^2}v(t,z,y)x_y(t,y)
\overline{g_k^*(z,\eta,c(t,y))}e^{-iy\eta}\,dzdy
\end{align*}
so that $II^3_k=II^3_{k1}+i\eta II^3_{k2}$. Let
\begin{equation}
\label{eq:der6}
\begin{split}
R^4_k(t,y)=& \frac{1}{2\pi}\int_{-\eta_0}^{\eta_0}
\left\{II^1_k(t,\eta)+II^2_k(t,\eta)+ II^3_{k1}(t,\eta)\right\}
e^{iy\eta}d\eta\,,\\
R^5_k(t,y)=&  \frac{1}{2\pi}\int_{-\eta_0}^{\eta_0}II^3_{k2}(t,\eta)e^{iy\eta}d\eta\,.
\end{split}
\end{equation}

Let $S^5_k$ and $S^6_k$ be operators defined by
\begin{gather*}
S^5_kf(t,y)=\frac{1}{2\pi}\int_{-\eta_0}^{\eta_0}
\int_{\R^2} v(t,z,y_1)f(y_1)\overline{\pd_cg_k^*(z,\eta,c(t,y_1))}
e^{i(y-y_1)\eta}\,dzdy_1d\eta\,,\\
S^6_kf(t,y)=-\frac{1}{2\pi}\int_{-\eta_0}^{\eta_0}\int_{\R^2}
v(t,z,y_1)f(y_1)\overline{\pd_zg_k^*(z,\eta,c(t,y_1))}
e^{i(y-y_1)\eta}\,dzdy_1d\eta
\end{gather*}
so that
\begin{equation}
  \label{eq:der7}
\begin{split}
& \mathbf{1}_{[-\eta_0,\eta_0]}(\eta)II^4_k(t,\eta)=
\sqrt{2\pi}\mF_y(S^5_k(c_t-6c_yx_y))\,,\\
& \mathbf{1}_{[-\eta_0,\eta_0]}(\eta)II^5_k(t,\eta)
=\sqrt{2\pi}\mF_y\left\{S^6_k\left(x_t-2c-3(x_y)^2\right)+R^6_k\right\}\,,
\end{split}  
\end{equation}
where
$$R^6_k=-\frac{3}{\pi}\int_{-\eta_0}^{\eta_0}\int_{\R^2}\psi_{c(t,y_1),L}(z+4t)
v(t,z,y_1)\overline{\pd_zg_k^*(z,\eta,c(t,y_1))}e^{i(y-y_1)\eta}dy_1dzd\eta\,.$$
\par

Now we are in position to translate \eqref{eq:orth_t} into
a PDE form.
Using \eqref{eq:Gk}--\eqref{eq:l21F-1} and \eqref{eq:der4}--\eqref{eq:der7},
we can translate \eqref{eq:orth_t} into
\begin{equation}
  \label{eq:modpre}
  \begin{split}
&\wP_1\begin{pmatrix}  G_1 \\ G_2\end{pmatrix}
-\left(\pd_y^2(\wS_1+\wS_2)-\wS_3-\wS_4-\wS_5\right)
\begin{pmatrix} c_t-6c_yx_y \\ x_t-2c-3(x_y)^2\end{pmatrix}
\\ &
+\widetilde{\mathcal{A}}_1(t)
\begin{pmatrix}  \tc \\ x\end{pmatrix}
-\pd_y^2R_1+\wR^1+\pd_y\wR^2=0\,,    
  \end{split}
\end{equation}
where
\begin{align*}
& \wS_j=\begin{pmatrix}-S^j_1[\pd_c\varphi_c]  & S^j_1[\varphi_c']
\\ -S^j_2[\pd_c\varphi_c]  & S^j_2[\varphi_c']\end{pmatrix}
\enskip\text{for $j=1$, $2$,}\quad
\wS_3=\begin{pmatrix}
S^3_1[\psi] & 0
\\ S^3_2[\psi] & 0\end{pmatrix}\,,
\\ &
\wS_4=\begin{pmatrix}
S^3_1[\psi]((\sqrt{2/c}-1)\cdot)+S^4_1[\psi](\sqrt{2/c}\cdot)
& -2(S^3_1[\psi']+S^4_1[\psi'])((\sqrt{2c}-2)\cdot) 
\\ S^4_1[\psi]((\sqrt{2/c}-1)\cdot)+S^4_2[\psi](\sqrt{2/c}\cdot)
& -2(S^3_2[\psi']+S^4_2[\psi'])((\sqrt{2c}-2)\cdot)
\end{pmatrix}\,,\\
& \wS_5=\begin{pmatrix}  S^5_1 & S^6_1 \\ S^5_2 & S^6_2 \end{pmatrix}\,,
\quad
\widetilde{\mathcal{A}}_1(t)=
\begin{pmatrix} \ta_1(t,D_y) & 0  \\ \ta_2(t,D_y) & 0\end{pmatrix}\,,
\\ &\wR^1=R^3+R^4+R^6-\wS_4\begin{pmatrix} 0 \\ 2\tc\end{pmatrix}\,,\quad
\wR^2=R^5-\pd_yR^2\,,
\end{align*}
and  $R^j={}^t\!(R^j_1\,,\,R^j_2)$ for $1\le j\le 6$.
In $G_1$, the nonlinear terms $6(c/2)^{1/2}c_yx_y$ and $16x_{yy}\{((c/2)^{3/2}-1\}$
are critical because they are expected to decay like $t^{-1}$
as $t\to\infty$. To translate these nonlinearity into 
a divergence form, we will make use of the following change of variables. Let
\begin{align}
\label{eq:bdef}
& \tilde{x}(t,\cdot)=x(t,\cdot)-4t\,,\,\quad
b(t,\cdot)=\frac{1}{3}\wP_1\left\{\sqrt{2}c(t,\cdot)^{3/2}-4\right\}\,,\\
& \cC_1=\frac12\wP_1\left\{c(t,\cdot)^2-4\right\}\wP_1\,,\notag\\
& \wC_1=\begin{pmatrix}0 & 0 \\ 0 & \cC_1\end{pmatrix}\,,\quad
B_1=\begin{pmatrix}  2 & 0 \\ \frac12 & 2\end{pmatrix}\,,\quad
B_2=\begin{pmatrix}  6 & 16 \\ \mu_1 & 6\end{pmatrix}\,. \notag
\end{align}
We remark that $b\simeq \tilde{c}=c-2$ if $c$ is close to $2$
(see Claim~\ref{cl:b-capprox}).
By \eqref{eq:bdef}, we have
$b_t=\wP_1(c/2)^{1/2}c_t$, $b_y=\wP_1(c/2)^{1/2}c_y$ and it follows from
Lemma~\ref{lem:G} that 
\begin{equation}
  \label{eq:G-T}
\begin{split}
\wP_1\begin{pmatrix}G_1 \\ G_2 \end{pmatrix}=& 
-(B_1+\wC_1)\wP_1\begin{pmatrix} b_t-6(bx_y)_y \\ x_t-2c-3(x_y)^2 \end{pmatrix}
+B_2\begin{pmatrix}b_{yy} \\ x_{yy} \end{pmatrix} +\wP_1R^7\,,
\end{split}  
\end{equation}
where $R^7={}^t(R^7_1,R^7_2)$ and
\begin{equation}
  \label{eq:R6-def}
\begin{split}
R^7_1=& \left\{4\sqrt{2}c^{3/2}-16-12b\right\}x_{yy}-6(b_{yy}-c_{yy})
\\ & -6(2b_y-(2c)^{1/2}c_y)x_y-3c^{-1}(c_y)^2\,,\\
R^7_2=& 6\left\{\left(\frac{c}{2}\right)^{3/2}-1\right\}x_{yy}
+3\left(\frac{c}{2}\right)^{1/2}c_yx_y-3(bx_y)_y
\\ & -\mu_1(b_{yy}-c_{yy})+\mu_2\frac{2}{c}(c_y)^2\,.
\end{split}  
\end{equation}
Let 
$\cC_2=\wP_1\left\{\left(\frac{c(t,\cdot)}{2}\right)^{1/2}-1\right\}\wP_1$,
$\wC_2=\begin{pmatrix}\cC_2 & 0 \\ 0 & 0 \end{pmatrix}$,
$\bS_j=\wS_j(I+\wC_2)^{-1}$ for $1\le j\le 5$ and
\begin{equation}
  \label{eq:def-B3}
B_3=B_1+\wC_1+\pd_y^2(\bS_1+\bS_2)-\bS_3-\bS_4-\bS_5\,.  
\end{equation}
Note that $I+\wC_2$ is invertible as long as $c(t,\cdot)$ remains small in $Y$
and that $B_3$ is a bounded operator on $Y\times Y$ depending on $\tc$ and $v$.
Substituting \eqref{eq:G-T} into \eqref{eq:modpre}, we have
\begin{align*}
& B_3\wP_1\begin{pmatrix}b_t-6(bx_y)_y \\ x_t-2c-3(x_y)^2 \end{pmatrix}
\\ =& \left\{(B_2-\pd_y^2\wS_0)\pd_y^2+\widetilde{\mathcal{A}}_1(t)\right\}
\begin{pmatrix}b \\ x \end{pmatrix}
+\wP_1R^7+\wR^1+\pd_y\wR^2+\wR^3\,,
\end{align*}
where $\wR^3=R^8+R^9+R^{10}$ and
\begin{align*}
& R^8=6(B_3-B_1-\wC_1)
\begin{pmatrix}(I+\mathcal{C}_2)(c_yx_y)-(bx_y)_y \\ 0\end{pmatrix}\,,\\
& R^9=\pd_y^2\wS_0\begin{pmatrix} b_{yy}-c_{yy}\\ 0\end{pmatrix}\,,\quad
R^{10}=\widetilde{\mathcal{A}}_1(t)\begin{pmatrix}\tc -b \\ 0 \end{pmatrix}\,.
\end{align*}
Let
\begin{align*}
& \bM_1(T)=\sup_{0\le t\le T}\bigl\{\sum_{k=0}^1 (1+t)^{(2k+1)/4}
(\|\pd_y^k\tc(t,\cdot)\|_{L^2}+\|\pd_y^{k+1}x(t,\cdot)\|_{L^2})
\\ & \qquad +(1+t)(\|\pd_y^2\tc(t,\cdot)\|_{L^2}+\|\pd_y^3x(t,\cdot)\|_{L^2})
\bigr\}\,,
\\ & \bM_2(T)=(1+t)^{3/4}\|v(t,\cdot)\|_X\,,\quad
\bM_3(T)=\sup_{0\le t\le T}\|v(t,\cdot)\|_{L^2(\R^2)}\,.
\end{align*}
Then we have the following.
\begin{proposition}
\label{prop:modulation}
There exists a $\delta_3>0$ such that if 
$\bM_1(T)+\bM_2(T)+\eta_0+e^{-aL}<\delta_3$ for a $T\ge0$, then  
\begin{equation}
  \label{eq:modeq}
\begin{pmatrix}b_t \\ \tx_t\end{pmatrix}=\mathcal{A}(t)
\begin{pmatrix}b \\ \tx\end{pmatrix}\,+\sum_{i=1}^8 \cN_i
\end{equation}
where $B_4=B_1+\pd_y^2\wS_1-\wS_3=B_3|_{\tc=0\,,\,v=0}$,
\begin{align*}
& \mathcal{A}(t)=B_4^{-1}
\left\{(B_2-\pd_y^2\wS_0)\pd_y^2+\widetilde{\mathcal{A}}_1(t)\right\}
+\begin{pmatrix}  0 & 0 \\ 2 & 0\end{pmatrix}\,,\\
\cN_1=& \wP_1\begin{pmatrix}
6(b\tx_y)_y\\ 2(\tilde{c}-b)+3(\tx_y)^2
\end{pmatrix}\,,\quad \cN_2=B_3^{-1}\wP_1R^7\,,\\
\cN_3=& B_3^{-1}\wR^1\,,\quad \cN_4=B_3^{-1}\pd_y\wR^2\,,
\quad \cN_5=B_3^{-1}\wR^3\,,\\
\cN_6=& (B_3^{-1}-B_4^{-1})\widetilde{\mathcal{A}}_1(t)
\begin{pmatrix}b \\ x\end{pmatrix}\,,\quad
\cN_7=(B_3^{-1}-B_4^{-1})(B_2-\pd_y^2\wS_0)
\begin{pmatrix}b_{yy} \\ 0 \end{pmatrix}\,,\\ 
\cN_8=& (B_3^{-1}-B_4^{-1})(B_2-\pd_y^2\wS_0)
\begin{pmatrix} 0 \\ x_{yy} \end{pmatrix}
\,.\end{align*}
\end{proposition}
\begin{proof}
  Proposition~\ref{prop:continuation} implies that the
  \eqref{eq:decomp} persists on $[0,T]$ if $\delta_3$ is sufficiently
  small. Moreover Claims~\ref{cl:(1+wT2)^{-1}}--\ref{cl:invB_4} below
  imply that $B_3$, $B_4$ and $I+\wC_k$ are invertible if
  $\|\tc(t)\|_Y$, $\|v(t)\|_X$, $\eta_0$ and $e^{-aL}$ are
  sufficiently small.  Thus we have \eqref{eq:modeq}.
\end{proof}

\begin{claim}
  \label{cl:(1+wT2)^{-1}}
There exist positive constants $\delta$ and $C$ such that
if $\bM_1(T)\le \delta$, then for $s\in[0,T]$ and $k=1$, $2$,
\begin{align}
  \label{eq:T-bound1}
& \|\wC_k\|_{B(Y)}\le  C\bM_1(T)\la s\ra^{-1/2}\,,\\
  \label{eq:T-bound}
& \|\wC_k\|_{B(Y,Y_1)}\le C\bM_1(T)\la s\ra^{-1/4}\,,\\
\notag
& \|(I+\wC_k)^{-1}\|_{B(Y)}+\|(I+\wC_k)^{-1}\|_{B(Y_1)}\le C\,.
\end{align}
\end{claim}
Claim~\ref{cl:(1+wT2)^{-1}} immediately follows from Claim~\ref{cl:T1-T2}
in Appendix~\ref{ap:opS} and the fact that $Y_1\subset Y$ and
$\|\wC_k\|_{B(Y_1)}\lesssim \|\cC_k\|_{B(Y,Y_1)}$.  

\begin{claim}
  \label{cl:invB_3}
There exist positive constants $\delta$ and $C$ such that
if $\bM_1(T)+\bM_2(T)+\eta_0^2+e^{-aL}\le \delta$, then 
$$\|B_3^{-1}\|_{B(Y)}\le C \quad\text{and}\quad
\|B_3^{-1}\|_{B(Y_1)}\le C\quad\text{for $s\in[0,T]$.}$$
\end{claim}

\begin{claim}
\label{cl:invB_4}
There exist positive constants $C$ and $\delta$ such that
if $\eta_0^2+e^{-aL}\le \delta$, then 
$\|B_4^{-1}\|_{B(Y)}+\|B_4^{-1}\|_{B(Y_1)}\le C$.
\end{claim}
The proof of Claims~\ref{cl:invB_3} and \ref{cl:invB_4} will be given
in Appendix~\ref{sec:[B3,pdy]}.
\bigskip

\section{\`A priori estimates for $c(t,y)$ and $x_y(t,y)$}
\label{sec:apriori}
In this section, we will estimate $\bM_1(T)$
assuming that $\bM_i(T)$ ($1\le i\le 3)$, $\eta_0$ and $e^{-aL}$ are 
sufficiently small.

\begin{lemma}
  \label{lem:M1-bound}
There exist positive constants $\delta_4$ and $C$ such that if
$\bM_1(T)+\bM_2(T)+\eta_0+e^{-aL}\le \delta_4$, then
\begin{equation}
  \label{eq:M1-bound}
\bM_1(T)\le C\|v_0\|_{X_1}+C(\bM_1(T)+\bM_2(T))^2\,.
\end{equation}
\end{lemma}
To prove Lemma~\ref{lem:M1-bound}, we need the following.
\begin{claim}
  \label{cl:[B3,pdy]}
There exist positive constants $\eta_1$, $\delta$ and $C$ such that
if $\eta_0\in(0,\eta_1]$ and $\bM_1(T)\le \delta$, then 
\begin{equation*}
\|[\pd_y,B_3]\|_{B(Y,Y_1)} \le C(\bM_1(T)+\bM_2(T))\la s\ra^{-3/4}
\quad\text{for $s\in[0,T]$.}  
\end{equation*}
\end{claim}
The proof is given in Appendix~\ref{sec:[B3,pdy]}.
\begin{claim}
  \label{cl:difS1S3}
There exist  positive constants $\eta_1$, $\delta$ and $C$ such that
if $\eta_0\in(0,\eta_1]$ and $\bM_1(T)\le\delta$, then for $t\in[0,T]$,
\begin{gather*}
\|\bS_1-\wS_1\|_{B(Y,Y_1)}\le C\bM_1(T)\la t\ra^{-1/4}\,,\\
\|\bS_3-\wS_3\|_{B(Y,Y_1)}\le C\bM_1(T)\la t\ra^{-1/4}e^{-a(4t+L)}\,.
\end{gather*}
\end{claim}
Claim~\ref{cl:difS1S3} follows immediately from \eqref{eq:wS1-new},
\eqref{eq:wS3-new} and Claim~\ref{cl:(1+wT2)^{-1}}.

\begin{proof}[Proof of Lemma~\ref{lem:M1-bound}]
To apply Lemma~\ref{lem:decay-BB},
we will transform \eqref{eq:modeq} into a system of $b$ and $x_y$.
Let $A(t)=\diag(1,\pd_y)\mathcal{A}(t)\diag(1,\pd_y^{-1})$,
$B_5=B_1+\pd_y^2\wS_1$ and
\begin{align*}
& A_0=\diag(1,\pd_y)\left\{B_5^{-1}(B_2-\pd_y^2\wS_1)\pd_y^2
+\begin{pmatrix}0 & 0 \\ 2 & 0\end{pmatrix}\right\}\diag(1,\pd_y)^{-1}\,,\\
& A_1(t)=\diag(1,\pd_y)(B_4^{-1}-B_5^{-1})(B_2-\pd_y^2\wS_0)
\diag(\pd_y^2,\pd_y)+\diag(1,\pd_y)B_4^{-1}\widetilde{\mathcal{A}}_1(t)\,,
\end{align*}
where $\pd_y^{-1}=\mF_\eta^{-1}(i\eta)^{-1}\mF_y$.
Then $A(t)=A_0+A_1(t)$. 
Note that $\widetilde{\mathcal{A}}_1(t)=\linebreak
\widetilde{\mathcal{A}}_1(t)\diag(1,\pd_y^{-1})$.
Multiplying \eqref{eq:modeq} by $\diag(1,\pd_y)$ from the left,
we can transform \eqref{eq:modeq} into
\begin{equation}
  \label{eq:modeq2}
\pd_t\begin{pmatrix}b \\ x_y\end{pmatrix}=A(t)
\begin{pmatrix}b \\ x_y\end{pmatrix}\,+\sum_{i=1}^8\diag(1,\pd_y)\cN_i\,.
\end{equation}
\par
Now we will show that $A(t)$ satisfies the hypothesis 
\eqref{eq:H} of Lemma~\ref{lem:decay-BB}. Let
$A_0(\eta)$ be the Fourier transform of the operator $A_0$.
Then 
\begin{equation}
  \label{eq:A0A*}
\begin{split}
A_0(\eta)=& \begin{pmatrix}1 & 0\\ 0 & i\eta\end{pmatrix}
(B_1^{-1}+O(\eta^2))(B_2+O(\eta^2))
\begin{pmatrix}-\eta^2 & 0\\ 0 & i\eta\end{pmatrix}
+
\begin{pmatrix}  0 & 0 \\ 2i\eta & 0\end{pmatrix}
\\=&
\begin{pmatrix}-3\eta^2 & 8i\eta \\ i\eta(2+\mu_3\eta^2) & -\eta^2
\end{pmatrix}
+\begin{pmatrix} O(\eta^4) & O(\eta^3)\\ O(\eta^5) & O(\eta^4)\end{pmatrix}\,,
\end{split}  
\end{equation}
where 
$\mu_3=-\frac{\mu_1}{2}+\frac{3}{4}=\frac{1}{2}+\frac{\pi^2}{24}>1/8$.
Claim~\ref{cl:akbound} in Appendix~\ref{sec:Rk} implies
$\|A_1(t)\|_{B(Y)}\lesssim e^{-a(4t+L)}$.
Thus we prove that $A(t)=A_0+A_1(t)$ satisfies \eqref{eq:H}.
\par
Let $U(t,s)$ be the semigroup generated by $A(t)$. 
By the variation of the constant formula,
\begin{equation*}
\begin{pmatrix}b(t)\\ x_y(t)\end{pmatrix}
=U(t,0)\begin{pmatrix}b(0)\\ x_y(0)\end{pmatrix}
+\sum_{i=1}^8\int_0^t U(t,s) \diag(1,\pd_y)\cN_i(s)ds\,.
\end{equation*}
By  Lemma~\ref{lem:decomp2} and Claim~\ref{cl:b-capprox},
$$\|b(0)\|_{Y_1}+\|x_y(0)\|_{Y_1}\lesssim
\|\tc(0)\|_{Y_1}+\eta_0\|\tx(0)\|_{Y_1}\lesssim\|v_0\|_{X_1}\,.$$ 
Applying Lemma~\ref{lem:decay-BB} to the first term of the 
right hand side, we have for $k\ge0$,
\begin{equation}
  \label{eq:fN0}
\begin{split}
& \|\pd_y^kb(t)\|_Y+\|\pd_y^{k+1}x(t)\|_Y \lesssim  
(1+t)^{-(2k+1)/4}\|v_0\|_{X_1}
+ \fN_1^k+\fN_2^k\,,\\
& \fN_1^k=\int_0^t\|\pd_y^kU(t,s)\diag(1,\pd_y)\cN_1(s)\|_Y\,ds\,,\\
& \fN_2^k=\int_0^t\left\|\pd_y^kU(t,s)\sum_{i=2}^8\diag(1,\pd_y)\cN_i(s)
\right\|_Y\,ds\,.
\end{split}  
\end{equation}
\par 
Now we will estimate $\fN_i^k$ ($i=1$, $2$, $k=0$, $1$, $2$).
First, we  estimate $\fN_1^k$. Let
$n_1=6bx_y$ and $n_2=2(\tilde{c}-b)+3(x_y)^2$.
Then $\diag(1,\pd_y)\cN_1=\pd_y\wP_1{}^t(n_1, n_2)$.
Since $[\pd_y,U(t,s)]=0$,
\begin{equation}
  \label{eq:fN1-est0}
\pd_y^kU(t,s)\diag(1,\pd_y)\cN_1=\pd_y^{k+1}U(t,s){}^t\wP_1(n_1(s), n_2(s))\,.
\end{equation}
By \eqref{eq:Y1-L1}, Claim~\ref{cl:b-capprox}
and the fact that $[\pd_y,\wP_1]=0$,
\begin{equation}
  \label{eq:cN1-est}
\begin{split}
\|\wP_1n_1\|_{Y_1}+\|\wP_1n_2\|_{Y_1} \lesssim &
\|b\|_Y\|x_y\|_Y+\|x_y\|_Y^2+\|b-\tilde{c}\|_{Y_1}
\\ \lesssim & (1+\|\tc\|_{L^\infty})\|\tc\|_Y\|x_y\|_Y+\|\tx_y\|_Y^2+\|\tc\|_Y^2
\\ \lesssim & \bM_1(T)^2\la s\ra^{-1/2}\quad\text{for $s\in[0,T]$,}
\end{split}  
\end{equation}
\begin{equation}
  \label{eq:N1'}
  \begin{split}
& \|\pd_y\wP_1n_1\|_{Y_1}+\|\pd_y\wP_1n_2\|_{Y_1} \\ \lesssim &
\|b_y\|_Y\|x_y\|_Y+\|b\|_Y\|x_{yy}\|_Y+\|x_y\|_Y\|x_{yy}\|_Y+\|b_y-c_y\|_{Y_1}
\\ \lesssim & M_1(T)^2\la s\ra^{-1}\quad\text{for $s\in[0,T]$,}
  \end{split}
\end{equation}
\begin{equation}
  \label{eq:N1''}
\begin{split}
& \|\pd_y^2\wP_1n_1(s)\|_{Y_1}+\|\pd_y^2\wP_1n_2(s)\|_{Y_1}
\\ \lesssim & 
\|(bx_y)_{yy}\|_{L^1}+\|b_{yy}-\tilde{c}_{yy}\|_{Y_1}+\|((x_y)^2)_{yy}\|_{L^1}
\\ \lesssim & \|\tilde{c}\|_Y\|x_{yyy}\|_Y+\|c_y\|_Y\|x_{yy}\|_Y
+\|c_{yy}\|_Y\|x_y\|_Y+\|\tilde{c}\|_Y\|c_{yy}\|_Y+\|c_y\|_Y^2
\\ & +\|x_y\|_Y\|x_{yyy}\|_Y+\|x_{yy}\|_Y^2
\\ \lesssim & \bM_1(T)^2\la s\ra^{-5/4}\quad
\text{for any $s\in[0,T]$.}
\end{split}  
\end{equation}
Using Lemma~\ref{lem:decay-BB}, \eqref{eq:fN1-est0} with $k=0$ and
\eqref{eq:cN1-est}, we obtain
\begin{equation}
  \label{eq:fN1-esta}
\begin{split}
 \fN_1^0\lesssim & \bM_1(T)^2 \int_0^t \la t-s\ra^{-3/4}\la s\ra^{-1/2}ds
\lesssim \bM_1(T)^2\la t\ra^{-1/4}\quad\text{ for $t\in[0,T]$.}
\end{split}  
\end{equation}
In the last line, we use Claim~\ref{cl:abc}.
Using Lemma~\ref{lem:decay-BB}, \eqref{eq:cN1-est} for
$s\in [0,t/2]$ and \eqref{eq:N1'} for $s\in[t/2,t]$, we obtain
\begin{equation}
  \label{eq:fN1-estb}
\begin{split}
\fN_1^1\lesssim & \int_0^{t/2}\|\pd_y^2U(t,s)\|_{B(Y_1,Y)}
(\|\wP_1n_1(s)\|_{Y_1}+\|\wP_1n_2(s)\|_{Y_1})\,ds
\\ &+\int_{t/2}^t\|\pd_yU(t,s)\|_{B(Y_1,Y)}
(\|\pd_y\wP_1n_1(s)\|_{Y_1}+\|\pd_y\wP_1n_2(s)\|_{Y_1})\,ds
\\ \lesssim & 
\bM_1(T)^2 \left(\int_0^{t/2} \la t-s\ra^{-5/4}\la s\ra^{-1/2}ds
+\int_{t/2}^t\la t-s\ra^{-3/4}\la s\ra^{-1}ds\right)
\\ \lesssim & \bM_1(T)^2\la t\ra^{-3/4}\quad\text{ for $t\in[0,T]$.}
\end{split}  
\end{equation}
Similarly, we have
\begin{equation}
  \label{eq:fN1-estc}
\begin{split}
\fN_1^2\lesssim & \int_0^{t/2}\|\pd_y^3U(t,s)\|_{B(Y_1,Y)}
(\|\wP_1n_1(s)\|_{Y_1}+\|\wP_1n_2(s)\|_{Y_1})\,ds
\\ &+\int_{t/2}^t\|\pd_yU(t,s)\|_{B(Y_1,Y)}
(\|\pd_y^2\wP_1n_1(s)\|_{Y_1}+\|\pd_y^2\wP_1n_2(s)\|_{Y_1})\,ds
\\ \lesssim & 
\bM_1(T)^2\left(\int_0^{t/2}\la t-s\ra^{-7/4}\la s\ra^{-1/2}ds
+\int_{t/2}^t \la t-s\ra^{-3/4}\la s\ra^{-5/4}ds\right)
\\ \lesssim & \bM_1(T)^2\la t\ra^{-1}\quad\text{for $t\in[0,T]$.}
\end{split}  
\end{equation}
\par
The rest of nonlinear terms $\sum_{i=2}^8\diag(1,\pd_y)\cN_i$ can be rewritten
as a sum of $\cN'(t)$ and $\pd_y\cN''(t)$ satisfying
\begin{equation}
  \label{eq:N'N''}
  \begin{split}
& \|\cN'(t)\|_{Y_1}\lesssim (\bM_1(T)+\bM_2(T))^2\la t\ra^{-5/4}  
\quad\text{for $t\in[0,T]$,}\\
& \|\cN''(t)\|_{Y_1}\lesssim \bM_1(T)(\bM_1(T)+\bM_2(T))\la t\ra^{-1}
\quad\text{for $t\in[0,T]$.}
  \end{split}
\end{equation}
First, we prove decay estimates of $\fN_2^k$ ($k=0$, $1$, $2$)
presuming that \eqref{eq:N'N''} is true.
Then for $t\in[0,T]$ and $0\le k\le 2$,
\begin{equation}
  \label{eq:fN2-est}
\fN_2^k\lesssim (\bM_1(T)+\bM_2(T))^2\la t\ra^{-\min\{1,(2k+1)/4\}}\,.
\end{equation}
Indeed, Lemma~\ref{lem:decay-BB} implies that for $t\in[0,T]$,
\begin{align*}
\int_0^t \|\pd_y^kU(t,s)\cN'(s)\|_Y\,ds 
\lesssim & \int_0^t \|\pd_y^kU(t,s)\|_{B(Y_1,Y)}\|\cN'(s)\|_{Y_1}\,ds
\\ \lesssim &  (\bM_1(T)+\bM_2(T))^2
\int_0^t \la t-s\ra^{-(2k+1)/4}\la s\ra^{-5/4}ds\,,
\end{align*}
\begin{align*}
\int_0^t \|\pd_y^kU(t,s)\pd_y\cN''(s)\|_Y\,ds 
\lesssim & \int_0^t \|\pd_y^{k+1}U(t,s)\|_{B(Y_1,Y)}\|\cN''(s)\|_{Y_1}\,ds
\\ \lesssim &  \bM_1(T)(\bM_1(T)+\bM_2(T))
\int_0^t \la t-s\ra^{-(2k+3)/4}\la s\ra^{-1}ds\,.
\end{align*}
By  Claim~\ref{cl:abc},
$$\int_0^t \la t-s\ra^{-(2k+1)/4}\la s\ra^{-5/4}ds \lesssim
\la t\ra^{-(2k+1)/4}\quad\text{for $k=0$, $1$, $2$,}$$
and
$$\int_0^t \la t-s\ra^{-(2k+3)/4}\la s\ra^{-1}ds \lesssim
\begin{cases}
&  \la t\ra^{-3/4}\log\la t\ra\text{ for $k=0$,}\\
&  \la t\ra^{-1}\text{ for $k=1$, $2$.}
\end{cases}$$
Thus we have \eqref{eq:fN2-est} presuming \eqref{eq:N'N''}.
\par

Now we turn to estimate $\cN_k$ ($2\le k\le 8$).
First, we will estimate $\cN_2^k$. Let $E_1=\diag(1,0)$, $E_2=\diag(0,1)$ and
\begin{gather*}
\cN_{21}= B_3^{-1}E_1\wP_1 R^7\,,\enskip
\cN_{22}=(B_3^{-1}-B_1^{-1})E_2 \wP_1 R^7\,,\enskip
\cN_{23}=B_1^{-1}E_2 \wP_1 R^7\,.
\end{gather*}
Then $\cN_2=\cN_{21}+\cN_{22}+\cN_{23}$.
Claim~\ref{cl:invB_3} implies that 
$B_3^{-1}$ is uniformly bounded on $Y$ and $Y_1$ for $s\in[0,T]$.
Thus by Claim~\ref{cl:R6},
\begin{equation}
  \label{eq:cN21-est}
\|\cN_{21}\|_{Y_1}\lesssim \|R^7_1\|_{Y_1}
\lesssim  \bM_1(T)^2\la s\ra^{-5/4}\quad\text{for $s\in[0,T]$.}
\end{equation}
\par
Next we will estimate $\|\cN_{22}\|_{Y_1}$.
Since $(I+\wC_2)^{-1}B_1^{-1}E_2=\frac12E_2$, we have
$\bS_jB_1^{-1}=\frac12\wS_j$ for $1\le j\le 5$ and $\bS_3B_1^{-1}E_2\wP_1R^7= 0$.
Thus $\cN_{22}$ can be rewritten as
\begin{align*}
\cN_{22}=& B_3^{-1}(B_1-B_3)B_1^{-1}E_2\wP_1R^7
\\=& -B_3^{-1}\left\{\wC_1-\pd_y^2(\bS_1+\bS_2)+\bS_3+\bS_4+\bS_5\right\}
B_1^{-1}E_2\wP_1R^7
\\=&  -\frac12B_3^{-1}\left\{\wC_1-\pd_y^2(\wS_1+\wS_2)+\wS_4+\wS_5\right\}
E_2 R^7\,.
\end{align*}
By Claims~\ref{cl:S1}--\ref{cl:S5-S6}, \ref{cl:T1-T2} and
\ref{cl:R6}, we have for $s\in [0,T]$ and $j=1$, $2$,
\begin{align*}
\|\pd_y^2\wS^1E_2R^7\|_{Y_1}\lesssim & \eta_0
 \|\pd_yS^1_j[\varphi_c'](R^7_2)\|_{Y_1}
\lesssim \|\pd_yR^7_2\|_{Y_1}\lesssim M_1(T)^2\la s\ra^{-5/4}\,,\\
\|\pd_y^2\wS^2E_2R^7\|_{Y_1}\lesssim & \eta_0\|\pd_yS^2_j[\varphi_c'](R^7_2)\|_{Y_1}
\lesssim  \|\tc\|_Y\|\pd_yR^7_2\|_Y+\|c_y\|_Y\|R^7_2\|_Y
\\ \lesssim & M_1(T)^2\la s\ra^{-3/2}\,,  
\end{align*}
\begin{equation}
  \label{eq:wS4R7}
\begin{split}
& \left\|\wS_4 E_2R^7\right\|_{Y_1}\lesssim 
\sum_{j=3,4\,, k=1,2}\|S^j_k[\psi'](\sqrt{c}-\sqrt{2})R^7_2\|_{Y_1}
\\ \lesssim & 
\left(\|\|S^3_k[\psi']\|_{B(Y_1)}\|(\sqrt{c}-\sqrt{2})R^7_2\|_{L^1}
+\|\|S^4_k[\psi']\|_{B(Y,Y_1)}\|(\sqrt{c}-\sqrt{2})R^7_2\|_{L^2}\right)
\\ \lesssim & \bM_1(T)^3e^{-a(4s+L)}\la s\ra^{-3/2}\,,
\end{split}  
\end{equation}
\begin{align*}
& \|\wS^5E_2R^7\|_{Y_1}\lesssim \sum_{k=1,2}\|S^6_k(R^7_2)\|_{Y_1}
\lesssim \|v(t,\cdot)\|_X\|R^7_2\|_Y
\lesssim \bM_1(T)^2\bM_2(T)\la s\ra^{-2}\,,\\  
& \|\wC_1E_2R^7\|_{Y_1}\le \|\cC_1\|_{B(Y,Y_1)}\|R^7_2\|_Y
\lesssim \bM_1(T)^3\la s\ra^{-3/2}\,.
\end{align*}
Combining the above, we have
\begin{equation}
  \label{eq:cN22-est}
\|\cN_{22}\|_{Y_1}\lesssim \bM_1(T)^2\la s\ra^{-5/4}  \quad\text{for $s\in[0,T]$.}
\end{equation}
By Claim~\ref{cl:R6}
and the fact that $\diag(1,\pd_y)\cN_{23}=\frac12\pd_y(0,R^7_2)$, 
\begin{equation}
  \label{eq:cN23-est}
\|\diag(1,\pd_y)\cN_{23}\|_{Y_1}\lesssim \|\pd_yR^7_2\|_{Y_1}
\lesssim  \bM_1(T)^2\la s\ra^{-5/4}  \quad\text{for $s\in[0,T]$.}
\end{equation}
Combining \eqref{eq:cN21-est}, \eqref{eq:cN22-est} and \eqref{eq:cN23-est},
we have 
\begin{equation}
  \label{eq:cN2-est}
\|\diag(1,\pd_y)\cN_2\|_{Y_1} \lesssim  \bM_1(T)^2\la s\ra^{-5/4}
\quad\text{for $s\in[0,T]$.}
\end{equation}
\par
Next, we will estimate $\cN_3$.
Claims~\ref{cl:R3} and \ref{cl:R4-R5} imply that for $s\in[0,T]$,
$$\|R^3\|_{Y_1}+\|R^4\|_{Y_1}+\|R^6\|_{Y_1}\lesssim
(\bM_1(T)+\bM_2(T))^2\la s\ra^{-3/2}\,.$$
Using Claims~\ref{cl:S3}, \ref{cl:S4} and \eqref{eq:Y1-L1}, we can
show that for $s\in[0,T]$,
\begin{align*}
& \left\|\wS_4 \begin{pmatrix}0\\ \tc \end{pmatrix}
\right\|_{Y_1} \lesssim
\sum_{\substack{j=3,4\\ k=1,2}}\|S^j_k[\psi'](\sqrt{c}-\sqrt{2})\tc\|_{Y_1}
\lesssim  \bM_1(T)^2 e^{-a(4t+L)}\la s\ra^{-1/2}
\end{align*}
in the same way as \eqref{eq:wS4R7}. Thus for $s\in[0,T]$,
\begin{equation}
  \label{eq:cN3-est}
\|\cN_3(s)\|_{Y_1} \lesssim  \|B_3^{-1}\|_{B(Y_1)}\|\wR^1\|_{Y_1}
 \lesssim  (\bM_1(T)+\bM_2(T))^2\la s\ra^{-3/2}\,.
\end{equation}
\par

Next we will estimate $\cN_4^k$.
Let $\cN_{41}=B_3^{-1}\wR^2$ and $\cN_{42}= [B_3^{-1},\pd_y]\wR^2$.
Then $\cN_4=\pd_y\cN_{41}+\cN_{42}$. By Claims~\ref{cl:R1-R2}, \ref{cl:R4-R5} and
\eqref{eq:pdy-bound},
\begin{equation}
  \label{eq:cN40-est}
\|\wR^2\|_{Y_1}\lesssim \bM_1(T)(\bM_1(T)+\bM_2(T))\la s\ra^{-1}
\quad\text{for $s\in[0,T]$.}
\end{equation}
By Claims~\ref{cl:invB_3} and \ref{cl:[B3,pdy]},
\begin{equation}
  \label{eq:[b3-1,pdy]}
\|[B_3^{-1},\pd_y]\|_{B(Y,Y_1)} \lesssim  (\bM_1(T)+\bM_2(T))\la s\ra^{-3/4}
\quad\text{for $s\in[0,T]$.}  
\end{equation}
Combining \eqref{eq:cN40-est}, \eqref{eq:[b3-1,pdy]} and Claim~\ref{cl:invB_3},
we have for $s\in[0,T]$, 
\begin{gather}
  \label{eq:cN41-est}
\|\cN_{41}(s)\|_{Y_1}
\lesssim  \bM_1(T)(\bM_1(T)+\bM_2(T))\la s\ra^{-1}\,,
\\  \label{eq:cN42-est}
\|\cN_{42}(s)\|_{Y_1}\lesssim \bM_1(T)(\bM_1(T)+\bM_2(T))^2\la s\ra^{-7/4}\,.
\end{gather}
\par

Next we will estimate $\cN_5$.
Let $r_8=\wP_1\{(I+\mathcal{C}_2)(c_yx_y)-(bx_y)_y\}$. Then
$$\|r_8\|_Y\lesssim \|r_8\|_{Y_1}\lesssim \bM_1(T)^2\la s\ra^{-1}\,,
\quad \|\pd_yr_8\|_{Y_1}\lesssim \bM_1(T)^2\la s\ra^{-5/4}\,.$$
Here we use Claims~\ref{cl:(1+wT2)^{-1}} and \ref{cl:T-pdy}.
By \eqref{eq:def-B3} and \eqref{eq:pdy-bound},  
\begin{align*}
\|R^8\|_{Y_1}\lesssim & \|\bS_1\|_{B(Y_1)}\|\pd_yr_8\|_{Y_1}
+\left(\|[\pd_y,\bS_1]\|_{B(Y,Y_1)}+\sum_{j=2,4,5}\|\bS_j\|_{B(Y,Y_1)}\right)\|r_8\|_Y
\\ & +\|\bS_3\|_{B(Y_1)}\|r_8\|_{Y_1}\,.
\end{align*}
Combining the above with \eqref{eq:bS1-bound}--\eqref{eq:bS5-bound} and
\eqref{eq:[bS1,pdy]} in Appendix~\ref{sec:[B3,pdy]},
we have 
\begin{equation*}
\|R^8\|_{Y_1}\lesssim \bM_1(T)^2(1+\bM_1(T)+\bM_2(T))\la s\ra^{-5/4}
\quad\text{for $s\in[0,T]$.}  
\end{equation*}
Since $\|\pd_y^2\wS_0\|_{B(Y_1)}\lesssim\eta_0^2$ by Claim~\ref{cl:S1} and
\eqref{eq:pdy-bound}, it follows from Claim~\ref{cl:b-capprox}  that
$$\|R^9\|_{Y_1}\lesssim \|b_{yy}-c_{yy}\|_{Y_1}
\lesssim \bM_1(T)^2\la s\ra^{-5/4}\quad\text{for $s\in[0,T]$.}$$
By Claims~\ref{cl:wA1-bound} and ~\ref{cl:b-capprox},
$$\|R^{10}\|_{Y_1}\lesssim e^{-a(4s+L)}\|b-\tc\|_{Y_1}
\lesssim \bM_1(T)^2\la s\ra^{-1/2}e^{-a(4s+L)}\quad\text{for $s\in[0,T]$.}$$
Thus we have 
\begin{equation}
  \label{eq:cN5-est}
\|\cN_5\|_{Y_1}\lesssim
\|\wR^3(s)\|\lesssim \bM_1(T)^2\la s\ra^{-5/4}\text{ for $s\in[0,T]$.}  
\end{equation}
\par

Next we will estimate $\cN_6^k$.
Since the second column of $\widetilde{A}_1(t)$ is $0$,
$$\cN_6=(B_3^{-1}-B_4^{-1})\widetilde{A}_1(t)
\begin{pmatrix}  b\\ 0\end{pmatrix}\,.$$
By the definitions of $B_3$ and $B_4$,
$$B_3-B_4=\wC_1+\pd_y^2(\bS_1-\wS_1)+\pd_y^2\bS_2-(\bS_3-\wS_3)-\bS_4-\bS_5\,.$$
Hence it follows from Claims~\ref{cl:(1+wT2)^{-1}}, \ref{cl:difS1S3},
\eqref{eq:wS2bound} and \eqref{eq:bS4-bound}--\eqref{eq:pdy-bound} that
\begin{align*}
\|B_4-B_3\|_{B(Y,Y_1)} \le 
& \|\wC_1\|_{B(Y,Y_1)}+\sum_{j=1,3}\|\pd_y^2(\bS_j-\wS_j)\|_{B(Y,Y_1)}
+\sum_{j=2,4,5}\|\pd_y^2\bS_j\|_{B(Y,Y_1)}
\\ \lesssim & (\bM_1(T)+\bM_2(T))\la s\ra^{-1/4} \quad\text{for $s\in[0,T]$.}
  \end{align*}   
In view of Claims~\ref{cl:invB_3} and \ref{cl:invB_4} and the above,
we have for $s\in[0,T]$,
\begin{equation}
  \label{eq:invb3-b4}
  \begin{split}
\|B_3^{-1}-B_4^{-1}\|_{B(Y,Y_1)} \le & 
\|B_4^{-1}\|_{B(Y_1)}\|B_4-B_3\|_{B(Y,Y_1)}\|B_3^{-1}\|_{B(Y)}
\\ \lesssim & (\bM_1(T)+\bM_2(T))\la s\ra^{-1/4}\,.    
  \end{split}
\end{equation}
Combining Claim~\ref{cl:wA1-bound} and \eqref{eq:invb3-b4}, we have
for $s\in[0,T]$,
\begin{equation}
  \label{eq:cN6-est}
  \begin{split}
\|\cN_6\|_{Y_1}\lesssim &
\|B_3^{-1}-B_4^{-1}\|_{B(Y,Y_1)}\|\widetilde{\mathcal{A}}_1(t)\|_{B(Y)}\|b\|_Y
\\ \lesssim  & (\bM_1(T)+\bM_2(T))\bM_1(T)\la s\ra^{-1/2}e^{-a(4s+L)}\,.
  \end{split}
\end{equation}
\par
Since $\|\wS_0\|_{B(Y)}\lesssim1$ by Claim~\ref{cl:S1},
it follows from \eqref{eq:pdy-bound} and \eqref{eq:invb3-b4} that
for $s\in[0,T]$,
\begin{equation}
  \label{eq:cN7-est}
\|\cN_7\|_{Y_1}\lesssim \|B_3^{-1}-B_4^{-1}\|_{B(Y,Y_1)}\|b_{yy}\|_Y
\lesssim   (\bM_1(T)+\bM_2(T))\bM_1(T)\la s\ra^{-5/4}\,.
\end{equation}
\par
Finally, we will estimate $\cN_8^k$. Let
\begin{align*}
& 
\cN_{81}=(B_4^{-1}-B_3^{-1})\pd_y^2\wS_0\begin{pmatrix}0 \\ x_{yy} \end{pmatrix}
\,,\quad
\wR^4= B_4^{-1}B_2\begin{pmatrix}0 \\ x_{yy} \end{pmatrix}\,,
\\ & 
\cN_{82}=B_3^{-1}(B_4-B_3+\wC_1)\wR^4\,,\quad
\cN_{83}=B_1^{-1}(B_3-B_1-\wC_1)B_3^{-1}\wC_1\wR^4\,,
\\ & \cN_{84}=-B_1^{-1}\wC_1(I-B_3^{-1}\wC_1)\wR^4\,.
\end{align*}
Then $\cN_8=\sum_{1\le j\le 4}\cN_{8j}$.
Since $[\pd_y,\wS_0]=0$, we have 
$$\left\|\pd_y\wS_0  \begin{pmatrix} 0 \\ x_{yy} \end{pmatrix}\right\|_Y
\lesssim \|x_{yyy}\|_Y.$$
Combining the above with \eqref{eq:pdy-bound} and \eqref{eq:invb3-b4},
we see that for $s\in[0,T]$,
\begin{equation}
  \label{eq:cN81-est}
  \begin{split}
\|\cN_{81}\|_{Y_1}  \lesssim  &
\eta_0(\bM_1(T)+\bM_2(T))\la s\ra^{-1/4}\|x_{yyy}\|_Y
\\ \lesssim & \eta_0(\bM_1(T)+\bM_2(T))\bM_1(T)\la s\ra^{-5/4}\,.    
  \end{split}
\end{equation}
Next we will estimate $\cN_{82}$. Let
$$n_3=\pd_y(\bS_1-\wS_1+\bS_2)\pd_y\wR^4\,,\enskip
n_4=\pd_y[\pd_y,\bS_1+\bS_2]\wR^4\,,\enskip
n_5=(\bS_3-\wS_3+\bS_4+\bS_5)\wR_4\,.$$
Then $\cN_{82}=B_3^{-1}(n_5-n_3-n_4)$.
Claim~\ref{cl:invB_4} and the fact that  $[\pd_y,B_4]=0$ imply that 
for $s\in[0,T]$,
\begin{equation}
  \label{eq:wR4-est}
  \begin{split}
& \|\wR^4\|_Y\lesssim \|x_{yy}\|_Y\lesssim \bM_1(T)\la s\ra^{-3/4}\,,\\ 
& \|\pd_y\wR^4\|_Y\lesssim \|x_{yyy}\|_Y\lesssim \bM_1(T)\la s\ra^{-1}\,.
  \end{split}
\end{equation}
We see that $\|n_3\|_{Y_1} \lesssim \bM_1(T)^2\la s\ra^{-5/4}$ follows from
Claim~\ref{cl:difS1S3}, \eqref{eq:wS2bound}, \eqref{eq:pdy-bound}
and \eqref{eq:wR4-est} and that 
$\|n_4\|_{Y_1} \lesssim \bM_1(T)^2\la s\ra^{-3/2}$ follows from
\eqref{eq:[bS1,pdy]}, \eqref{eq:[pdy,wS2]}, \eqref{eq:pdy-bound}
and \eqref{eq:wR4-est}.
By Claim~\ref{cl:difS1S3}, \eqref{eq:bS4-bound}, \eqref{eq:bS5-bound}
and \eqref{eq:wR4-est}, 
$$\|n_5\|_{Y_1}\lesssim \bM_1(T)\la s\ra^{-3/4}
(\bM_1(T)e^{-a(4s+L)}\la s\ra^{-1/4}+\bM_2(T)\la s\ra^{-3/4})\,.$$
Thus we have
\begin{equation}
  \label{eq:cN82-est}
\|\cN_{82}\|_{Y_1}\lesssim \bM_1(T)(\bM_1(T)+\bM_2(T))\la s\ra^{-5/4}
\quad\text{for $s\in[0,T]$.}
\end{equation}
Next we will estimate $\cN_{83}$. Let
\begin{align*}
& n_6=[\pd_y,\bS_1+\bS_2]B_3^{-1}\wC_1\wR^4\,,
\quad n_7=(\bS_1+\bS_2)[\pd_y,B_3^{-1}]\wC_1\wR^4\,,
\\ & n_8=(\bS_1+\bS_2)B_3^{-1}\pd_y\wC_1\wR^4\,,\quad
n_9=(\bS_3+\bS_4+\bS_5)B_3^{-1}\wC_1\wR^4\,.
\end{align*}
Then $\cN_{83}=B_1^{-1}\pd_y(n_6+n_7+n_8)-B_1^{-1}n_9$.
By Claims~\ref{cl:(1+wT2)^{-1}}, \ref{cl:invB_3}, \eqref{eq:[bS1,pdy]}
and \eqref{eq:[pdy,wS2]},
\begin{align*}
\|n_6\|_{Y_1}\lesssim & \sum_{j=1,2}
\|[\pd_y,\bS_j]\|_{B(Y,Y_1)}\|B_3^{-1}\|_{B(Y)}\|\wC_1\|_{B(Y)}\|\wR^4\|_Y{Y_1}
\\ \lesssim & \bM_1(T)^2\la s\ra^{-2}\quad\text{for $s\in[0,T]$.}
\end{align*}
Using Claims~\ref{cl:invB_3}, \ref{cl:invB_3},
\eqref{eq:bS1-bound}--\eqref{eq:bS5-bound}, \eqref{eq:wR4-est}
and \eqref{eq:[pdy,T]}, we can obtain
\begin{gather*}
\|n_7\|_{Y_1}\lesssim \bM_1(T)^2(\bM_1(T)+\bM_2(T))\la s\ra^{-2}\,,\quad
\|n_8\|_{Y_1}\lesssim \bM_1(T)^2\la s\ra^{-5/4}\,, \\
\|n_9\|_{Y_1}\lesssim \bM_1(T)^2\la s\ra^{-1}(e^{-a(4s+L)}+\bM_2(T)\la s\ra^{-3/4})\,. 
\end{gather*}
Combining the above, we have
\begin{equation}
  \label{eq:cN83-est}
  \|\cN_{83}\|_{Y_1}\lesssim \bM_1(T)^2\la s\ra^{-5/4}\,.
\end{equation}
Since $\diag(1,\pd_y)B_1^{-1}\wC_1=\frac12\pd_y\wC_1$,
\begin{align*}
2\diag(1,\pd_y)\cN_{84}=& 
\left\{[\pd_y,\wC_1]B_3^{-1}+\wC_1[\pd_y,B_3^{-1}]+(\wC_1B_3^{-1}-I)\pd_y\right\}
\wC_1\wR^4\,.
\end{align*}
Combining Claims~\ref{cl:invB_3}, \ref{cl:T1-T2}, \ref{cl:T-pdy} and
\ref{cl:[B3,pdy]} with \eqref{eq:wR4-est}, we obtain
\begin{equation}
  \label{eq:cN84'}
\|\diag(1,\pd_y)\cN_{84}\|_{Y_1}\lesssim
\bM_1(T)^2\la s\ra^{-5/4}\quad\text{for $s\in[0,T]$.} 
\end{equation}
Thus for $s\in[0,T]$,
\begin{equation}
  \label{eq:cN8-est}
\sum_{1\le j\le 4}\|\diag(1,\pd_y)\cN_{8j}\|\lesssim
\bM_1(T)(\bM_1(T)+\bM_2(T))\la s\ra^{-5/4}
\end{equation}
follows from \eqref{eq:cN81-est}, \eqref{eq:cN82-est}, \eqref{eq:cN83-est} and
\eqref{eq:cN84'}.
\par

Now let 
$\cN'=\diag(1,\pd_y)\left(\sum_{\substack{2\le i\le 8\\ i\ne4}}\cN_i
+\cN_{42}\right)$ and $\cN''=\diag(1,\pd_y)\cN_{41}$.
Then \eqref{eq:N'N''} follows immediately from 
\eqref{eq:cN2-est}, \eqref{eq:cN3-est}, \eqref{eq:cN41-est},
\eqref{eq:cN42-est}, \eqref{eq:cN5-est}, \eqref{eq:cN6-est},
\eqref{eq:cN7-est} and \eqref{eq:cN8-est}.
By \eqref{eq:fN0} and Claim~\ref{cl:b-capprox},
\begin{equation}
  \label{eq:fNfinal}
\begin{split}
\|\pd_y^k\tc(t)\|_Y+\|\pd_y^{k+1}x(t)\|_Y \lesssim &
\la t\ra^{-(2k+1)/4}\|v_0\|_{X_1}+\fN_1^k+\fN_2^k
\\ & +\bM_1(T)^2\la t\ra^{-\min\{(2k+3)/4,3/2\}}\,.
\end{split}    
\end{equation}
for $0\le k\le 2$ and $t\in[0,T]$.
Combining \eqref{eq:fN1-esta}--\eqref{eq:fN1-estc} and \eqref{eq:fN2-est} with
\eqref{eq:fNfinal}, we obtain \eqref{eq:M1-bound}.
This completes the proof of Lemma~\ref{lem:M1-bound}.
\end{proof}
\bigskip

\section{$L^2$ bound on $v(t,z,y)$}
\label{sec:L2norm}
In this section, we will estimate $\bM_3(T)$
assuming that $\bM_1(T)$ and $\bM_2(T)$ are small.
First, we will show a variant of the $L^2$ conservation law on $v$.
\begin{lemma}
  \label{lem:L2conserve}
Let $a\in(0,2)$ and $T>0$.
Suppose $v(t)\in C([0,T];X\cap L^2(\R^2))$ is 
a solution of \eqref{eq:v} and that $v(t)$, $c(t)$ and $x(t)$
satisfy \eqref{eq:orth}, \eqref{eq:C1cx} and \eqref{eq:bd-vc}.
Then
$$Q(t,v):=\int_{\R^2}\left\{v(t,z,y)^2
-2\psi_{c(t,y),L}(z+4t)v(t,z,y)\right\}\,dzdy$$
satisfies for $t\in[0,T]$,
\begin{multline*}
Q(t,v)= Q(0,v)+
2\int_0^t\int_{\R^2}\left(\ell_{11}+\ell_{12}+6\varphi_{c(s,y)}'(z)\tpsi_{c(s,y)}(z)
\right)v(s,z,y)\,dzdyds
\\  -2\int_0^t\int_{\R^2} \ell\psi_{c(t,y),L}(z+4t)\,dzdyds
-6\int_{\R^2}\varphi_{c(t,y)}'(z)v(t,z,y)^2\,dzdy
\\  -6\int_{\R^2}v(t,z,y)
\bigl\{c_{yy}(t,y)\int_{-\infty}^z\pd_c\varphi_{c(t,y)}(z_1)\,dz_1
\\ +c_y(t,y)^2\int_{-\infty}^z \pd_c^2\varphi_{c(t,y)}(z_1)\,dz_1\bigr\}\,dzdy\,.
\end{multline*}
\end{lemma}
\begin{proof}
Suppose $v_0\in H^3(\R^2)$ and $\pd_x^{-1}v_0\in H^2(\R^2)$.
Then as in Section~\ref{sec:modulation}, we have
$v(t)\in C([0,T];H^3(\R^2))$
and $\pd_x^{-1}\pd_yv(t)\in C([0,T];H^1(\R^2))$.
Using \eqref{eq:v}, we have
\begin{align*}
\frac{d}{dt}\int_{\R^2}v(t,x,y)^2\,dxdy=&
2 \int v\left(\mL_{c(t,y)}v+\ell+\pd_z(N_1+N_2)+N_3\right)\,dzdy
\\=&  2\int\ell vdzdy +6\int (\tpsi_c'-\varphi_c')v^2\,dzdy\,,
\end{align*}
and
\begin{multline*}
\frac{d}{dt}\int_{\R^2} \psi_{c(t,y),L}(z+4t)v(t,z,y)\,dzdy
= \int \left(c_t\pd_c\tpsi_c+4\tpsi_c'\right)v
\\  +\int \tpsi_c\left\{\mL_{c(t,y)}v+\ell+\pd_z(N_1+N_2)+N_3\right\}\,.
\end{multline*}
Since $(\pd_z^{-1}\pd_y^2)v(z,y)=-\int_z^\infty \pd_y^2v(z_1,y)\,dz_1$, we have
\begin{multline*}
 \int_{\R^2} \tpsi_c\mL_cv\, dzdy=
-\int \tpsi_c\pd_z(\pd_z^2-2c+6\varphi_c)v\,dzdy
-3\int\tpsi_c\pd_z^{-1}\pd_y^2v\,dzdy
\\= \int v(\tpsi_c'''-2c\tpsi_c'+6\varphi_c\tpsi_c')\,dzdy
\\  +3\int_{\R^2}v\left\{c_{yy}\int^z_{-\infty}\pd_c\tpsi_c
+(c_y)^2\int^z_{-\infty}\pd_c^2\tpsi_c\right\}\,dzdy\,,
\end{multline*}
where $\pd_c^k\tpsi_c=\pd_c^k\psi_{c(t,y),L}(x+4t)$ for $k\ge0$.
By integration by parts, we have
\begin{equation*}
  \int_{\R^2} \tpsi_c\pd_zN_1\,dzdy=3\int \tpsi'v^2\,dzdy\,,
\end{equation*}
and
\begin{align*}
& \int (\pd_zN_2+N_3)\tpsi_c\,dzdy\\ =&
-\int v\left\{\left(x_t-2c-3(x_y)^2\right)\tpsi_c'
+3x_{yy}\tpsi_c+6c_yx_y\pd_c\tpsi_c+6\tpsi_c'\tpsi_c\right\}\,dzdy\,.
\end{align*}
Combining the above, we have
  \begin{align*}
&\frac{d}{dt}\int_{\R^2}\left\{v(t,z,y)^2
-2\psi_{c(t,y),L}(z+4t)v(t,z,y)\right\}\,dzdy
\\= & 2\int_{\R^2}\left\{\left(
\ell_{11}+\ell_{12}+6\varphi_{c(t,y)}'(z)\tpsi_{c(t,y)}(z)\right)v(t,z,y)
-\ell\psi_{c(t,y),L}(z+4t)\right\}\,dzdy
\\ & +6\int_{\R^2}v(t,z,y)
\left\{c_{yy}(t,y)\int_\R\left(\pd_c\varphi_{c(t,y)}(z_1)
-\pd_c\psi_{c(t,y),L}(z_1+4t)\right)dz_1\right\}\,dzdy
\\ & +6\int_{\R^2}v(t,z,y)
\left\{c_y(t,y)^2\int_\R\left(\pd_c^2\varphi_{c(t,y)}(z_1)
-\pd_c^2\psi_{c(t,y),L}(z_1+4t)\right)dz_1\right\}\,dzdy
\\ & -6\int_{\R^2}\ell_{13}^* v(t,z,y)\,dzdy
-6\int_{\R^2}\varphi_{c(t,y)}'(z)v(t,z,y)^2\,dzdy\,,
  \end{align*}
where
$$\ell_{13}^*=
c_{yy}(t,y)\int_{-\infty}^z\pd_c\varphi_{c(t,y)}(z_1)\,dz_1
+c_y(t,y)^2\int_{-\infty}^z \pd_c^2\varphi_{c(t,y)}(z_1)\,dz_1\,.$$
Since $\int_\R\pd_c^k\varphi_c(x)\,dx=\int_\R\pd_c^k\psi_{c,L}(x)\,dx$ 
for $k\ge0$ by \eqref{eq:0mean}, we see that Lemma~\ref{lem:L2conserve}
holds provided $v_0$ and $\pd_x^{-1}\pd_yv_0$ are smooth.
\par
For general $v_0\in X\cap L^2(\R^2)$, we can prove Lemma~\ref{lem:L2conserve}
by a standard limiting argument. 
The mapping 
\begin{equation}
\label{map1:c}
L^2(\R^2)\ni v_0\mapsto \tv(t)\in C([0,T];L^2(\R^2))
\end{equation}
is continuous for any $T>0$ by \cite{MST}.
On the other hand, it follows from \eqref{eq:virial-fx} that
a solution $\tv(t)$ of \eqref{eq:v-fix} satisfies
$\sup_{t\in[0,T]}\|\tv(t)\|_X\le C$,
where $C$ is a constant depending only on $T$, $\|\tv(0)\|_{L^2(\R^2)}$ and
$\|\tv(0)\|_X$. Thus the mapping
\begin{equation}
\label{map2:c}
X\cap L^2(\R^2)\ni v_0\mapsto \tv(t)\in C([0,T];L^2(\R^2;e^{ax}dxdy))
\end{equation}
is continuous since
$\|u\|_{L^2(\R^2;e^{ax}dxdy)}\lesssim \|u\|_{L^2(\R^2)}^{1/2}\|u\|_X^{1/2}$
for every $u\in X\cap L^2(\R^2)$.
If $\eta_0$ is sufficiently small, it is clear from
Lemma~\ref{lem:decomp} and Remark~\ref{rem:decomp} that
$(\tc(t),\tx(t))\in Y\times Y$ as well as its time derivate
depends continuously on $v(t)\in L^2(\R^2;e^{ax}dxdy)$.
This completes the proof of Lemma~\ref{lem:L2conserve}.
\end{proof}

Using Lemma~\ref{lem:L2conserve}, we will estimate the upper bound of 
$\|v(t)\|_{L^2}$.
\begin{lemma}
  \label{lem:M4-bound}
Let $a\in(0,1)$ and $\delta_4$ be as in Lemma~\ref{lem:M1-bound}.
Then there exists a positive constant $C$ such that if
$\bM_1(T)+\bM_2(T)+\eta_0+e^{-aL}\le \delta_4$, then
$$\bM_3(T)\le  C(\|v_0\|_{L^2(\R^2)}+\bM_1(T)+\bM_2(T))\,.$$
\end{lemma}
\begin{proof}
Remark~\ref{rem:decomp} and Proposition~\ref{prop:continuation} tell us
that we can apply Lemma~\ref{lem:L2conserve} for $t\in[0,T]$
if $\bM_1(T)$ and $\bM_2(T)$ are sufficiently small.
\par
Since we have for $j$, $k\ge0$ and $z\in\R$,
\begin{equation}
  \label{eq:exp-localized}
  \pd_z^j\pd_c^k\varphi_c(z)\lesssim e^{-2a|z|}\,,\quad
\int_{-\infty}^z\pd_c^j\varphi_c(z_1)dz_1\lesssim \min(1,e^{2az})\,,
\end{equation}
it follows that
\begin{equation}
  \label{eq:l2pf1}
\begin{split}
& \left|\int_{\R^2}(\ell_{11}+\ell_{12}-\ell_{13}^*)v\,dzdy\right|
\\  \lesssim &
\bigl(\|c_t-6c_yx_y\|_{L^2}+\|x_t-2c-3(x_y)^2\|_{L^2}+\|x_{yy}\|_{L^2}
+\|c_{yy}\|_{L^2}+\|c_y\|_{L^4}^2\bigr)\|v\|_X\,,
\end{split}  
\end{equation}
\begin{equation}
  \label{eq:l2pf2}
\left|\int_{\R^2} \varphi_{c(t,y)}'(z)v(t,z,y)^2\,dzdy\right|
\lesssim \|v\|_X^2\,,  
\end{equation}
\begin{equation}
  \label{eq:l2pf3}
\begin{split}
\left|\int_{\R^2}\varphi_{c(t,y)}'(z)\tpsi_{c(t,y)}(z)v(t,z,y)\,dzdy\right|
\lesssim & \|e^{az}\tpsi_{c(t,y)}\|_{L^2_{yz}}\|v\|_Xe^{-a(4t+L)}\,.  
\end{split}  
\end{equation}
\par

Next,  we will estimate $\int_{\R^2}\ell \tpsi_{c(t,y)}$.
In view of \eqref{eq:exp-localized}, we see that
$\ell_{11}+\ell_{12}$ is exponentially localized and that
\begin{equation}
  \label{eq:l2pf4}
\begin{split}
&  \left|\int_{\R^2}(\ell_{11}+\ell_{12})\tpsi_{c(t,y)}(t,z,y)dzdy\right|
\\ \le & \|e^{-az}(\ell_{11}+\ell_{12})\|_{L^2_{yz}}
\|e^{az}\tpsi_{c(t,y)}\|_{L^2_{yz}}
\\ \lesssim &
\left(\|c_t-6c_yx_y\|_{L^2}+\|x_t-2c-3(x_y)^2\|_{L^2}+\|x_{yy}\|_{L^2}\right)
\|e^{az}\tpsi_{c(t,y)}\|_{L^2_{yz}}\,.
\end{split}  
\end{equation}
By integration by parts and the fact that
$\|\pd_c\tpsi_{c(t,y)}(z)\|_{L^\infty_yL^2_z}\lesssim 1$,
\begin{equation}
\label{eq:l2pf5}
\begin{split}
\int_{\R^2}\ell_{21}\tpsi_{c(t,y)}(t,z,y)dzdy=
& \frac12 \frac{d}{dt}\int_{\R^2}\psi_{c(t,y)}(z)^2dzdy
+O(\|c_yx_y\|_{L^2}\|\tc\|_{L^2})\,.
\end{split}
\end{equation}
Similarly,
\begin{equation}
  \label{eq:l2pf6}
\begin{split}
\left|\int_{\R^2}\ell_{22}\tpsi_{c(t,y)}dzdy\right|
=& 3\left|\int_{\R^2}\left\{\varphi_{c(t,y)}'(z)-x_{yy}(t,y)\right\}
\tpsi_{c(t,y)}(z)^2  dzdy\right|
\\ \lesssim & \|e^{az}\tpsi\|_{L^2_{yz}}^2
+\|x_{yy}\|_{L^\infty}\|\psi_{c(t,y),L}\|_{L^2_{yz}}^2\,.
\end{split}  
\end{equation}
Since
$\|\ell_{13}\|_{L_z^\infty L_y^2}+\|\ell_{23}\|_{L_z^\infty L_y^2}
\lesssim \|c_{yy}\|_{L^2}+\|c_y\|_{L^4}^2$ and 
$\|\psi_{c(t,y),L}\|_{L^1_zL^2_y}=O(\|\tc\|_{L^2})$,
\begin{equation}
\label{eq:l2pf7}
\sum_{j=1,2}\left|\int_{\R^2}\ell_{j3}\tpsi_{c(t,y)}\,dzdy\right|
\lesssim \|\tc\|_{L^2}(\|c_{yy}\|_{L^2}+\|c_y\|_{L^4}^2)\,.
\end{equation}

In view of the definition of $\tpsi$,
\begin{equation}
\label{eq:psinorm}
\begin{split}
& \|\tpsi_{c(t,y)}\|_X \lesssim \|\tc\|_{L^2}e^{-a(4t+L)}\,,\\
& \|\tpsi_{c(t,y)}\|_{L^2(\R^2)}=2\sqrt{2}\|\sqrt{c}-\sqrt{2}\|_{L^2(\R)}
\|\psi\|_{L^2(\R)}\lesssim\|\tc\|_{L^2}\,.  
\end{split}
\end{equation}
Claims~\ref{cl:wA1-bound}, \ref{cl:b-capprox} and \eqref{eq:modeq} imply
that for $t\in[0,T]$,
\begin{align*}
& \|c_t\|_Y+\|x_t-2c-3(x_y)^2\|_{L^2} \lesssim \|b_t\|_Y+\|x_t-2c-3(x_y)^2\|_{L^2} 
\\ \lesssim & \|c_{yy}\|_Y+\|x_{yy}\|_Y
+\|\widetilde{\mathcal{A}}_1(t)\|_{B(Y)}\|b\|_Y+\|(bx_y)_y\|_Y
+\sum_{i=2}^8\|\cN_i\|_Y
\\ \lesssim & \bM_1(T)\la t\ra^{-3/4}+\sum_{i=2}^8\|\cN_i\|_Y\,.
\end{align*}
Following the proof of Lemma~\ref{lem:M1-bound}, we see that
$$\sum_{2\le i\le 8}\|\cN_i\|_Y\lesssim (\bM_1(T)+\bM_2(T))^2\la t\ra^{-1}\,.$$
Thus we have
\begin{equation}
  \label{eq:cx_t-bound}
\|c_t\|_Y+\|x_t-2c-3(x_y)^2\|_{L^2} 
\lesssim \bM_1(T)\la t\ra^{-3/4}+(\bM_1(T)+\bM_2(T))^2\la t\ra^{-1}\,.
\end{equation}
\par
Combining \eqref{eq:l2pf1}--\eqref{eq:l2pf7}, \eqref{eq:psinorm}
and \eqref{eq:cx_t-bound} with Lemma~\ref{lem:L2conserve},
we see that for $t\in(0,T]$,
\begin{equation}
  \label{eq:cons1}
\begin{split}
& \Bigl[Q(s,v)+8\|\psi\|_{L^2}^2\|\sqrt{c(s)}-\sqrt{2}\|_{L^2(\R)}^2
\Bigr]_{s=0}^{s=t}
\\ \lesssim & (\bM_1(T)+\bM_2(T))^2\int_0^t \la s\ra^{-5/4}\,ds
\lesssim (\bM_1(T)+\bM_2(T))^2\,.
\end{split}  
\end{equation}
Since $\|\sqrt{c(0)}-\sqrt2\|_{L^2}
\lesssim \|\tc(0)\|_Y\lesssim \|v_0\|_{X_1}$ and 
$$Q(t,v)=\|v(t)\|_{L^2(\R^2)}^2+O(\|\tc(t)\|_Y\|v(t)\|_{L^2(\R^2)})\,,$$
Lemma~\ref{lem:M4-bound} follows immediately from \eqref{eq:cons1}.
Thus we complete the proof.
\end{proof}
\bigskip

\section{Low frequencies bound of $v(t,x,y)$ in $y$}
\label{sec:lowbound}
Let $v_1(t)=P_1(0,2M)v(t)$. Since $v_1(t)$ does not include high frequency
modes in the $y$ variable, we can estimate $v_1(t)$ in the similar mannar
as generalized KdV equations (\cite{PW}) by using the semigroup
estimates obtained in Section~\ref{sec:semigroup}. In this section,
we will estimate $v_1(t)$ in the exponentially weighted space $X$.

\begin{lemma}
  \label{lem:nonresonant-ylow}
Let $\eta_0$, $a$ and $M$ be positive constants satisfying $\nu_0<a<2$
and $\nu(2M)>a$.
Suppose that $v(t)$ is a solution of \eqref{eq:v}.
Then there exist positive constants $b_1$, $\delta_5$ and $C$ such that
if $\bM_1(T)+\bM_2(T)<\delta_5$, then for $t\in[0,T]$,
\begin{equation*}
\left\|v_1(t,\cdot)\right\|_X \le  Ce^{-2b_1\eta_0^2t}\|v(0,\cdot)\|_X
+\left\{\bM_1(T)+\bM_2(T)\sum_{i=1}^3\bM_i(T)\right\}\la t\ra^{-3/4}\,.
\end{equation*}
\end{lemma}

Let $\chi(\eta)$ be a nonnegative smooth function such that
$\chi(\eta)=1$ if $|\eta|\le 1$ and $\chi(\eta)=0$ if $|\eta|\ge2$.
Let $\chi_M(\eta)=\chi(\eta/M)$ and
$$P_{\le M}u:=\frac{1}{2\pi}\int_{\R^2}\chi_M(\eta)
\hat{u}(\xi,\eta)e^{i(x\xi+y\eta)}d\xi d\eta,\quad P_{\ge M}=I-P_{\le M}.$$
To estimate $v_1(t)$, we need the following.
\begin{claim}
\label{cl:pv1}
There exists a positive constant $C$ such that
\begin{equation}
  \label{eq:pv1}
\|P_{\le M}u\|_{L^1_xL^2_y} \le C\sqrt{M}\|u\|_{L^1(\R^2)}\,.
\end{equation}
\end{claim}
\begin{proof}
Applying Young's inequality to
\begin{equation}
  \label{eq:Mink}
(P_{\le M}u)(x,y)=\frac{1}{\sqrt{2\pi}}
\int_\R\mF^{-1}(\chi_M)(y-y_1)u(x,y_1)dy_1\,,
\end{equation}
we have 
\begin{equation*}
\|P_{\le M}u\|_{L^2_y}\le  \|\mF^{-1}(\chi_M)\|_{L^2(\R)}
 \|u(x,\cdot)\|_{L^1(\R)} \lesssim  \sqrt{M}\|u(x,\cdot)\|_{L^1(\R)}\,.
\end{equation*}
Integrating the above over $\R$ in $x$, we obtain \eqref{eq:pv1}.
\end{proof}

\begin{proof}[Proof of Lemma~\ref{lem:nonresonant-ylow}]
Let $v_2(t)=P_2(\eta_0,M)v(t)$.
Then 
\begin{equation}
  \label{eq:v1n}
  \pd_tv_2=\mL v_2+P_2(\eta_0,2M)\{\ell+\pd_x(N_1+N_2+N_2')+N_3\}\,,
\end{equation}
where $N_2'=2\tc(t,y)v(t,z,y)+6(\varphi(z)-\varphi_{c(t,y)}(z))v(t,z,y)$.
Hereafter we abbreviate $P_2(\eta_0,2M)$ as $P_2$.
By Proposition~\ref{prop:semigroup-p2} and Corollary 
\ref{cor:semigroup-p2}, we may assume that
\begin{align*}
& \|e^{t\mL}P_2f\|_X\le  Ke^{-2b_1\eta_0^2t}\|f\|_X\,,\\
& \|e^{t\mL}P_2\pd_zf\|_X \le  K(1+t^{-1/2})e^{-2b_1\eta_0^2t}\|_X\,,\\
& \|e^{t\mL}P_2\pd_zf\|_X\le  K(1+t^{-3/4})e^{-2b_1\eta_0^2t}\|e^{az}f\|_{L^1_zL^2_y}\,,
\end{align*}
where $K$ and $b_1$ are positive constants independent of $f\in X$ and $t>0$.
Applying the semigroup estimates Lemma~\ref{prop:semigroup-p2} and
Corollary~\ref{cor:semigroup-p2}
to \eqref{eq:v1n}, we have
\begin{align*}
\|v_2(t)\|_X\lesssim & e^{-2b_1\eta_0^2t}\|v_2(0)\|_X
+\int_0^t e^{-2b_1\eta^2(t-s)}(t-s)^{-3/4}\|e^{az}N_1(s)\|_{L^1_zL^2_y}\,ds
\\ & +\int_0^t e^{-2b_1\eta^2(t-s)}(t-s)^{-1/2}(\|N_2(s)\|_X+\|N_2'(s)\|_X)\,ds
\\ & +\int_0^t e^{-2b_0\eta^2(t-s)}(\|\ell(s)\|_X+\|N_3(s)\|_X)\,ds\,.
\end{align*}
By Claim~\ref{cl:pv1},
\begin{align*}
\|e^{az}P_2N_1\|_{L^1_zL^2_y}\lesssim & \sqrt{M}\|v\|_{L^2}\|v\|_X
\\ \lesssim  & \sqrt{M}\bM_2(T)\bM_3(T)\la t\ra^{-3/4}\quad
\text{for $t\in[0,T]$.}
\end{align*}
By \eqref{eq:cx_t-bound},
we have for $t\in[0,T]$,
\begin{equation}
  \label{eq:ell1-est}
\begin{split}
\|\ell_1\|_X\lesssim & 
\|x_t-2c-3(x_y)^2\|_{L^2}+\|c_t-6c_yx_y\|_{L^2}+\|x_{yy}\|_{L^2}
+\|c_{yy}\|_{L^2}+\|c_y\|_{L^4}^2
\\ \lesssim & (\bM_1(T)+\bM_2(T)^2)\la t\ra^{-3/4}\,,
\end{split}  
\end{equation}
\begin{equation}
  \label{eq:ell2-est}
\begin{split}
\|\ell_2\|_X\lesssim & 
e^{-a(4t+L)}\bigl(\|c_t-6c_yx_y\|_Y+\|x_t-2c-3(x_y)^2\|_Y+\|\tc\|_Y+
\\ & +\|x_{yy}\|_Y +\|c_{yy}\|_Y+\|c_y\|_{L^4}^2\bigr)
\\ \lesssim & e^{-a(4t+L)}(\bM_1(T)+\bM_2(T)^2)\la t\ra^{-1/4}\,,
\end{split}  
\end{equation}
and
\begin{align*}
\|N_2\|_X+\|N_2'\|_X
\lesssim & (\|x_t-2c-3(x_y)^2\|_{L^\infty}+\|\tc\|_{L^\infty})\|v\|_X
\\ \lesssim & (\bM_1(T)+\bM_2(T))\bM_2(T)\la t\ra^{-5/4}\,.
\end{align*}
Here we use $\sup_{y,z}(|\varphi_{c(t,y)}(z)-\varphi(z)|+|\tpsi_{c(t,y)}|)\lesssim 
\|\tc(t)\|_{L^\infty}$.
Since $\|\pd_yP_2\|_{B(X)}\lesssim M$,
\begin{align*}
\|P_2N_3\|_X\lesssim & M(\|x_y\|_{L^\infty}+\|x_{yy}\|_{L^\infty})
\|v\|_X
\\ \lesssim & M\bM_1(T)\bM_2(T)\la t\ra^{-5/4} \quad\text{for $t\in[0,T]$.}
\end{align*}

As long as  $v(t)$ satisfies the orthogonality condition \eqref{eq:orth} and
$\tc(t,y)$ remains small, we have
\begin{equation}
  \label{eq:v1-v1n}
 \|v_1(t)-v_2(t)\|_X
\lesssim \sup_y|\tc(t,y)|\|v_1(t)\|_X\,,
\end{equation}
and $\frac12\|v_2(t)\|_X\le \|v_1(t)\|_X\le 2\|v_2(t)\|_X$.
Thus we have for $t\in[0,T]$,
\begin{align*}
\|v_1(t)\|_X \lesssim & e^{-2b_1\eta_0^2t}\|v(0)\|_X
+\bM_1(T)\int_0^t e^{-2b_1\eta_0^2(t-s)}\la s\ra^{-3/4}ds
\\ & + \bM_2(T) \sum_{i=1}^3\bM_i(T)
\int_0^t e^{-2b_1\eta_0^2(t-s)}\left\{1+(t-s)^{-3/4}\right\}\la s\ra^{-3/4}ds
\\ & \lesssim e^{-2b_1\eta_0^2t}\|v(0)\|_X
+\left\{\bM_1(T)+\bM_2(T) \sum_{i=1}^3\bM_i(T)\right\}\la t\ra^{-3/4}\,.
\end{align*}
Thus we complete the proof.
\end{proof}
\bigskip

\section{Virial estimates}
\label{sec:virial}
If we apply the argument in Section~\ref{sec:lowbound} to $P_{\ge M}v(t)$,
it requires boundedness of $\|v(t)\|_{L^p(\R^2)}$ with $p>2$, which remains
unknown even for small solutions around $0$.
Instead of the semigroup estimate in Section~\ref{sec:semigroup},
we will make use of a virial estimates of $v$ in the exponentially weighted
space.  We remark that the virial estimate for $L^2$-solutions
to the KP-II equation \eqref{KPII_integrated} was shown in \cite{dBM}.

\begin{lemma}
\label{lem:virial}
Let $a\in(0,2)$ and $v$ be a solution to \eqref{eq:v}. 
There exist positive constants $\delta_6$, $M$ and $C$ such that
if $\sum_{i=1}^3\bM_i(T)<\delta_6$, then
\begin{equation*}
\|v(t)\|_X^2 \le  e^{-2at}\|v(0)\|_X^2
+C\int_0^t e^{-2a(t-s)}\left(\|\ell(s)\|_X^2+\|P_{\le M}v(s)\|_X^2\right)\,ds\,.
\end{equation*}
\end{lemma}
To prove Lemma~\ref{lem:virial}, we use the following.
\begin{claim}
\label{cl:aniso}
Let  $a>0$ and $p_n(x)=e^{2an}(1+\tanh a(x-n))$.
There exists a $C>0$ such that for every $n\in\N$
\begin{equation}
  \label{eq:aniso}
\left(\int_{\R^2}p_n(x)u^4(x,y)\,dxdy\right)^{1/2} \le C
\int_{\R^2}p_n'(x)\Big((\pd_x u)^2+(\pd_x^{-1}\pd_y u)^2+u^2\Big)(x,y)\,dxdy
\,.
\end{equation}
\end{claim}
Claim~\ref{cl:aniso} follows in exactly the same way as
\cite[Lemma~2]{MST_KPI} and \cite[Claim 5.1]{MT}. So we omit the proof.

\begin{proof}[Proof of Lemma~\ref{lem:virial}]
Let $p_n$ be as in Claim~\ref{cl:aniso}. Then
$p_n(z)\uparrow e^{2az}$ and $p_n'(z)\uparrow 2ae^{2az}$
as $n\to\infty$ and $0<p_n'(z)\le ap_n(z)$,
$|p_n'''(z)|\le 4a^2p_n'(z)$ and $ap_n(z)^2=e^{2az}p_n'(z)$ for $z\in\R$.
\par
First, we will derive a virial identity for $v(t)$
assuming $v_0\in H^3(\R^2)$ and $\pd_x^{-1}v_0\in H^2(\R^2)$ so that
$v(t)\in C([0,T];H^3(\R^2))$ and $\pd_x^{-1}v(t)\in C([0,T];H^2(\R^2))$.
Multiplying \eqref{eq:v} by $2e^{2at}p_n(z)v(t,z,y)$ and integrating
the resulting equation by part, we have for $t\in[0,T]$,
\begin{equation}
  \label{eq:vir1}
  \begin{split}
&  \frac{d}{dt}\left(e^{2at}
\int_{\R^2} p_n(z)v(t,z,y)^2\,dzdy\right) 
+e^{2at}\int_{\R^2}p_n'(z)\left(\mathcal{E}(v)-4v^3\right)(t,z,y)\,dzdy
\\ =& e^{2at}\left\{
\int_{\R^2} 2ap_n(z)v(t,z,y)^2\,dzdy+\sum_{k=1}^3III_k(t)\right\}\,,
  \end{split}  
\end{equation}
where $\mathcal{E}(v)=3(\pd_zv)^2+3(\pd_z^{-1}\pd_yv)^2+4v^2$,
\begin{align*}
III_1=&2 \int_{\R^2}p_n(z)\ell v(t,z,y)\,dzdyds\,,\\
III_2=&-\int_{\R^2}p_n'(z)\left((\tx_t(t,y)-3x_y(t,y)^2\right)v(t,z,y)^2\,dzdy\,,
\\  III_3=&
\int_{\R^2}\biggl\{p_n'''(z)+6[\pd_z, p_n(z)]
\left(\varphi_{c(t,y)}(z)-\psi_{c(t,y),L}(z+4t)\right)\biggr\}v(t,z,y)^2\,dzdy\,.
\end{align*}
Integrating \eqref{eq:vir1} over $[0,t]$, we have
\begin{equation}
  \label{eq:vir2}
  \begin{split}
&  e^{2at}\int_{\R^2} p_n(z)v(t,z,y)^2\,dzdy
+\int_0^t e^{2as}\int_{\R^2}p_n'(z)\left(\mathcal{E}(v)-4v^3\right)(s,z,y)\,dzdyds
\\ =&\int_{\R^2} p_n(z)v(0,z,y)^2\,dzdy
+ \int_0^t e^{2as}\int_{\R^2} 2ap_n(z)v(s,z,y)^2\,dzdyds
\\ & +
\int_0^t e^{2as}\int_{\R^2} \left\{III_1(s)+III_2(s)+III_3(s)\right\}\,ds\,,
  \end{split}    
\end{equation}
We can prove \eqref{eq:vir2}
for any $v(t)\in C([0,T];L^2(\R^2))\cap L^\infty([0,T];X)$ satisfying
$\sum_{i=1}^3\bM_i(T)<\delta_6$ in the same way as the proof
of Lemma~\ref{lem:L2conserve}.
\par

By the Schwarz inequality and Claim~\ref{cl:aniso},
\begin{equation}
  \label{eq:vir3}
\begin{split}
\left|\int p_n'(z)v(t,z,y)^3dzdy\right| \le &
\|v(t)\|_{L^2}\left(\int_{\R^2}p_n'(z)^2v(t,z,y)^4dzdy\right)^{1/2}
\\ \lesssim & \|v(t)\|_{L^2}\int_{\R^2}p_n'(z) \mathcal{E}(v(t,z,y))dzdy\,.
\end{split}  
\end{equation}
By the Schwarz inequality,
$$\left|III_1\right| \le 
\int p_n'(z) v^2\,dzdy +\int \frac{p_n(z)^2}{p_n'(z)} \ell^2 \,dzdy\,.$$
Since $Y\subset H^1(\R)$, we have
$\sup_{t\in[0,T]\,,\,y\in\R}|\tx_t(t,y)-3x_y(t,y)^2| \lesssim \bM_1(T)$
from \eqref{eq:cx_t-bound} and
$$\left|III_2\right| \lesssim \bM_1(T)\int_{\R^2} p_n(z) v(t,z,y)^2\,dzdy\,.$$
Let 
$$M=\sup_{n,z}\frac{|p_n'''(z)|}{p_n'(z)}
+6\sup_{n,t,y,z}\frac{\left|[\pd_z, p_n(z)]
\left(\varphi_{c(t,y)}(z)-\psi_{c(t,y),L}(z+4t)\right)\right|}{p_n'(z)}\,.$$
Then
$$\left|III_3\right| \le M\int_{\R^2} p_n(z) v(t,z,y)^2\,dzdy\,.$$
Let $v_{<}=P_{\le M}v$ and $v_{>}=P_{\ge M}v$.
For $y$-high frequencies, the potential term can be absorbed into the
left hand side. Indeed it follows from
Plancherel's theorem and the Schwarz inequality that
\begin{align*}
& \int_{\R^2}
p_n'(z)\left((\pd_zv_{>})^2+(\pd_z^{-1}\pd_yv_{>})^2\right)(t,z,y)\,dzdy
\\ =& \int_{\R^2}p_n'(z)\left(|\pd_z\mF_y(v_{>})|^2
+\eta^2|\pd_z^{-1}\mF_y(v_{>})|^2\right)(t,z,\eta)\,dzd\eta
\\ \ge & 2M\int_{\R^2}p_n'(z)v_{>}(t,z,y)^2\,dzdy\,.
\end{align*}
Combining the above, we have for $t\in[0,T]$,
  \begin{align*}
& e^{2at}\int_\R p_n(z)v(t,z,y)^2dzdy \le \int_\R p_n(z)v(0,z,y)^2\,,dzdy
\\ & +\int_0^t e^{2as}\frac{p_n(z)^2}{p_n'(z)}\ell(s)^2\,dzdyds
+M\int_0^t e^{2as}p_n'(z)v_<(s,z,y)^2\,dzdyds
  \end{align*}
if $\delta_6$ is sufficiently small.
By passing to the limit as $n\to\infty$, we obtain Lemma~\ref{lem:virial}.
Thus we complete the proof.
\end{proof}

Combining Lemmas~\ref{lem:nonresonant-ylow} and \ref{lem:virial},
we obtain the following.
\begin{lemma}
  \label{lem:exp-bound}
Let $a$ and $M$ be as in Lemmas~\ref{lem:nonresonant-ylow} and
\ref{lem:virial}.
There exist positive constants $\delta_7$ and $C$ such that if
$\sum_{i=1}^3\bM_i(T)\le  \delta_7$, then
\begin{equation}
  \label{eq:M4-bound}
\bM_2(T)\le C(\|v_0\|_X+\bM_1(T))\,.
\end{equation}
\end{lemma}

\begin{proof}
Since $\chi_M(\eta)=0$ for $\eta\in\R\setminus[-2M,2M]$, we have
$\|P_{\le M}v(t)\|_X\le \|v_1(t)\|_X$. 
Combining Lemma~\ref{lem:virial} with Lemma~\ref{lem:nonresonant-ylow},
\eqref{eq:ell1-est} and \eqref{eq:ell2-est},
we have for $t\in[0,T]$,
$$\|v(t)\|_X \lesssim
e^{-b_2t}\|v(0)\|_X+\left\{
\bM_1(t)+\bM_2(T)(\bM_1(T)+\bM_2(T)+\bM_3(T))\right\}\la t\ra^{-3/4}\,.$$
Since $\|v(0)\|_X\lesssim \|v_0\|_X$ by Lemma~\ref{lem:decomp},
we obtain \eqref{eq:M4-bound} if $\delta_7$ is sufficiently small.
Thus we complete the proof.
\end{proof}
\bigskip

\section{Proof of Theorem~\ref{thm:exp-stability}}
\label{sec:thm1}
Now we are in position to complete the proof of Theorem~\ref{thm:exp-stability}.
\begin{proof}
Since the KP-II equation has the scaling invariance,
we may assume that $c_0=2$ without loss of generality.
Let $\delta_*=\min_{0\le i\le 7}\delta_i/2$.
\par
Since $v_0\in H^1(\R^2)\cap X$, 
$$\tv(t,x,y)=u(t,x+4t,y)-\varphi(x)\in C([0,\infty);X\cap H^1(\R^2))$$
(see \cite{MST} and Proposition~\ref{prop:LWPX}).
If $\|v_0\|_X+\|v_0\|_{L^2}$ is sufficiently small,
Lemma~\ref{lem:decomp} and Remark~\ref{rem:decomp} imply
that there exists $T>0$ and $(c(t),x(t))$ satisfying
\eqref{eq:decomp}, \eqref{eq:orth}, \eqref{eq:C1cx} and
$$\|\tc(t)\|_Y+\|\tx(t)\|_Y\lesssim \|\tv(t)\|_X
\quad\text{for $t\in[0,T]$,}$$
and it follows that  $v(t)\in C([0,T]; X\cap L^2(\R^2))$ and
\begin{equation}
  \label{eq:tot-1}
\bM_{tot}(T):=\bM_1(T)+\bM_2(T)+\bM_3(T)\le \frac{\delta_*}{2}\,.
\end{equation}
By Proposition~\ref{prop:continuation}, we can extend the
decomposition \eqref{eq:decomp} satisfying \eqref{eq:orth} beyond
$t=T$. Let $T_1\in(0,\infty]$ be the maximal time such that the
decomposition \eqref{eq:decomp} with \eqref{eq:orth} exists for
$t\in[0,T_1]$ and $\bM_{tot}(T_1)\le \delta_*$.  Suppose $T_1<\infty$.
Then it follows from Lemmas~\ref{lem:M1-bound}, \ref{lem:M4-bound} and
\ref{lem:exp-bound} that
\begin{equation}
  \label{eq:tot-2}
\bM_{tot}(T_1)\lesssim \|v_0\|_{X_1}+\|v_0\|_{L^2(\R^2)}+\bM_{tot}(T_1)^2\,.  
\end{equation}
If $\|v_0\|_{X_1}+\|v_0\|_{L^2(\R^2)}$ is sufficiently small,
then $\bM_{tot}(T_1)\le \delta_*/2$ follows from \eqref{eq:tot-2},
which contradicts to the definition of $T_1$.
Thus we prove $T_1=\infty$ and
\begin{equation}
  \label{eq:tot-3}
  \bM_{tot}(\infty)\lesssim \|v_0\|_{X_1}+\|v_0\|_{L^2}\,.
\end{equation}
\par
Now we will prove \eqref{OS} and \eqref{AS}.
By \eqref{eq:decomp}, \eqref{eq:psinorm} and \eqref{eq:tot-3},
\begin{align*}
\|u(t,x,y)-\varphi_{c(t,y)}(x-x(t,y))\|_{L^2(\R^2)}\le &
\|v(t)\|_{L^2(\R^2)}+\|\tpsi_{c(t,y)}\|_{L^2(\R^2)}
\\ \lesssim &  \bM_3(\infty)+\bM_1(\infty)\,,
\end{align*}
\begin{multline*}
\left\|e^{ax}(u(t,x+x(t,y),y)-\varphi_{c(t,y)}(x))\right\|_{L^2}
\le  \|v(t)\|_X+ \|\tpsi_{c(t,y)}\|_X
\\ \lesssim  
\bM_2(\infty)\la t\ra^{-3/4}+\bM_1(\infty)e^{-a(4t+L)}\la t\ra^{-1/4}\,. 
\end{multline*}
Since $\|f\|_{L^\infty}\lesssim \|f\|_Y^{1/2}\|\pd_yf\|_Y^{1/2}$ 
for any $f\in Y$, we see that \eqref{phase-sup} and \eqref{phase2}
follow immediately from \eqref{eq:tot-3} and \eqref{eq:cx_t-bound}.
Thus we complete the proof of Theorem~\ref{thm:exp-stability}.
\end{proof}
\bigskip

\section{Proof of Theorem~\ref{thm:instability}}
\label{sec:thm2}
In this section, we will prove orbital instability of line solitons.
For the purpose, we will utilize that $(b,x_y)$ is a solution to
the diffusion wave equation \eqref{eq:modeq2} and its profile can be
approximated by the heat kernel in some region.
\par

\begin{proof}[Proof of Theorem~\ref{thm:instability}]
First we remark that if $\|u(t,x)-\varphi_{c_0}(x-x_0)\|_{L^2(\R^2)}<\infty$,
then $x_0=2c_0t$. Indeed, it follow from \cite{MST} that
$u(t,x,y)-\varphi_{c_0}(x-2c_0t)\in L^2(\R^2)$ for every $t\ge0$ and
\begin{align*}
& \|\varphi_{c_0}(x-2c_0t)-\varphi_{c_0}(x-x_0)\|_{L^2(\R^2)}
\\ \le & \|u(t,x,y)-\varphi_{c_0}(x-x_0)\|_{L^2(\R^2)}
+\|u(t,x,y)-\varphi_{c_0}(x-2c_0t)\|_{L^2(\R^2)}<\infty\,,
\end{align*}
whereas $\|\varphi_{c_0}(\cdot-2c_0t)-\varphi(\cdot-x_0)\|_{L^2(\R^2)}=\infty$
if $x_0\ne 2c_0t$.
\par
On the other hand, Theorem~\ref{thm:exp-stability} implies that 
\begin{align*}
& \|u(t,\cdot)-\varphi_{c_0}(x-2c_0t)\|_{L^2(\R^2)} \\
\ge & \|\varphi_{c_0}(x-x(t,y))-\varphi_{c_0}(x-2c_0t)\|_{L^2(\R^2)}
-\|u(t,x,y)-\varphi_{c(t,y)}(x-x(t,y))\|_{L^2(\R^2)}
\\ & -\|\varphi_{c(t,y)}(x-x(t,y))-\varphi_{c_0}(x-x(t,y))\|_{L^2(\R^2)}
\\ \gtrsim & \|x(t,y)-2c_0t\|_{L^2(\R)}-O(\eps)\,.
\end{align*}
Thus to prove orbital instability of line solitons, it suffices to
show that $\|x(t,\cdot)\|_{L^2(\R)}$ grows up as $t\to\infty$.
\par
Now we will construct a solution satisfying $\|x(t,\cdot)\|_Y\gtrsim t^{1/4}$
as $t\to\infty$. We may assume that $c_0=2$ without loss of generality.
If $b(0)$ and $x_y(0)$ are sufficiently small and $\int_{\R}b(0)\,dy$ is nonzero,
then  $e^{tA_0}(b(0),x_y(0))$ is expected to be the main part of
the solution $(b(t),x_y(t))$. To investigate the behavior of
$e^{tA_0}(b(0),x_y(0))$, we represent the semigroup $e^{tA_0}$ by using 
the heat kernel $H_t(y)$. Let
$$A_{0,1}(\eta)=
\begin{pmatrix}-3\eta^2 & 8i\eta \\ i\eta(2+\frac18\eta^2) & -\eta^2
\end{pmatrix}\,,\quad
A_{0,2}(\eta)=\frac{1}{\eta^3}\left(A_0(\eta)-A_{0,1}(\eta)\right)
\,.$$
Note that $A_{0,1}$ is equal to $A_*$ in Lemma~\ref{lem:mu=1/8}
and that $A_{0,2}
=\begin{pmatrix} O(\eta) & O(1) \\ O(1) & O(\eta)\end{pmatrix}$.
\par
Let $U_1(t,s)$ be the $2\times2$ matrix such that 
$$\pd_tU_1(t)=A_1(t,0)U_1(t)\,,\quad \lim_{t\to\infty}U_1(t)=I\,.$$
Since $|A_1(t,0)|\lesssim e^{-a(4t+L)}$ for $t\ge0$, we have
$\sup_{\tau \ge t}|U_1(\tau)-I|\lesssim e^{-a(4t+L)}$.
Now let 
$$\begin{pmatrix}b(t) \\ x_y(t)\end{pmatrix}
=U_1(t)\begin{pmatrix}b_1(t)\\ b_2(t)  \end{pmatrix}\,.$$
Then
\begin{equation}
  \label{eq:b1b2}
\pd_t\begin{pmatrix} b_1(t) \\ b_2(t)\end{pmatrix}=
(A_0(D_y)+D_yA_2(t,D_y))\begin{pmatrix} b_1(t) \\ b_2(t)\end{pmatrix}
+\sum_{i=1}^8 U_1(t)^{-1}\diag(1,\pd_y)\cN_i\,,
\end{equation}
where $A_2(t,\eta)=A_{21}(t,\eta)+A_{22}(t,\eta)$ and
\begin{gather*}
A_{21}(t,\eta)=\frac{U_1(t)^{-1}A_0(\eta)U_1(t)-A_0(\eta)}{\eta}\,,\\  
A_{22}(t,\eta)=U_1(t)^{-1}\frac{A_1(t,\eta)-A_1(t,0)}{\eta}U_1(t)\,.
\end{gather*}
Clearly, we have $\|A_2(t,D_y)\|_{B(Y)}\lesssim e^{-a(4t+L)}$.
By the variation of constants formula,
\begin{equation*}
\begin{pmatrix} b_1(t) \\ b_2(t) \end{pmatrix}
= e^{tA_{0,1}}\begin{pmatrix} b_1(0) \\ b_2(0) \end{pmatrix}
+IV_1+IV_2+IV_3\,,
\end{equation*}
where
\begin{align*}
IV_1=& \int_0^t e^{(t-s)A_{0,1}}D_y^3A_{0,2}
\begin{pmatrix}  b_1(s) \\ b_2(s) \end{pmatrix}\,ds\,,\\
IV_2=& \int_0^t e^{(t-s)A_{0,1}}D_yA_2(s,D_y)
\begin{pmatrix}  b_1(s) \\ b_2(s) \end{pmatrix}\,ds\,,\\
IV_3=& \sum_{i=1}^8 \int_0^t e^{(t-s)A_{0,1}}U_1(s)^{-1}\diag(1,\pd_y)\cN_i(s)\,ds\,.
\end{align*}
Let $h\in C_0^\infty(-\eta_0,\eta_0)$ such that $h(0)=1$ and let
$$u(0,x,y)=\varphi_{2+c_*(y)}(x)-\psi_{2+c_*(y),L}(x)\,,\quad
c_*(y)=2+\eps (\mathcal{F}_\eta^{-1}h)(y)\,.$$
Then it follows from Lemma~\ref{lem:decomp} that
$\tc(0,y)=c_*(y)$, $x(0,y)\equiv0$ and $v(0,\cdot)=0$.
Since $\|b(0)-\tc(0)\|_{Y_1}\lesssim \|\tc(0)\|_Y^2$
by Claim~\ref{cl:b-capprox} and $\|U_1(t)-I\|_{B(Y)}\lesssim e^{-a(4t+L)}$,
\begin{equation}
  \label{eq:bcdiff}
\left\|\begin{pmatrix}b_1(0)\\ b_2(0)\end{pmatrix}
-\begin{pmatrix}c_*\\ 0\end{pmatrix}\right\|_{Y_1}=  
\left\|U_1(0)^{-1}\begin{pmatrix}b(0)\\ 0\end{pmatrix}
-\begin{pmatrix}\tc(0)\\ 0\end{pmatrix}\right\|_{Y_1}
\lesssim \eps(\eps+e^{-aL})\,.
\end{equation}
Since $\|e^{tA_{0,1}}\|_{B(Y_1,Y)}\lesssim (1+t)^{-1/4}$ by
Lemma~\ref{lem:mu=1/8}, it follows from \eqref{eq:bcdiff} that
$$\left\|e^{tA_{0,1}}\begin{pmatrix}b_1(0) \\ b_2(0)\end{pmatrix}
-e^{tA_{0,1}}\begin{pmatrix}c_* \\ 0\end{pmatrix}\right\|_{Y}
\lesssim \eps(\eps+e^{-aL})(1+t)^{-1/4}\,.$$
By Corollary~\ref{cor:mu=1/8},
\begin{align*}
\left\|e^{tA_{0,1}}\begin{pmatrix}c_* \\ 0\end{pmatrix}-
\frac14  \begin{pmatrix}
2\left(e^{4t\pd_y}(H_{2t}*c_*)+e^{-4t\pd_y}(H_{2t}*c_*)\right) \\
e^{4t\pd_y}(H_{2t}*c_*)-e^{-4t\pd_y}(H_{2t}*c_*)
  \end{pmatrix}\right\|_Y\lesssim \eps\la t\ra^{-3/4}\,.
\end{align*}
Since $c_*=\eps\mF^{-1}h$ and $h(0)=1$, it follows from Plancherel's theorem
that
\begin{multline*}
\left\|\wP_1\{e^{\pm 4t\pd_y}(H_{2t}*c_*)-\eps e^{\pm 4t\pd_y}H_{2t}\}\right\|_Y
= \eps\left\|e^{-2t\eta^2}(h(\eta)-h(0))\right\|_{L^2(-\eta_0,\eta_0)}
\\ \lesssim  \eps\sup_\eta|h'(\eta)| \|\eta e^{-2t\eta^2}\|_{L^2(-\eta_0,\eta_0)}
 \lesssim  \eps \la t\ra^{-3/4}\,.
\end{multline*}
Thus we have
\begin{equation}
  \label{eq:IV-0}
  \begin{split}
& \left\|e^{tA_{0,1}}
\begin{pmatrix}b_1(0)\\b_2(0)\end{pmatrix}
-\frac{\eps}{4}\wP_1
\begin{pmatrix}2\left(e^{4t\pd_y}H_{2t}+e^{-4t\pd_y}H_{2t}\right)
\\ e^{4t\pd_y}H_{2t}-e^{-4t\pd_y}H_{2t} \end{pmatrix}
\right\|_Y
\\ \lesssim & 
\eps(\eps+e^{-aL})\la t\ra^{-1/4}+\eps\la t\ra^{-3/4}\,.
  \end{split}
\end{equation}
\par
In view of the proof of Theorem~\ref{thm:exp-stability}, we have
for $k=0$ and $1$,
$$\sup_{t\ge0} \la t\ra^{(2k+1)/4}(\|b_1(t)\|_Y+\|b_2(t)\|_Y)
\lesssim \bM_1(\infty) \lesssim \eps\,.$$
By Lemma~\ref{lem:mu=1/8} and Claim~\ref{cl:abc},
\begin{equation}
  \label{eq:IV-1}
\begin{split}
\|IV_1\|_Y\lesssim & \int_0^t \|\pd_y^2e^{(t-s)A_{0,1}}\|_{B(Y)}
\left(\|\pd_yb_1(s)\|_Y+\|\pd_yb_2(s)\|_Y\right)\,ds
\\ \lesssim & \bM_1(t)\int_0^t\la t-s\ra^{-1}\la s\ra^{-3/4}\,ds
\lesssim  \eps\la t\ra^{-3/4}\log\la t\ra\,,
\end{split}  
\end{equation}
and
\begin{equation}
  \label{eq:IV-2}
\begin{split}
\|IV_2\|_Y\lesssim & \int_0^t \|\pd_ye^{(t-s)A_{0,1}}\|_{B(Y)}
\|A_2(s,D_y)\|_{B(Y)}\left(\|b_1(s)\|_Y+\|b_2(s)\|_Y\right)\,ds
\\ \lesssim & \bM_1(t)\int_0^t\la t-s\ra^{-1/2}e^{-a(4s+L)}\la s\ra^{-1/4}\,ds
\lesssim \eps\la t\ra^{-1/2}\,.
\end{split}  
\end{equation}
Using Lemma~\ref{lem:mu=1/8}, we can prove 
\begin{equation}
  \label{eq:IV-3}
\|IV_3\|_Y \lesssim   \la t\ra^{-1/4}(\bM_1(\infty)^2+\bM_2(\infty)^2)
\lesssim  \eps^2\la t\ra^{-1/4}
\end{equation}
in exactly the same way as the proof of Lemma~\ref{lem:M1-bound}.
Combining \eqref{eq:IV-0}--\eqref{eq:IV-3},  we have
$$\left\| \begin{pmatrix}b_1(t) \\ b_2(t)\end{pmatrix}
-\frac{\eps}{4}\wP_1\begin{pmatrix}
2\left(e^{4t\pd_y}H_{2t}+e^{-4t\pd_y}H_{2t}\right)\\
e^{4t\pd_y}H_{2t}-e^{-4t\pd_y}H_{2t}\end{pmatrix}\right\|_Y
\lesssim \eps(\eps+e^{-aL})\la t\ra^{-1/4}+\eps\la t\ra^{-1/2}\,.$$
Since $|U_1(t)-I|\lesssim e^{-a(4t+L)}$ and $\|(I-\wP_1)H_{2t}\|_{L^2}\lesssim e^{-2t\eta_0^2}$,
\begin{equation}
  \label{eq:x_y}
\left\|x_y(t)-\frac{\eps}{4}\left(e^{4t\pd_y}H_{2t}-e^{-4t\pd_y}H_{2t}\right)
\right\|_{L^2}
\lesssim \eps(\eps+e^{-aL})\la t\ra^{-1/4}+\eps\la t\ra^{-1/2}\,.
\end{equation}
\par
Now let $d_1$ and $d_2$ be constants satisfying $d_2>d_1>1$ and
let $y_1\in [-4t+4(d_1-1)\sqrt{t},-4t+4d_1\sqrt{t}]$, $y_2\in [-4t+4d_2\sqrt{t},-4t+4(d_2+1)\sqrt{t}]$ for $t\ge0$.
By \eqref{eq:x_y}, there exist positive constants $C_1$, $C_2$ and $t_1$ such that
\begin{align*}
x(t,y_2)-x(t,y_1)=& \int_{y_1}^{y_2}x_y(t,\tilde{y})\,d\tilde{y}
\\ \ge & \frac\eps4\int_{y_1}^{y_2}\left(H_{2t}(\tilde{y}+4t)-H_{2t}(\tilde{y}-4t)\right)\,d\tilde{y}
\\ & -(y_2-y_1)^{1/2}\left\|x_y(t)-\frac\eps4
\left(e^{4t\pd_y}H_{2t}-e^{-4t\pd_y}H_{2t}\right)\right\|_{L^2}
\\ \gtrsim & \eps\int_{d_1\sqrt{t}}^{d_2\sqrt{t}}H_{2t}(y)\,dy-C_1\eps(\eps+e^{-aL}+\la t\ra^{-1/4})
\\=& \eps\left(\erf(\sqrt{2}d_2)-\erf(\sqrt{2}d_1)\right)-C_1\eps(\eps+e^{-aL}+\la t\ra^{-1/4})\,,
\\ \ge & C_2\eps\quad\text{for $t\ge t_1$,}
\end{align*}
if $\eps$ and $e^{-aL}$ are sufficiently small.  Recall that
$\erf(x)=\frac{2}{\sqrt{\pi}}\int_0^x e^{-x_1^2}\,dx_1$.  Since $L$ is
an auxiliary parameter introduced in Section~\ref{sec:decomp} which can
be chosen arbitrary large, we see that $|x(t,y)|\ge C_2\eps/2$ either
on $[-4t+4(d_1-1)\sqrt{t},-4t+4d_1\sqrt{t}]$ or on
$[-4t+4d_2\sqrt{t},-4t+4(d_2+1)\sqrt{t}]$.  Therefore
$\|x(t)\|_Y\gtrsim \eps\la t\ra^{1/4}$. Thus we complete the proof.
\end{proof}
\bigskip

\section{Proof of Theorem~\ref{thm:Burgers}}
\label{sec:thm3}
In order to prove Theorem~\ref{thm:Burgers}, we will show that
the first order asymptotics of solutions to \eqref{eq:b1b2} 
around $y=\pm4t+O(\sqrt{t})$ is given by a sum of self-similar
solutions to the Burgers equations.
We apply the scaling argument by Karch (\cite{Karch}) to obtain the asymptotics
of \eqref{eq:b1b2} and use a virial type estimate to show that
interaction between $b_1(t,y)$ and $b_2(t,y)$ tends to $0$ around
$y=\pm4t+O(\sqrt{t})$ as $t\to\infty$.
Since $\sup_{t>0}t^{1/4}(\|b_1(t)\|_{L^2(\R)}+\|b_2(t)\|_{L^2(\R)})\ll1$,
we have the uniqueness of the limiting profile.
\par
Roughly speaking, a solution of \eqref{eq:b1b2} can be decomposed into
two parts that move to the opposite direction.
Now we recenter each component of solutions to \eqref{eq:b1b2}
and diagonalize the equations.  Let $A_*(\eta)$, $\Pi_*(\eta)$
and $\omega(\eta)$ be as \eqref{eq:A*} with $\mu=\mu_3$.
By the change of variables
\begin{gather*}
\bb(t,y)=\begin{pmatrix}b_1(t,y) \\ b_2(t,y) \end{pmatrix}\,,
\quad \Pi_1(t,\eta)=\frac{1}{4i}\Pi_*(\eta)\diag(e^{4it\eta},e^{-4it\eta})\,,\\
\bd(t,y)=\mathcal{F}_\eta^{-1}\Pi_1(t,\eta)^{-1}
(\mathcal{F}_y\bb)(t,\eta)\,,  
\end{gather*}
we have
\begin{equation}
  \label{eq:bd}
\pd_t\bd=\{2\pd_y^2I+\pd_y(A_3(t,D_y)+A_4(t,D_y))\}\bd
+A_5(t,D_y)\sum_{i=1}^8\diag(1,\pd_y)\cN_i\,,
\end{equation}
where
\begin{gather*}
A_3(t,\eta)=(\omega(\eta)-4)\begin{pmatrix}1 & 0 \\ 0 & -1\end{pmatrix}
-i\eta^{-1}\Pi_1(t,\eta)^{-1}(A_0(\eta)-A_*(\eta))\,,\\
A_4(t,\eta)=-i\Pi_1(t,\eta)^{-1}A_2(t,\eta)\Pi_1(t,\eta)\,,
\quad A_5(t,\eta)=\Pi_1(t,\eta)^{-1}U_1(t)^{-1}\,.
\end{gather*}
To detect the dominant part of the equation,
let us consider the rescaled solution
$\bd_\lambda(t,y)=\lambda\mathbf{d}(\lambda^2t,\lambda y)$.
Our aim is to find a self-similar profile $\bd_\infty(t,y)$ such that
\begin{equation}
\label{eq:ss-bd}
\lambda\bd_\infty(\lambda^2t,\lambda y)=\bd_\infty(t,y)\,,  
\end{equation}
and that for any $t_1$ and $t_2$ satisfying $0<t_1\le t_2<\infty$ and 
any $R>0$,
\begin{equation}
\label{eq:bd-c}
\lim_{\lambda\to\infty}\sup_{t\in[t_1,t_2]}
\|\bd_\lambda(t,y)-\bd_\infty(t,y)\|_{L^2(|y|<R)}=0\,.
\end{equation}
If \eqref{eq:ss-bd} and \eqref{eq:bd-c} hold, then letting
$\lambda=t^{1/2}\to\infty$, we have 
\begin{equation}
  \label{eq:conv}
  \begin{split}
t^{1/4} \|\bd(t,\cdot)-\bd_\infty(t,\cdot)\|_{L^2(|y|<R\sqrt{t})}
=& \lambda^{1/2}
\|\bd(\lambda^2,\cdot)-\bd_\infty(\lambda^2,\cdot)\|_{L^2(|y|<\lambda R)}
\\=& \|\bd_\lambda(1,\cdot)-\bd_\infty(1,\cdot)\|_{L^2(|y|<R)}\to0\,.
  \end{split}
\end{equation}
To prove \eqref{eq:bd-c}, we need the upper bounds of
$\bd_\lambda(t)$ for $k=0$, $1$ that do not depend on $\lambda\ge1$.
\begin{lemma}
\label{lem:bd-bound}
Let $\eps$ be as in Theorem~\ref{thm:exp-stability}. 
Then there exists a positive constants $C$ such that
for any $\lambda\ge1$ and $t\in(0,\infty)$,
\begin{gather}
  \label{eq:bd-bound1}
\sum_{k=0,1}\|\pd_y^k\bd_\lambda(t,\cdot)\|_{L^2}\le C\eps t^{-(2k+1)/4} \,,
\quad \|\pd_y^2\bd_\lambda(t,\cdot)\|_{L^2}\le C\eps\lambda^{1/2}t^{-1}\,,
\\  \label{eq:bd-bound2}
\|\pd_t\bd_\lambda(t,\cdot)\|_{H^{-2}} \le C(t^{-1/4}+t^{-3/2})\eps\,.
\end{gather}
\end{lemma}
\begin{proof}
Since $\bM_1(\infty)\lesssim \eps$ by \eqref{eq:tot-3}, we have
$$ \sum_{k=0,1}\sup_{t\ge0}\la t\ra^{(2k+1)/4}\|\pd_y^k \bd(t)\|_Y+
\la t\ra\|\pd_y^2\bd(t)\|_Y \lesssim \eps\,.$$
Thus we have
\begin{align*}
\|\pd_y^k\bd_\lambda(t,\cdot)\|_{L^2}=& 
\lambda^{(2k+1)/2}\|\pd_y^k\bd(\lambda^2t,\cdot)\|_Y
\\ \lesssim & \lambda^{(2k+1)/2}(1+\lambda^2t)^{-(2k+1)/4}\eps
\lesssim t^{-(2k+1)/4}\eps \quad\text{for $k=0$, $1$,}
\end{align*}
$$\|\pd_y^2\bd_\lambda(t,\cdot)\|_{L^2} \lesssim  \lambda^{5/2}
\|\pd_y^2\bd(\lambda^2t,\cdot)\|_Y
 \lesssim  \lambda^{5/2}(1+\lambda^2t)^{-1}\eps
\lesssim \lambda^{1/2}t^{-1}\eps\,.$$
Thus we prove \eqref{eq:bd-bound1}.
\par
Next we will show \eqref{eq:bd-bound2}.
Let $\cN'(t,y)+\pd_y\cN''(t,y)=\diag(1,\pd_y)\sum_{i=2}^8\cN_i(t,y)$,
\begin{gather*}
\cN_\lambda'(t,y)=\lambda^3\cN'(\lambda^2t,\lambda y)\,,\quad
\cN_\lambda''(t,y)=\lambda^2\cN''(\lambda^2t,\lambda y)\,,\\
\widetilde{\cN}(t,y)=
\wP_1\begin{pmatrix}n_1(t,y) \\ n_2(t,y) \end{pmatrix}\,,\quad
\widetilde{\cN}_\lambda(t,y)=\lambda^2\widetilde{\cN}(\lambda^2t,\lambda y)\,.
\end{gather*}
Then \eqref{eq:bd} can be rewritten as
\begin{equation}
  \label{eq:bd-l}
  \begin{split}
\pd_t\bd_\lambda=\{2\pd_y^2I&+\lambda\pd_y(A_3(\lambda^2t,\lambda^{-1}D_y)
+A_4(\lambda^2t,\lambda^{-1}D_y))\}\bd_\lambda 
\\ & +A_5(\lambda^2t,\lambda^{-1}D_y)
\{\pd_y(\widetilde{\cN}_\lambda+\cN_\lambda'')+\cN_\lambda'\}\,.
  \end{split}
\end{equation}
By \eqref{eq:bd-l},
\begin{equation}
  \label{eq:bd-ll}
\begin{split}
\|\pd_t\bd_\lambda\|_{H^{-2}}\le & 2\|\bd_\lambda\|_{L^2} 
+\lambda\|A_3(\lambda^2t,\lambda^{-1}D_y)\bd_\lambda\|_{H^{-1}}
+\lambda\|A_4(\lambda^2t,\lambda^{-1}D_y)\bd_\lambda\|_{L^2} \\
& +\|A_5(\lambda^2t,\lambda^{-1}D_y)(\widetilde{\cN}_\lambda+\cN_\lambda'')
\|_{L^2}+\|A_5(\lambda^2t,\lambda^{-1}D_y)\cN_\lambda'\|_{L^2}
\end{split}  
\end{equation}
Now we will estimate each term of the right hand side.
By \eqref{eq:A0A*} and the fact that $\omega(\eta)=4+O(\eta^2)$,
we have
\begin{equation}
  \label{eq:A3-bound}
|A_3(\lambda^2t,\lambda^{-1}\eta)|\lesssim \lambda^{-2}\eta^2\,.  
\end{equation}
Thus by \eqref{eq:bd-bound1} and Plancherel's theorem, 
\begin{equation}
  \label{eq:lim-la1}
\begin{split}
\lambda\|A_3(\lambda^2t,\lambda^{-1}D_y)\bd_\lambda\|_{H^{-1}}\lesssim &
\lambda^{-1}\|\la \eta\ra^{-1}\eta^2(\mF_y\bd_\lambda)(t,\eta)\|_{L^2}
\\ \lesssim & \lambda^{-1}\|\pd_y\bd_\lambda(t,\cdot)\|_{L^2}
\lesssim \lambda^{-1}t^{-3/4}\eps\,.
\end{split}  
\end{equation}
Since $\|A_4(t,D_y)\|_{B(Y)}\lesssim \|A_2(t,D_y)\|_{B(Y)}\lesssim
e^{-a(4t+L)}$, it follows from \eqref{eq:bd-bound1} and the scaling argument
that
\begin{equation}
  \label{eq:lim-la2}
\begin{split}
\lambda\|A_4(\lambda^2t,\lambda^{-1}D_y)\bd_\lambda(t,\cdot)\|_{L^2}
= & \lambda^{3/2}\|A_4(\lambda^2t,D_y)\bd(\lambda^2t,\cdot)\|_Y
\\ \lesssim & \lambda^{3/2}e^{-a(4\lambda^2t+L)}\|\bd(\lambda^2t,\cdot)\|_Y
\\ \lesssim & \lambda t^{-1/4}e^{-a(4\lambda^2t+L)}\eps
 \lesssim \lambda^{-1/4}t^{-7/8}\eps\,.
\end{split}  
\end{equation}
\par
Following the proof of Lemma~\ref{lem:M1-bound},
we have for $t\ge0$,
\begin{equation}
  \label{eq:cN'-est}
\|\widetilde{\cN}\|_Y\lesssim \la t\ra^{-3/4}\eps^2\,\enskip
\|\cN'\|_Y\lesssim \la t\ra^{-3/2}\eps^2\,,\enskip
\|\cN''\|_Y\lesssim \la t\ra^{-5/4}\eps^2\,.
\end{equation}
Nonlinear terms decay $t^{-1/4}$ times faster in \eqref{eq:cN'-est}
than those in \eqref{eq:cN1-est} and \eqref{eq:N'N''} because
$Y$ and $Y_1$ have the same scaling as $L^2(\R)$ and $L^1(\R)$,
respectively.  By \eqref{eq:cN'-est},
\begin{align*}
& \|\widetilde{\cN}_\lambda\|_{L^2}
=\lambda^{3/2}\|\widetilde{\cN}(\lambda^2t,\cdot)\|_Y
\lesssim \lambda^{3/2}(1+\lambda^2t)^{-3/4}\eps^2
\lesssim t^{-3/4}\eps^2\,,
\\  &
\|\cN_\lambda'\|_{L^2}= \lambda^{5/2}\|\cN'(\lambda^2t,\cdot)\|_Y
\lesssim \lambda^{5/2}(1+\lambda^2t)^{-3/2}\eps^2
\lesssim \lambda^{-1/2}t^{-3/2}\eps^2\,,
\\  &
\|\cN_\lambda''\|_{L^2}=\lambda^{3/2}\|\cN''(\lambda^2t,\cdot)\|_{L^2}
\lesssim \lambda^{3/2}(1+\lambda^2t)^{-5/4}\eps^2
\lesssim \lambda^{-1/4}t^{-7/8}\eps^2\,.
\end{align*}
Since $\sup_{\lambda\ge1}\|A_5(\lambda^2t,\lambda^{-1}D_y)\|_{B(L^2)}\lesssim 1$,  
we have
\begin{gather}
\label{eq:lim-la3}
\|A_5(\lambda^2t,\lambda^{-1}D_y)\widetilde{\cN}_\lambda\|_{L^2}
\lesssim t^{-3/4}\eps^2\,,
\\ \label{eq:lim-la4}
\|A_5(\lambda^2t,\lambda^{-1}D_y)\cN_\lambda''\|_{L^2}
\lesssim \lambda^{-1/4}t^{-7/8}\eps^2\,,
\\ \label{eq:lim-la5}
\|A_5(\lambda^2t,\lambda^{-1}D_y)\cN_\lambda'\|_{L^2}
\lesssim  \lambda^{-1/2}t^{-3/2}\eps^2\,.
\end{gather}
Combining \eqref{eq:bd-bound1} and \eqref{eq:bd-ll}--\eqref{eq:lim-la2},
\eqref{eq:lim-la3}--\eqref{eq:lim-la5}, we obtain \eqref{eq:bd-bound2}.
\end{proof}
By Lemma~\ref{lem:bd-bound} and the Arzel\`a-Ascoli theorem,
we have the following.
\begin{corollary}
\label{cor:bd-bound}
There exist a sequence $\{\lambda_n\}_{n\ge1}$ satisfying
$\lim_{n\to\infty}\lambda_n=\infty$ and $\bd_\infty(t,y)$ such that
\begin{align*}
& \bd_{\lambda_n}(t,\cdot)\to \bd_\infty(t,\cdot)
\quad\text{weakly star in $L^\infty_{loc}((0,\infty);H^1(\R))$,}\\
& \pd_t\bd_{\lambda_n}(t,\cdot)\to \pd_t\bd_\infty(t,\cdot)
\quad\text{weakly star in $L^\infty_{loc}((0,\infty);H^{-2}(\R))$,}\\
& \sup_{t>0}t^{1/4}\|\bd_\infty(t)\|_{L^2}\le C\eps\,,
\end{align*}
where $C$ is a constant given in Lemma~\ref{lem:bd-bound}.
Moreover, for any $R>0$ and $t_1$, $t_2$ with $0<t_1\le t_2<\infty$,
$$\lim_{n\to\infty}\sup_{t\in[t_1,t_2]}
\|\bd_{\lambda_n}(t,\cdot)-\bd_\infty(t,\cdot)\|_{L^2(|y|\le R)}=0\,.$$
\end{corollary}
Next we will show that $\bd_\infty(t)$ is a self-similar solution
to a system of Burgers equations. 
To begin with, we will prove the following.
\begin{lemma}
  \label{lem:Burgers}
Let $\bd_\infty(t)={}^t(d_+(t,y),d_-(t,y))$. Then for $t>0$ and $y\in\R$,
\begin{equation}
  \label{eq:bar-bd}
\left\{\begin{aligned}
& \pd_td_+=2\pd_y^2d_++4\pd_y(d_+^2)\,,  \\
& \pd_td_-=2\pd_y^2d_--4\pd_y(d_-)^2\,,
\end{aligned}\right.
\end{equation}
and
\begin{equation}
  \label{eq:B-ini}
\lim_{t\downarrow0}\int_\R \bd_\infty(t,y)h(y)\,dy
=\sqrt{2\pi}(\mF_y\tilde{\bd})(0,0)h(0)
\quad\text{for any $h\in H^2(\R)$,}
\end{equation}
where 
$$\tilde{\bd}(t,y)=\bd(t,y)+\int_t^\infty A_5(s,D_y)\cN'(s,y)\,ds\,.$$
\end{lemma}
\begin{proof}
Let $\tilde{\bd}_\lambda(t,y)=\lambda\tilde{\bd}(\lambda t,\lambda^2y)$.
The limiting profile of $\bd_\lambda(t)$ and  $\tilde{\bd}_\lambda(t)$
as $\lambda\to\infty$ are the same for every $t>0$.
Indeed, it follows from \eqref{eq:lim-la5} that
\begin{equation}
  \label{eq:dif-bds}
  \begin{split}
\|\tilde{\bd}_\lambda(t,\cdot)-\bd_\lambda(t,\cdot)\|_{L^2}
\lesssim & \int_t^\infty
\left\|A_5(\lambda^2s,\lambda^{-1}D_y)\cN_\lambda'(s,y)\right\|_{L^2}\,ds
\\ \lesssim & \lambda^{-1/2}\int_t^\infty\tau^{-3/2}\,d\tau
\lesssim \lambda^{-1/2}t^{-1/2}\,.    
  \end{split}
\end{equation}
By \eqref{eq:bd-l},
\begin{equation}
  \label{eq:tbd-l}
  \begin{split}
\pd_t\tilde{\bd}_\lambda= & 2\pd_y^2\bd_\lambda
+\lambda\pd_y\{(A_3(\lambda^2t,\lambda^{-1}D_y)
+A_4(\lambda^2t,\lambda^{-1}D_y))\}\bd_\lambda 
\\ & +\pd_yA_5(\lambda^2t,\lambda^{-1}D_y)(\widetilde{\cN}_\lambda+\cN_\lambda'')\,,
  \end{split}  
\end{equation}
and we have $\sup_{\lambda\ge1}\|\pd_t\tilde{\bd}_\lambda(t,\cdot)\|_{H^{-2}}
\lesssim t^{-1/4}+t^{-7/8}$ from \eqref{eq:bd-bound1},
\eqref{eq:lim-la1}, \eqref{eq:lim-la2}, \eqref{eq:lim-la3}, \eqref{eq:lim-la4}
and \eqref{eq:tbd-l}.
Thus for$t>s>0$ and $h\in H^2(\R)$,
$$\left|\int_\R\tilde{\bd}_\lambda(t,y)h(y)\,dy
-\int_\R\tilde{\bd}_\lambda(s,y)h(y)\,dy\right|\le C
\{(t-s)^{3/4}+(t-s)^{1/8}\}\,,$$
where $C$ is a constant independent of $\lambda$.
Passing to the limit as $s\downarrow0$ in the above, we obtain for $t>0$,
\begin{equation}
  \label{eq:tbd-0}
  \left|\int_\R\tilde{\bd}_\lambda(t,y)h(y)\,dy
-\int_\R\tilde{\bd}_\lambda(0,y)h(y)\,dy\right|\le C(t^{3/4}+t^{1/8})\,.
\end{equation}
Since $\bd(0,\cdot)\in Y_1$ by the definition
and $\|\cN'(\tau,\cdot)\|_{Y_1}\lesssim \la \tau\ra^{-5/4}$, we have
$$\tilde{\bd}_1(0,y)=\bd(0,y)+\int_0^\infty A_5(\tau,D_y)
\cN'(\tau,y)\,d\tau\in Y_1\,,$$
and it follows from Lebesgue's dominated convergence theorem that
\begin{align*}
\int_\R \tilde{\bd}_\lambda(0,y)h(y)\,dy
=&  \int_\R (\mF_y\tilde{\bd}_1)(0,\lambda^{-1}\eta)(\mF_y^{-1}h)(\eta)\,d\eta
\to  \sqrt{2\pi}(\mF_y\tilde{\bd}_1)(0,0)h(0)
\end{align*}
as $\lambda\to\infty$ for any $h\in H^1(\R)$.
Letting $\lambda=\lambda_n$ and passing to the limit as $n\to\infty$
in \eqref{eq:tbd-0}, we see that \eqref{eq:B-ini} follows from
Corollary~\ref{cor:bd-bound}.
\par

Next, we will show \eqref{eq:bar-bd}.
By Corollary~\ref{cor:bd-bound} and \eqref{eq:dif-bds},
\begin{equation}
  \label{eq:wlim1}
\pd_t\tilde{\bd}_{\lambda_n}-2\pd_y^2\bd_{\lambda_n}
\to \pd_t\bd_\infty-2\pd_y^2\bd_\infty\,.  
\end{equation}
By \eqref{eq:lim-la1}, \eqref{eq:lim-la2} and \eqref{eq:lim-la4},
\begin{equation}
  \label{eq:wlim2}
\lambda\{(A_3(\lambda_n^2t,\lambda_n^{-1}D_y)
+A_4(\lambda_n^2t,\lambda_n^{-1}D_y))\}\bd_{\lambda_n}
+A_5(\lambda^2t,\lambda^{-1}D_y)\cN_{\lambda_n}''\to 0\,,
\end{equation}
as $n\to\infty$ in $\mathcal{D}'((0,\infty)\times \R)$.
\par
Now we investigate the limit of $\widetilde{\cN}_\lambda$.
Let $$\bd_\lambda(t,y)=
\begin{pmatrix}d_{+,\lambda}(t,y)\\ d_{-,\lambda}(t,y)\end{pmatrix}\,.$$
By the definition of $\bd(t)$, 
$$\begin{pmatrix}  b(t,\cdot) \\ x_y(t,\cdot)\end{pmatrix}
=U_1(t)\Pi_1(t,D_y)\bd(t,\cdot)\,.$$
Since  $|U_1(t)-I|\lesssim e^{-a(4t+L)}$ and
\begin{equation}
\label{eq:Pi-est}
\left|(4i)^{-1}\Pi_*(\lambda^{-1}\eta)
-\begin{pmatrix}2 & 2\\ 1 & -1 \end{pmatrix}\right|
\lesssim \lambda^{-1}\eta\\quad \text{for $\lambda\ge1$,}
\end{equation}
we have
\begin{equation}
  \label{eq:bx-d} 
\begin{split}
& \left\|\lambda
\begin{pmatrix} b(\lambda^2t,\lambda \cdot)\\ x_y(\lambda^2t,\lambda \cdot)
\end{pmatrix}
+\begin{pmatrix}8 & 8\\ 4 & -4 \end{pmatrix}
\begin{pmatrix}e^{4\lambda t\pd_y}d_{+,\lambda}(t,\cdot)\\
e^{-4\lambda t\pd_y}d_{-,\lambda}(t,\dot)\end{pmatrix}\right\|_{L^2}
\\ \lesssim & (\lambda^{-1}+e^{-4a\lambda^2t})\|\bd_\lambda(t)\|_{H^1}
\lesssim (\lambda^{-1}+e^{-4a\lambda^2t})(t^{-1/4}+t^{-3/4})\,.
\end{split}
\end{equation}
Recall that $(n_1,n_2)=(6bx_y,2(\tc-b)+3(x_y)^2)$.
Since $\|n_1(t)\|_{L^1}+\|n_2\|_{L^1}\lesssim t^{-1/2}$,
\begin{equation}
  \label{eq:pl1}
  \begin{split}
& \|\widetilde{\cN}_\lambda-\lambda^2(n_1,n_2)(\lambda^2t, \lambda\cdot)\|_{H^{-1}}
\\ = &  \lambda \|\la \eta\ra^{-1}(\mF_yn_1,\mF_yn_2)
(\lambda^2t,\lambda^{-1}\eta)\|_{L^2(|\eta|\ge \lambda \eta_0)}
\\ \lesssim & \lambda^{1/2}(\|n_1(\lambda^2t,\cdot)\|_{L^1}
+\|n_2(\lambda^2t,\cdot)\|_{L^1}) \lesssim (\lambda t)^{-1/2}\,.
  \end{split}
\end{equation}
Combining \eqref{eq:pl1} with \eqref{eq:bx-d}, we have
\begin{equation}
\label{eq:n1-est}
\begin{split}
&\left\|\lambda^2(\wP_1n_1)(\lambda^2t,\lambda\cdot)-
12\{(e^{4\lambda t\pd_y}d_{+,\lambda})^2-(e^{-4\lambda t\pd_y}d_{-,\lambda})^2
\}\right\|_{H^{-1}}
\\  \lesssim &  C(t)(\lambda^{-1/2}+e^{-4a\lambda^2t})\,,
\end{split}
\end{equation}
where $C(t)$ is a monotone decreasing function of $t$.
Claim~\ref{cl:b-capprox} implies
$$\left\|b-\tc-\frac18\wP_1(b^2)\right\|_Y \lesssim
\|b\|_{L^\infty}^2\|b\|_{L^2}\lesssim \|\bd\|_{L^2}^2\|\pd_y\bd\|_{L^2}\,,$$
whence 
$$\lambda^2\left\|\left(b-\tc-\frac18\wP_1(b^2)\right)(\lambda^2t,\lambda y)
\right\|_Y \lesssim 
\lambda^{-1}\|\bd_\lambda\|_{L^2}^2\|\pd_y\bd_\lambda\|_{L^2}
\lesssim \lambda^{-1}t^{-5/4}\,.$$
We can obtain
$\lambda^2\left\|(I-\wP_1)b^2(\lambda^2t,\lambda y)\right\|_{H^{-1}}
\lesssim  (\lambda t)^{-1/2}$ in the same as  \eqref{eq:pl1}.
Combining the above with \eqref{eq:bx-d} and \eqref{eq:pl1}, we have 
\begin{equation}
  \label{eq:n2-est}
  \begin{split}
\bigl\| \lambda^2(\wP_1n_2)(\lambda^2t,\lambda \cdot)&
-2\{ (e^{4\lambda t\pd_y}d_{+,\lambda})^2
-4(e^{4\lambda t\pd_y}d_{+,\lambda})(e^{-4\lambda t\pd_y}d_{-,\lambda})
\\ & +(e^{-4\lambda t\pd_y}d_{-,\lambda})^2\}\bigr\|_{H^{-1}}
\lesssim C'(t)(\lambda^{-1/2}+e^{-4a\lambda^2t})\,,
  \end{split}
\end{equation}
where $C'(t)$ is a monotone decreasing function of $t$.
Since
\begin{equation}
\label{eq:inv-Pi}
\left|4i\Pi_*(\lambda^{-1}\eta)^{-1}
-\frac{1}{4} \begin{pmatrix}1 & 2 \\ 1 & -2 \end{pmatrix}\right|
\lesssim \lambda^{-1} \eta \quad \text{for $\lambda\ge1$,}
\end{equation}
it follows from \eqref{eq:n1-est} and \eqref{eq:n2-est} that
\begin{equation}
 \label{eq:wcNl}
 \begin{split}
& \left\|A_5(\lambda^2t, \lambda^{-1}D_y)\widetilde{\cN}_\lambda
-2\begin{pmatrix} 2d_{+,\lambda}^2
-2d_{+,\lambda}(e^{-8\lambda t\pd_y}d_{-,\lambda})-(e^{-8\lambda t\pd_y}d_{-,\lambda})^2
\\
(e^{8\lambda t\pd_y}d_{+,\lambda})^2+2(e^{8\lambda t\pd_y}d_{+,\lambda})d_{-,\lambda}
-2d_{-,\lambda}^2\end{pmatrix}\right\|_{H^{-2}}
\\ \lesssim & C(t)(\lambda^{-1/2}+e^{-4a\lambda^2t}) \,,
   \end{split}
\end{equation}
where $C(t)$ is a monotone decreasing function of $t$.
\par
Next, we will show that $e^{\pm 4\lambda t\pd_y}d_{+,\lambda}$
locally tends to $0$ around $y=\pm 4\lambda t$.
Let $\alpha>0$ and $\zeta_\pm(y)=1\pm\tanh \alpha y$.
Then we have $0\le\zeta_\pm(y)\le 2$, $0<\pm\zeta_\pm'(y)\le\alpha$ and
$0\le |\zeta_\pm''(y)/\zeta_\pm'(y)|\le2\alpha$ for $y\in\R$.
By \eqref{eq:bd-l}, 
\begin{align*}
& \frac12\frac{d}{dt}\int_\R \zeta_+(y-8\lambda t)d_{+,\lambda}^2(t,y)\,dy
+4\lambda\int_\R \zeta_+'(y-8\lambda t)d_{+,\lambda}^2(t,y)\,dy
\\ \le  & 
\int_\R \zeta_+''(y-8\lambda t)(d_{+,\lambda})^2(t,y)\,dy
-2\int_\R \zeta_+(y-8\lambda t)(\pd_yd_{+,\lambda})^2(t,y)\,dy
+V\,,
\end{align*}
where
\begin{align*}
V =&(2+\alpha)\lambda
\|A_3(\lambda^2t,\lambda^{-1}D_y)\bd_\lambda(t)\|_{L^2}\|\bd_\lambda(t)\|_{H^1}
\\ 
& +(2+\alpha)\lambda
\|A_4(\lambda^2t,\lambda^{-1}D_y\bd_\lambda(t))\|_{L^2}\|\bd_\lambda(t)\|_{H^1}
\\ 
&  +(2+\alpha)\|\bd_\lambda(t)\|_{H^1}
\|A_5(\lambda^2t,\lambda^{-1}D_y)(\widetilde{\cN}_\lambda(t)+\cN''_\lambda(t))\|_{L^2}
\\
& +2\|\bd_\lambda(t)\|_{L^2}
\|A_5(\lambda^2t,\lambda^{-1}D_y)\cN'_\lambda(t)\|_{L^2}\,.
\end{align*}
Using Lemma~\ref{lem:bd-bound} and \eqref{eq:A3-bound}, we have
\begin{equation}
  \label{eq:a3-bound}
\lambda \|A_3(\lambda^2t,\lambda^{-1}D_y)\bd_\lambda(t)\|_{L^2}
\lesssim \lambda^{-1/2}t^{-1}\,.
\end{equation}
By Lemma~\ref{lem:bd-bound}, \eqref{eq:lim-la2},
\eqref{eq:lim-la3}--\eqref{eq:lim-la5} and \eqref{eq:a3-bound},
\begin{align*}
V \lesssim  \lambda^{-1/4}(t^{-9/8}+t^{-9/4})+t^{-1}+t^{-3/2}\,.
\end{align*}
If $\alpha$ is sufficiently small, it follows that
\begin{equation}
  \label{eq:mono-y}
  \begin{split}
& \frac{d}{dt}\int_\R \zeta_+(y-8\lambda t)d_{+,\lambda}^2(t,y)\,dy
+4\lambda\int_\R \zeta_+'(y-8\lambda t)d_{+,\lambda}^2(t,y)\,dy
\\ \le  & C\lambda^{-1/4}(t^{-9/8}+t^{-9/4})+C(t^{-1}+t^{-3/2}),,
  \end{split}
\end{equation}
where $C$ is a positive constant independent of $t>0$ and $\lambda\ge1$.
Let $0<t_1<t_2<\infty$.
Integrating \eqref{eq:mono-y} over $[t_1,t_2]$, we obtain
\begin{equation}
  \label{eq:mono-y2}
0<\int_{t_1}^{t_2}\int_\R \zeta_+'(y-8\lambda t)d_{+,\lambda}^2(t,y)\,dydt
\le   C(t_1,t_2)\lambda^{-1}\,,
\end{equation}
where $C(t_1,t_2)$ is a constant independent of $\lambda\ge 1$.
We can prove
\begin{equation}
  \label{eq:mono-y3}
0<-\int_{t_1}^{t_2}\int_\R \zeta_-'(y+8\lambda t)d_{-,\lambda}^2(t,y)\,dydt
\le   C(t_1,t_2)\lambda^{-1}\,,
\end{equation}
in exactly the same way.
By \eqref{eq:mono-y2} and \eqref{eq:mono-y3},
$$\lim_{\lambda\to\infty}\|d_{\pm,\lambda}\|_{L^2([t_1,t_2]\times B_R^\pm)}=0$$
for any $R>0$, where
$B_R^\pm=\{y\in\R\mid |y\mp 8\lambda t|\le R\}$.
Combining the above with Corollary~\ref{cor:bd-bound}, \eqref{eq:bd-l},
\eqref{eq:wlim1}, \eqref{eq:wlim2} and \eqref{eq:wcNl}, we see that $\bd_\infty$
satisfies \eqref{eq:bar-bd}.
\end{proof}

Now we are in position to prove Theorem~\ref{thm:Burgers}.
\begin{proof}[Proof of Theorem~\ref{thm:Burgers}]
By Corollary~\ref{cor:bd-bound},
$$\|\pd_y(d_\pm^2)(t,\cdot)\|_{H^{-2}}
\lesssim \|\bd_\infty(t,\cdot)\|_{L^2}^2\lesssim t^{-1/2}\,,$$
whence 
$\pd_y(d_\pm(t)^2)\in L^1_{loc}((0,\infty);H^{-2}(\R))$.
Combining the above with \eqref{eq:bar-bd} and \eqref{eq:B-ini},
we have $d_\pm(t)\in C([0,\infty);H^{-2}(\R))$ and
\begin{equation}
  \label{eq:int-B}
d_\pm(t)=c_\pm H_{2t}\pm 4\int_0^t e^{2(t-s)\pd_y^2}\pd_y(d_\pm(s)^2)\,ds\,,
\end{equation}
where $(c_+, c_-)=\sqrt{2\pi}(\mF_y\tilde{\bd}_1)(0,0)$.
If we choose $m_\pm\in(-2\sqrt{2},2\sqrt{2})$ so that
$$\int_\R u_B^\pm (t,y)\,dy
=\frac{1}{2}\log\left(\frac{2\sqrt{2}\pm m_\pm }{2\sqrt{2}\mp m_\pm}\right)
=c_\pm\,,$$
then $u_B^\pm(t)$ is also a solution to \eqref{eq:int-B} satisfying
$\sup_{t>0}t^{1/4}\|u_B^\pm(t)\|_{L^2}\lesssim \eps$.
Let $\|u\|_W=\sup_{t>0}t^{1/4}\|u(t,\cdot)\|_{L^2(\R)}$.
Since $\|\pd_ye^{2t\pd_y^2}\|_{B(L^1;L^2)}\lesssim t^{-3/4}$,
\begin{align*}
\|d_\pm-u_B^\pm\|_W\le & 4\sup_{t>0}t^{1/4}
\int_0^t\|\pd_ye^{2(t-s)\pd_y^2}\|_{B(L^1;L^2)}\|d_\pm(s)^2-u_B^\pm(s)^2\|_{L^1}\,ds
\\ \lesssim & (\|d_\pm\|_W+\|u_B^\pm\|_W)\|d_\pm-u_B^\pm\|_W
t^{1/4}\int_0^t(t-s)^{-3/4}s^{-1/2}\,ds
\\ \lesssim & \eps\|d_\pm-u_B^\pm\|_W\,.
\end{align*}
Thus we have $d_\pm(t,y)=u_B^\pm(t,y)$ for small $\eps$ and
\eqref{eq:bd-c} follows from the uniqueness of the limiting profile
$\bd_\infty(t,y)=(d_+(t,y),d_-(t,y))$.
Obviously $\bd_\infty(t,y)=(u_B^+(t,y),u_B^-(t,y))$  satisfies \eqref{eq:ss-bd}.
Now Theorem~\ref{thm:Burgers} follows immediately from \eqref{eq:conv},
\eqref{eq:bx-d} and the definition of $b(t,y)$.
Thus we complete the proof.
\end{proof}
\bigskip

\appendix
\section{Proof of Lemma~\ref{lem:G}}
\label{ap:G}
To prove Lemma~\ref{lem:G}, we need the following.
\begin{claim}
  \label{cl:phi-int}
Let $\varphi_c(x)=c\sech^2(\sqrt{c/2}x)$, $\varphi=\varphi_2$
and  $\pd_c^k\varphi=\pd_c^k\varphi_c|_{c=2}$ for $k\in\N$. Then
  \begin{align}
\label{eq:phi-int1}
& \int_\R \varphi(x)\,dx=4\,,\quad 
\int_\R \varphi(x)^2dx=\frac{16}{3}\,,
\\ \label{eq:phi-int2}  & 
\int_\R \varphi(x)\pd_c\varphi(x)\,dx=
-\int_\R \varphi'(x)\left(\int^x_{-\infty}\pd_c\varphi(z)dz\right)=2\,,\\
\label{eq:phi-int3}
& \int_\R \varphi(x)\left(\int_x^\infty\pd_c\varphi(z)dz\right)\,dx=
\int_\R \varphi(x)\left(\int^x_{-\infty}\pd_c\varphi(z)dz\right)=2\,,\\
\label{eq:phi-int4}
& \int_\R \varphi(x)\left(\int_x^\infty\pd_c^2\varphi(z)dz\right)\,dx=-\frac12\,,
\enskip
\int_\R \pd_c\varphi(x)\left(\int^x_{-\infty}\pd_c\varphi(z)dz\right)\,dx=\frac12
\,,
\\
\label{eq:phi-int5}
& \int_\R \left(\int_x^\infty\pd_c\varphi(z)dz\right)
\left(\int^x_{-\infty}\pd_c\varphi(z)dz\right)\,dx=
\frac{1}{6}-\frac{\pi^2}{36}\,,
\\ &
\label{eq:phi-int6}
\int_\R \left(\int_x^\infty\pd_c^2\varphi(z)dz\right)
\left(\int^x_{-\infty}\pd_c\varphi(z)dz\right)=\frac{\pi^2}{96}-\frac{1}{16}
\,.
\end{align}
\end{claim}
\begin{proof}
Eq.~\eqref{eq:phi-int1} can be obtained by using the change of variable
$s=\tanh x$.
Since
\begin{equation}
  \label{eq:phi-scaling}
\varphi_c(x)=\frac{c}{2}\varphi(\mathstrut{\sqrt{c/2}}x)\,,  
\end{equation}
$$\int\varphi\pd_c\varphi dx=\frac12\frac{d}{dc}\left(\frac{c}{2}\right)^{3/2}
\bigr|_{c=2}\int\varphi^2dx=2\,.$$
Using \eqref{eq:phi-scaling}, we have
\begin{equation}
  \label{eq:phi-7}
  \begin{split}
& \pd_c\varphi(x)=\frac{x}{4}\varphi'(x)+\frac{1}{2}\varphi(x)
=\frac{1}{4}(x\varphi)'+\frac14\varphi\,,\\
& \pd_c^2\varphi(x)=\frac{x^2}{16}\varphi''(x)+\frac{3x}{16}\varphi'(x)
=\frac{1}{16}(x^2\varphi'+x\varphi)'-\frac{1}{16}\varphi\,,\\
& \int_{\pm\infty}^x\pd_c\varphi=\frac{x\varphi}{4}+\frac{\tanh x\mp1}{2}\,,
\quad
\int_x^\infty\pd_c^2\varphi=-\frac{x^2\varphi'+x\varphi}{16}
+\frac{\tanh x-1}{8}\,.
  \end{split}
\end{equation}
By \eqref{eq:phi-7} and the fact that $\varphi$ is even,
\begin{gather*}
\int \varphi\left(\int_x^\infty\pd_c\varphi\right)=
\int \varphi\left(\int^x_{-\infty}\pd_c\varphi\right)=\frac12\int\varphi=2\,,\\
\int \varphi\int_x^\infty\pd_c^2\varphi=-\frac{1}{8}\int\varphi=-\frac12\,,\\
\int_\R \pd_c\varphi\int^x_{-\infty}\pd_c\varphi=\frac18\int\varphi=\frac12\,.
\end{gather*}
By \eqref{eq:phi-7},
\begin{align*}
& \int_\R \left(\int_x^\infty\pd_c\varphi(z)dz\right)
\left(\int^x_{-\infty}\pd_c\varphi(z)dz\right)\,dx
\\ =&  \int_\R \left\{\frac14-
\left(\frac{x}{4}\varphi+\frac{1}{2}\tanh x\right)^2\right\}\,dx
\\=& \frac14\int\sech^2x-\frac12\int x\sech^3x\sinh x
-\frac14\int x^2\sech^4xdx
\\=& -\frac14\int x^2\sech^4xdx=-\frac{\pi^2-6}{36}\,.
\end{align*}
Here use the fact that $\int_0^\infty \frac{x}{e^x+1}\,dx=\frac{\pi^2}{12}$.
We have \eqref{eq:phi-int6} in the same way.
\end{proof}

Now we are in position to prove Lemma~\ref{lem:G}.

\begin{proof}[Proof of Lemma~\ref{lem:G}]
By Claims~\ref{cl:phi-int} and \ref{cl:gk-approx},
\begin{align*}
G_1=&\int_\R \ell_1 \varphi_c(z)dz
\\ = &3x_{yy}\int\varphi_c^2-(c_t-6c_yx_y)\int\varphi_c\pd_c\varphi_c
\\ & +3c_{yy}\int\varphi_c\int_z^\infty\pd_c\varphi_c
+3(c_y)^2\int\varphi_c\int_z^\infty\pd_c^2\varphi_c
\\=& 16x_{yy}\left(\frac{c}{2}\right)^{3/2}
-2(c_t-6c_yx_y)\left(\frac{c}{2}\right)^{1/2}+6c_{yy}-\frac{3}{2}(c_y)^2
\left(\frac{2}{c}\right)\,,
\end{align*}
and 
\begin{align*}
& \left(\frac{c}{2}\right)^{-3/2}G_2=
 \int_\R \ell_1\left(\int_{-\infty}^x\pd_c\varphi_c(z)dz\right)dz
\\=&(x_t-2c-3(x_y)^2)
\int \varphi_c'\left(\int_{-\infty}^x\pd_c\varphi_c\right)
\\ & +3x_{yy}\int \varphi_c\left(\int_{-\infty}^x\pd_c\varphi_c\right)
-(c_t-6c_yx_y)\int \pd_c\varphi_c\left(\int_{-\infty}^x\pd_c\varphi_c\right)
\\ & +3c_{yy}\int\left(\int_z^\infty\pd_c\varphi_c\right)
\left(\int_{-\infty}^x\pd_c\varphi_c\right)
+3(c_y)^2\int\left(\int_z^\infty \pd_c^2\varphi_c\right)
\left(\int_{-\infty}^x\pd_c\varphi_c\right)
\\=&
-2(x_t-2c-3(x_y)^2)\left(\frac{c}{2}\right)^{1/2}+6x_{yy}
-\frac12(c_t-6c_yx_y)\left(\frac{c}{2}\right)^{-1}
+\mu_1c_{yy}\left(\frac{c}{2}\right)^{-3/2}
\\ & +\mu_2(c_y)^2\left(\frac{c}{2}\right)^{-5/2}\,.
\end{align*}  
\end{proof}

\section{Operator norms of $S^j_k$ and $\wC_k$}
\label{ap:opS}
\begin{claim}
\label{cl:S1}
There exist positive constants $\eta_1$ and  $C$ such that
for $\eta\in (0,\eta_1]$, $j\in\Z_{\ge0}$, $k=1$, $2$ and $f\in L^2(\R)$,
\begin{align}
\label{eq:s1a}
& \|\pd_y^jS_k^1[q_c](f)(t,\cdot)\|_Y \le C\|e^{az}q_2\|_{L^2}
\|\pd_y^j\wP_1f\|_Y \,,\\
\label{eq:s1c}
& \|\pd_y^jS_k^1[q_c](f)(t,\cdot)\|_{Y_1}
\le C\|e^{az}q_2\|_{L^2}\|\pd_y^j\wP_1f\|_{Y_1}\,,\\
\label{eq:S1,pd_y}
& \bigl[\pd_y, S^1_k[q_c]\bigr]=0\,.
\end{align}
\end{claim}
\begin{proof}
Since the Fourier transform of $S_k^1f$ can be written as
\begin{equation}
  \label{eq:F-S1k}
\begin{split}
\mF_y(S_k^1f)(t,\eta)=
\mathbf{1}_{[-\eta_0,\eta_0]}(\eta)
\hat{f}(\eta)\int dz q_2(z)\overline{g_{k1}^*(z,\eta,2)}\,,
\end{split}  
\end{equation}
we have
$[\pd_y,S_k^1]=i\mF_\eta^{-1}[\eta, \mF_y(S_k^1f)(t,\eta)]=0$.
Since
$$\sup_{\eta\in[-\eta_0,\eta_0]}\left|\int dz\, q_2(z)
\overline{g_{k1}^*(z,\eta,2)}\right| \lesssim \|e^{az}q_2(z)\|_{L^2}$$
by Claim~\ref{cl:gk-approx}, we see that \eqref{eq:s1a} and \eqref{eq:s1c}
follow immediately from \eqref{eq:F-S1k} and \eqref{eq:S1,pd_y}.
\end{proof}
\begin{claim}
  \label{cl:S2}
There exist positive constants $\eta_1$, $\delta$ and $C$ such that
if $\eta\in(0,\eta_1]$ and $\bM_1(T)\le \delta$,
then for $k=1$, $2$, $t\in[0,T]$ and $f\in L^2(\R)$,
\begin{equation}
\label{eq:s2a}
 \|S_k^2[q_c](f)(t,\cdot)\|_{Y_1}  \le C
\sup_{c\in[2-\delta,2+\delta]}\left(\|e^{az}q_c\|_{L^2}+\|e^{az}\pd_cq_c\|_{L^2}\right)
\|\tc\|_Y\|f\|_{L^2}\,,
\end{equation}
\begin{equation}
\label{eq:s21a}
  \begin{split}
& \|\pd_yS_k^2[q_c](f)(t,\cdot)\|_{Y_1} \\ \le & C\sum_{i=1,2}
\sup_{c\in[2-\delta,2+\delta]}
 \|e^{az}\pd_c^iq_c\|_{L^2} (\|c_y\|_Y\|f\|_{L^2}+\|\tc\|_Y\|\pd_yf\|_{L^2})\,,    
  \end{split}
\end{equation}
\begin{align}
\label{eq:s2b}
& \|S_k^2[q_c](f)(t,\cdot)\|_Y  \le C
\sum_{0\le i\le 2}\sup_{c\in[2-\delta,2+\delta]} \|e^{az}\pd_c^iq_c\|_{L^2}\|\tc\|_{L^\infty}\|f\|_{L^2}\,,\\
\label{eq:s2-pdy}
& \|[\pd_y,S_k^2[q_c]]f(t,\cdot)\|_{Y_1}  \le C
\sum_{0\le i\le 3}\sup_{c\in[2-\delta,2+\delta]} \|e^{az}\pd_c^iq_c\|_{L^2}\|c_y\|_Y\|f\|_{L^2}\,.
  \end{align}
\end{claim}
\begin{proof}
By the definition of $g_{k2}^*$,
$$\sup_{c\in[2-\delta,2+\delta]}\sup_{|\eta| \le \eta_0}
\left|\int_\R g_{k2}^*(z,\eta,c)dz\right| \lesssim 
\sup_{c\in[2-\delta,2+\delta]}\left(\|e^{az}q_c\|_{L^2_z}+\|e^{az}\pd_cq_c\|_{L^2_z}\right)
\,.$$
Since
$$\mF_y(S_k^2[q_c]f)(t,\eta)=\int dy e^{-iy\eta} f(y)\tc(t,y)
\frac{\mathbf{1}_{[-\eta_0,\eta_0]}(\eta)}{\sqrt{2\pi}}
\int dz\, \overline{g_{k2}^*(z,\eta,c(t,y))}\,,$$
we have 
\begin{align*}
\|S^2_k[q_c](f)(t,\cdot)\|_{Y_1} =& 
\left\|\mF_y(S_k^2[q_c](f))(t,\eta)\right\|_{L^\infty[-\eta_0,\eta_0]}
\\ \lesssim & 
\sup_{c\in[2-\delta,2+\delta]}\left(\|e^{az}q_c\|_{L^2_z}+\|e^{az}\pd_cq_c\|_{L^2_z}\right)
\int \left|f(y)\tc(t,y)\right|dy
\\ \lesssim & 
\sup_{c\in[2-\delta,2+\delta]}\left(\|e^{az}q_c\|_{L^2_z}+\|e^{az}\pd_cq_c\|_{L^2_z}\right)
\|f\|_{L^2}\|\tc\|_Y\,.
\end{align*}

Next, we will prove \eqref{eq:s2b}. Let
\begin{align*}
& S_{k1}^2[q_c](f)(t,y)=\frac{1}{2\pi} \int_{-\eta_0}^{\eta_0}\int_{\R^2}
f(y_1)\tc(t,y_1)\overline{g_{k2}^*(z,\eta,2)}e^{i(y-y_1)\eta}dy_1dzd\eta\,,\\
& S_{k2}^2[q_c](f)(t,y)=\frac{1}{2\pi} \int_{-\eta_0}^{\eta_0}\int_{\R^2}
f(y_1)\tc(t,y_1)^2\overline{g_{k4}^*(z,\eta,c(t,y_1))}
e^{i(y-y_1)\eta}dy_1dzd\eta\,,
\end{align*}
where 
$$g_{k4}^*(z,\eta,c)=\frac{g_{k2}^*(z,\eta,c)-g_{k2}^*(z,\eta,2)}{c-2}\,.$$
Then $S^2_k[q_c]= S^2_{k1}[q_c]+S^2_{k2}[q_c]$ and we can prove
\begin{equation}
  \label{eq:s21s22}
\begin{split}
& \|S^2_{k1}[q_c]f(t,\cdot)\|_Y\lesssim \sum_{0\le i\le 2}\sup_{c\in[2-\delta,2+\delta]}
\|e^{az}\pd_c^iq_c\|_{L^2_z}\|\tc\|_{L^\infty}\|f\|_{L^2}\,,\\
& \|S^2_{k2}[q_c]f(t,\cdot)\|_{Y_1}\lesssim \sum_{0\le i\le 2}\sup_{c\in[2-\delta,2+\delta]}
\|e^{az}\pd_c^iq_c\|_{L^2_z}\|\tc\|_{L^4}^2\|f\|_{L^2}\,,
\end{split}  
\end{equation}
in exactly the same way as \eqref{eq:s1a} and \eqref{eq:s2a}.
Since 
$$ \|S^2_k[q_c]f(t,\cdot)\|_Y\lesssim  \|S^2_{k1}[q_c]f(t,\cdot)\|_Y+
\|S^2_{k1}[q_c]f(t,\cdot)\|_{Y_1}\,,$$
\eqref{eq:s2b} follows from \eqref{eq:s21s22}.
\par
Now we will show \eqref{eq:s2-pdy}.
Noting that
\begin{align*}
 \mF_y(\left[\pd_y\,,\, S^2_{k1}[q_c]\right]f)(t,\eta)=&
\frac{\mathbf{1}_{[-\eta_0,\eta_0]}(\eta)}
{\sqrt{2\pi}}\int_{\R}f(y)c_y(t,y)e^{-iy\eta}\,dy
\int_{\R}\overline{g_{k2}^*(z,\eta,2)}\,dz\,,\\
\mF_y(\left[\pd_y\,,\,S^2_{k2}[q_c]\right]f)(t,\eta)=&
\frac{\mathbf{1}_{[-\eta_0,\eta_0]}(\eta)}{\sqrt{2\pi}}\int_{\R^2}f(y)\pd_y
\left(\tc(t,y)^2\overline{g_{k4}^*(z,\eta,c(t,y))}\right)
\\ & \hskip6cm \times e^{-iy\eta}\,dzdy\,,
\end{align*}
we can prove \eqref{eq:s2-pdy} in the same way as \eqref{eq:s2a}. 
Eq.~\eqref{eq:s21a} immediately follows from \eqref{eq:s2a} and
\eqref{eq:s2-pdy}.
Thus we complete the proof.
\end{proof}

Next we will estimate the operator norm of $S^3[p](f)$.
\begin{claim}
\label{cl:S3}
There exist positive constants $\eta_1$ and $C$
such that for $\eta_0\in(0,\eta_1]$, $k=1$, $2$, $t\ge0$ and $f\in L^2(\R)$,
\begin{align}
\label{eq:s31a}
& \|S^3_k[p](f)(t,\cdot)\|_Y\le Ce^{-a(4t+L)}\|e^{az}p\|_{L^2}\|\wP_1f\|_{Y}\,,\\
\label{eq:s31b}
& \|S^3_k[p](f)(t,\cdot)\|_{Y_1}\le Ce^{-a(4t+L)}\|e^{az}p\|_{L^2}\|\wP_1f\|_{Y_1}\,.
\end{align}
Moreover,
\begin{equation}
  \label{eq:[pd,s31]}
[\pd_y,S_{k1}^3[p]]=0\,.
\end{equation}
\end{claim}
\begin{proof}
The Fourier transform of $S^3_kf$ is
\begin{equation}
  \label{eq:F-S3k}
\mF_y(S_k^3f)(t,\eta)=\mathbf{1}_{[-\eta_0,\eta_0]}(\eta)
\hat{f}(\eta)\int_\R p(z+4t+L)\overline{g_k^*(z,\eta)}\, dz\,.
\end{equation}
By Claim~\ref{cl:gk-approx}, 
\begin{equation}
\label{eq:s311a}
\begin{split}
  \left|\int p(z+4t+L)\overline{g_k^*(z,\eta)}dz\right|
\le & e^{-a(4t+L)}\|e^{az}p(z)\|_{L^2_z}\sup_{|\eta|\le \eta_0}
\|e^{-az}g_k^*(z,\eta)\|_{L^2_z}
\\ \lesssim & e^{-a(4t+L)}\|e^{az}p(z)\|_{L^2}\,.
\end{split}
\end{equation}
Combining \eqref{eq:F-S3k} and \eqref{eq:s311a}, we immediately have 
\eqref{eq:s31a} and \eqref{eq:s31b}. 
Eq.~\eqref{eq:[pd,s31]} clearly follows from the definition
of $S^3_k$. Thus we complete the proof.
\end{proof}
\begin{claim}
\label{cl:S4}
There exist positive constants $\eta_1$, $\delta$ and $C$ such that
if $\eta_0\in(0,\eta_1]$ and $M_1(T)\le \delta$,
then for $k=1$, $2$, $t\in[0,T]$ and $f\in L^2$,
  \begin{gather}
\label{eq:s33}
\|S^4_k[p](f)(t,\cdot)\|_{Y_1} \le Ce^{-a(4t+L)}\|e^{az}p\|_{L^2}
\|\tc\|_Y\|f\|_{L^2}\,,\\
\label{eq:[pdy,S4]}
\|[\pd_y,S^4_{k2}[p]]f(t,\cdot)\|_{Y_1} \le Ce^{-a(4t+L)}\|e^{az}p\|_{L^2}
\|c_y\|_Y\|f\|_{L^2}\,.
\end{gather}
\end{claim}
\begin{proof}
Since
\begin{align*}
&\mF_y(S^4_k[p](f))(t,\eta)\\=&
\frac{\mathbf{1}_{[-\eta_0,\eta_0]}(\eta)}{\sqrt{2\pi}}
\int_{\R^2} f(y)\tc(t,y)p(z+4t+L)\overline{g_{k3}^*(z,\eta,c(t,y)}
e^{-iy\eta}dzdy\,,  
\end{align*}
 we have
\begin{align*}
\|S_{k2}^4[p](f)(t,\cdot)\|_{Y_1} 
\lesssim & \|f\|_{L^2}\|\tc\|_{L^2}e^{-a(4t+L)}\|e^{az}p\|_{L^2}
 \sup_{\substack{c\in [2-\delta,2+\delta],\\ \eta\in[-\eta_0,\eta_0]}}
\|e^{-az}g_{k3}^*(z,\eta,c)\|_{L^2_z}
\\ \lesssim & e^{-a(4t+L)}\|f\|_{L^2}\|\tc(t)\|_Y\,.
\end{align*}
Noting that
\begin{align*}
\left[\pd_y,S^4_{k2}[p]\right](f)=\frac{1}{2\pi}\int_{-\eta_0}^{\eta_0}\int_{\R^2}
 f(y_1)p(z+4t+L)\pd_{y_1}\{c(t,y_1)\overline{g_{k3}^* (z,\eta,c(t,y_1))}\}&
\\ \times  e^{i(y-y_1)\eta}
dy_1dzd\eta\,, &
\end{align*}
we can prove \eqref{eq:[pdy,S4]} in the same way as \eqref{eq:s33}.
Thus we complete the proof.
\end{proof}

\begin{claim}
  \label{cl:S5-S6}
There exist positive constants $\eta_1$, $\delta$ and $C$ such that
if $\eta_0\in(0,\eta_1]$ and $\bM_1(T)\le\delta$,
then for $k=1$, $2$ and $t\in[0,T]$,
\begin{align*}
& \|S^5_kf\|_{Y_1}+\|S^6_kf\|_{Y_1}\le C\|v(t,\cdot)\|_X\|f\|_{L^2}\,.  
\end{align*}
\end{claim}
  \begin{proof}
Using the Schwarz inequality, we have
\begin{align*}
\|S^6_kf\|_{Y_1}=& \sup_{|\eta|\le \eta_0}\frac{1}{\sqrt{2\pi}} \left|
\int_{\R^2} v(t,z,y)f(y)\overline{\pd_zg_k^*(z,\eta,c(t,y))}e^{-iy\eta}\,dzdy\right|
\\ \lesssim & \sup_{c\in [2-\delta,2+\delta],\,\,\eta\in[-\eta_0,\eta_0]}
\|e^{-az}\pd_zg_k^*(z,\eta,c)\|_{L^2_z}
\|v(t)\|_X\|f\|_{L^2_y}\,.
\end{align*}  
Since $\|e^{-az}\pd_zg_k^*(z,\eta,c)\|_{L^2_z}$ is bounded for
$c\in(1,3)$ and $\eta\in[-\eta_0,\eta_0]$, we have
$$\|S^6_kf\|_{Y_1} \lesssim \|v(t,\cdot)\|_X\|f\|_{L^2}\,.$$
We can estimate $S^5_k$ in exactly the same way.
Thus we complete the proof.
\end{proof}

Next, we will estimate operator norms of $\wC_k$ $(k=1$, $2)$.
\begin{claim}
\label{cl:T1-T2} 
There exist positive constants $\delta$ and $C$ such that
if $\sup_{t\in[0,T]}\|\tc(t)\|_Y \le \delta$, then for
$k=1$, $2$ and $t\in[0,T]$,
\begin{gather*}
\|\cC_kf\|_Y \le C\|\tc\|_{L^\infty}\|\wP_1f\|_Y\,,\\
\|\cC_kf\|_{Y_1}\le C\|\tc\|_Y\|\wP_1f\|_Y\,.
\end{gather*}
\end{claim}
\begin{proof}
Since $\|\tc\|_{L^\infty}\lesssim \|\tc\|_Y$
by Remark~\ref{rem:smoothness}, it follows that
$|c^2-4|\le (2+O(\delta))|\tc|$. Thus we have
$$\|\cC_1f\|_Y\le \frac12
\left\|c^2-4\right\|_{L^\infty}\|\wP_1f\|_{L^2}
\lesssim \|\tc\|_{L^\infty}\|\wP_1f\|_Y\,,
$$
\begin{align*}
\|\cC_1f\|_{Y_1}=&\frac12\left\|\mF\left(c^2-4\right)
\ast\mF(\wP_1f)\right\|_{L^\infty([-\eta_0,\eta_0])}
\\ \lesssim & \left\|c^2-4\right\|_{L^2}\|\wP_1f\|_{L^2}
\lesssim \|\tc\|_Y\|\wP_1f\|_Y\,.
\end{align*}
We can estimate $\cC_2$ in exactly the same way.
Thus we complete the proof.
\end{proof}
\begin{claim}
  \label{cl:T-pdy}
There exist positive constants $\delta$ and $C$ such that
if $\sup_{t\in[0,T]}\|\tc(t)\|_Y \le \delta$, then for
$k=1$, $2$ and $t\in[0,T]$,
\begin{gather*}
 \| [\pd_y,\cC_k] f\|_Y\le C\|c_y\|_{L^\infty}\|f\|_{L^2}\,,\\
 \| [\pd_y,\cC_k] f\|_{Y_1}\le C\|c_y\|_Y\|f\|_{L^2}\,.
\end{gather*}
\end{claim}
\begin{proof}
Since $[\pd_y,\cC_1]=\wP_1cc_y\wP_1$,
$$\|[\pd_y,\cC_1]f\|_Y\lesssim \|c_y\|_{L^\infty}\|f\|_{L^2}\,,$$
\begin{align*}
  \|[\pd_y,\cC_1]\|_{Y_1} =& 
\left\|\mF\left(cc_y\right)\ast\mF(\wP_1f)
\right\|_{L^\infty([-\eta_0,\eta_0])}
\lesssim  \|c_y\|_Y\|\wP_1f\|_Y\,.
\end{align*}
We can prove the estimate for $[\pd_y,\cC_2]$ in the same way.
Thus we complete the proof.
\end{proof}

\section{Proof of Claims~\ref{cl:[B3,pdy]}, \ref{cl:invB_3} and \ref{cl:invB_4}}
\label{sec:[B3,pdy]}

\begin{proof}[Proof of Claims~\ref{cl:invB_3} and \ref{cl:invB_4}]
Claims~\ref{cl:S1}--\ref{cl:S3} and \ref{cl:(1+wT2)^{-1}} imply that
for $s\in[0,T]$,
\begin{equation}
  \label{eq:bS1-bound}
  \begin{split}
\|\bS_1\|_{B(Y)}\lesssim & \|\wS_1\|_{B(Y)}\|(1+\wC_2)^{-1}\|_{B(Y)}
\\ \lesssim & \sum_{k=1,2}\left(\|S^1_k[\pd_c\varphi_c]\|_{B(Y)}
+\|S^1_k[\varphi_c']\|_{B(Y)}\right)\lesssim 1\,,
  \end{split}
\end{equation}
\begin{equation}
  \label{eq:wS2bound}
  \begin{split}
\|\bS_2\|_{B(Y,Y_1)}\lesssim & \|\wS_2\|_{B(Y,Y_1)}\|(1+\wC_2)^{-1}\|_{B(Y)}
\\ \lesssim & \sum_{k=1,2}\left(\|S^2_k[\pd_c\varphi_c]\|_{B(Y)}
+\|S^2_k[\varphi_c']\|_{B(Y)}\right)
\\ \lesssim & \|\tc\|_Y \lesssim  \bM_1(T)\la s\ra^{-1/4}\,,
  \end{split}
\end{equation}
\begin{equation}
\label{eq:bS3-bound}
  \begin{split}
& \|\bS_3\|_{B(Y)}
\lesssim \sum_{k=1,2}\|S^3_k[\psi]\|_{B(Y)}\lesssim e^{-a(4s+L)}\,,
\\ &
\|\bS_3\|_{B(Y_1)}\lesssim \sum_{k=1,2}\|S^3_k[\psi]\|_{B(Y_1)}\lesssim e^{-a(4s+L)}\,.
  \end{split}  
\end{equation}
By Claims~\ref{cl:S4} and \ref{cl:(1+wT2)^{-1}},
\begin{align*}
\|\bS_4\|_{B(Y,Y_1)}  \lesssim &
\sum_{k=1,2}
\left(\|S^3_k[\psi]((\sqrt{2/c}-1)\cdot)\|_{B(Y,Y_1)}+
\|S^4_k[\psi](c^{-1/2}\cdot)\|_{B(Y,Y_1)}\right)
\\ & +\sum_{k=1,2}
\left\|(S^3_k[\psi']+S^4_k[\psi'])((\sqrt{c}-\sqrt{2})\cdot)\right\|_{B(Y,Y_1)}
\\ \lesssim & 
\sum_{k=1,2}(\|S^3_k[\psi]\|_{B(Y_1)}+\|S^3_k[\psi']\|_{B(Y_1)})\|\tc\|_Y
\\ & +
\sum_{k=1,2}
\left(\|S^4_k[\psi]\|_{B(Y,Y_1)}
+\|S^4_k[\psi']\|_{B(Y,Y_1)}\|\tc\|_{L^\infty}\right)\,.  
\end{align*}
Thus we have
\begin{equation}
 \label{eq:bS4-bound}
\|\bS_4\|_{B(Y,Y_1)}  \lesssim \bM_1(T)\la s\ra^{-1/4}e^{-a(4s+R)}\,.
\end{equation}
By Claims~\ref{cl:S5-S6} and \ref{cl:(1+wT2)^{-1}},
\begin{equation}
 \label{eq:bS5-bound}
\|\bS_5\|_{B(Y,Y_1)}\lesssim\sum_{k=1,2}
\left(\|S^5_k\|_{B(Y,Y_1)}+\|S^6_k\|_{B(Y,Y_1)}\right)
\lesssim  \bM_2(T)\la s\ra^{-3/4}\,.
\end{equation}
Obviously,
\begin{equation}
  \label{eq:pdy-bound}
  \|\pd_y\|_{B(Y)}+\|\pd_y\|_{B(Y_1)}\lesssim \eta_0\,.
\end{equation}
By \eqref{eq:def-B3}, \eqref{eq:T-bound1},
\eqref{eq:bS1-bound}--\eqref{eq:pdy-bound} and the fact that $Y_1\subset Y$,
\begin{align*}
\|B_3-B_1\|_{B(Y)} \le & \|\wC_1\|_{B(Y)}+\eta_0^2\sum_{j=1,2}\|\bS_j\|_{B(Y)}
+\sum_{j=3,4,5}\|\bS_j\|_{B(Y)}
\\  \lesssim & \bM_1(T)+\bM_2(T)+\eta_0^2+e^{-aL}\,.  
\end{align*}
Since $B_1$ is invertible, we see that $\|B_3^{-1}\|_{B(Y)}$ is bounded
for $t\in [0,T]$ if $\|B_3-B_1\|_{B(Y)}$ remains small on $[0,T]$.
We can prove the boundedness of $\|B_3^{-1}\|_{B(Y_1)}$ in the same way.
This completes the proof of Claim~\ref{cl:invB_3}.
\par
Using Claims~\ref{cl:S1} and \ref{cl:S3}, we can prove
\begin{gather}
\label{eq:wS1-new}
\|\wS_1\|_{B(Y)}+\|\wS_1\|_{B(Y_1)}\lesssim 1\,,\\
\label{eq:wS3-new}
\|\wS_3\|_{B(Y)}+\|\wS_3\|_{B(Y_1)}\lesssim e^{-a(4t+L)}
\quad\text{for $t\ge0$,}
\end{gather}
in the same way as \eqref{eq:bS1-bound} and \eqref{eq:bS3-bound}.
Claim~\ref{cl:invB_4} immediately follows from 
\eqref{eq:pdy-bound}, \eqref{eq:wS1-new} and \eqref{eq:wS3-new}.
Thus we complete the proof.
\end{proof}

\begin{proof}[Proof of Claim~\ref{cl:[B3,pdy]}]
In view of \eqref{eq:def-B3},
$$[\pd_y,B_3]=[\pd_y,\wC_1]+\sum_{j=1,2}\pd_y^2[\pd_y,\bS_j]-
\sum_{j=3,4,5}[\pd_y,\bS_j]\,,$$
Now we will estimate each term of the right hand side.
By Claim~\ref{cl:T-pdy} and the definition of $\wC_k$,
\begin{equation}
  \label{eq:[pdy,T]}
\|[\pd_y,\wC_k]\|_{B(Y,Y_1)}\lesssim \bM_1(T)\la s\ra^{-3/4}
\enskip\text{for $k=1$, $2$ and $s\in[0,T]$.}
\end{equation}
\par

Since $[\pd_y,\wS_1]=0$ by \eqref{eq:S1,pd_y}, 
we have $[\pd_y,\bS_1] =\bS_1[\wC_2,\pd_y](1+\wC_2)^{-1}$. Thus by
Claim~\ref{cl:(1+wT2)^{-1}},  \eqref{eq:bS1-bound} and \eqref{eq:[pdy,T]},
\begin{equation}
  \label{eq:[bS1,pdy]}
\begin{split}
\|[\pd_y,\bS_1]\|_{B(Y,Y_1)} 
\lesssim & \|\bS_1\|_{B(Y_1)}\|[\wC_2,\pd_y]\|_{B(Y,Y_1)}\|(1+\wC_2)^{-1}\|_{B(Y)}
\\ \lesssim & 
\bM_1(T)\la s\ra^{-3/4}\quad\text{for $s\in[0,T]$.}
\end{split}  
\end{equation}
Applying Claims~\ref{cl:(1+wT2)^{-1}}, \ref{cl:S2}, \ref{cl:T-pdy}
and \eqref{eq:wS2bound} to
$[\pd_y,\bS_2]=\{[\pd_y,\wS_2]+\bS_2[\wC_2,\pd_y]\}(I+\wC_2)^{-1}$, we obtain
\begin{equation}
  \label{eq:[pdy,wS2]}
  \begin{split}
& \|[\pd_y,\bS_2]\|_{B(Y,Y_1)}\lesssim 
\|[\pd_y,\wS_2]\|_{B(Y,Y_1)}+\|\bS_2\|_{B(Y,Y_1)}\|[\pd_y,\wC_2\|_{B(Y)}
\\ \lesssim &
\sum_{k=1,2}\left(\left\|\left[\pd_y, S^2_k[\pd_c\varphi_c]\right]\right\|_{B(Y,Y_1)}
+\left\|\left[\pd_y, S^2_k[\varphi_c']\right]\right\|_{B(Y,Y_1)}\right)
\\ & + \|\bS_2\|_{B(Y,Y_1)}\|[\pd_y,\cC_2]\|_{B(Y)}
\\\lesssim & \|c_y\|_Y+\|\tc\|_Y\|c_y\|_{L^\infty} \lesssim \bM_1(T)\la s\ra^{-3/4}
\quad\text{for $s\in[0,T]$.}
  \end{split}
\end{equation}
Since $[\pd_y,S^3_k]=0$ by \eqref{eq:[pd,s31]}, we have
$[\pd_y,\bS_3]=\bS_3[\wC_2,\pd_y](I+\wC_2)^{-1}$.
Hence it follows from Claims~\ref{cl:(1+wT2)^{-1}}, \ref{cl:S3},
\eqref{eq:bS3-bound}  and \eqref{eq:[pdy,T]} that
\begin{equation}
  \label{eq:[bS3,pdy]}
\|[\pd_y,\bS_3]\|_{B(Y,Y_1)}\lesssim \|\bS_3\|_{B(Y_1)}\|c_y\|_Y \lesssim
  \bM_1(T)\la s\ra^{-3/4}e^{-a(4s+L)}\,.
\end{equation}
By \eqref{eq:bS4-bound}, \eqref{eq:bS5-bound} and \eqref{eq:pdy-bound},
we have for $s\in[0,T]$, 
\begin{gather}
  \label{eq:[bS4,pdy]}
\|[\pd_y,\bS_4]\|_{B(Y,Y_1)}\lesssim 
\eta_0\bM_1(T)\la s\ra^{-1/4}e^{-a(4s+L)}\,,\\
  \label{eq:[bS5,pdy]}
\|[\pd_y,\bS_5]\|_{B(Y,Y_1)}\lesssim \eta_0\bM_2(T)\la s\ra^{-3/4}\,.
\end{gather}
Combining \eqref{eq:[pdy,T]}--\eqref{eq:[bS5,pdy]},
we obtain Claim~\ref{cl:[B3,pdy]}.
Thus we complete the proof.
\end{proof}

\section{Estimates of $R^k$}
\label{sec:Rk}
\begin{claim}
  \label{cl:R1-R2}
There exist positive constants $\delta$ and $C$ such that if
$\bM_1(T)\le \delta$, then for $t\in[0,T]$,
\begin{equation*}
\|R^2_k(t,\cdot)\|_{Y_1}\le CM_1(T)^2\la t\ra^{-1}\,,\quad
\|\pd_yR^2_k(t,\cdot)\|_{Y_1}\le CM_1(T)^2\la t\ra^{-5/4}\,.
\end{equation*}
\end{claim}
\begin{proof}
By Claims~\ref{cl:S1}, \ref{cl:S2} and \eqref{eq:Y1-L1},
\begin{align*}
\|R^2_k\|_{Y_1}\lesssim & 
\|\tc\|_Y(\|x_{yy}\|_Y+\|c_{yy}\|_Y)+\|c_y\|_Y^2(1+\|\tc\|_{L^\infty})
\\ \lesssim & \bM_1(T)^2\la s\ra^{-1}\,.
\end{align*}
We can estimate $\|\pd_yR^2_k\|_{Y_1}$ in the same way.
Thus we complete the proof.
\end{proof}

\begin{claim}
  \label{cl:R3}
There exist  positive constants $\delta$ and $C$ such that
if $\bM_1(T)\le \delta$, then for $t\in[0,T]$,
$\|R^3_k(t,\cdot)\|_{Y_1}\le C\la t\ra^{-1/2}e^{-a(4t+L)}\bM_1(T)^2$.
\end{claim}
\begin{proof}
We decompose $R^3_k$ into three terms. Let
\begin{align*}
R^3_{k1}=& \frac{1}{2\pi}\int_{-\eta_0}^{\eta_0}\int_{\R^2}
\left\{\pd_z^3\tpsi_{c(t,y_1)}(z)
-3c_{yy}(t,y_1)\int_z^\infty\pd_c\tpsi_{c(t,y_1)}(z_1)\,dz_1 \right\}
\\ & \qquad\qquad \times
\overline{g_k^*(z,\eta,c(t,y_1))}e^{i(y-y_1)\eta}\,dy_1dzd\eta
 \\ & -\frac{1}{2\pi}\int_{-\eta_0}^{\eta_0}\int_{\R^2}
\left\{\tc(t,y_1)\psi'''(z+4t+L)
-3c_{yy}(t,y_1)\int_{z+4t+L}^\infty \psi(z_1)\,dz_1 \right\}
\\ & \qquad\qquad \times
\overline{g_k^*(z,\eta)}e^{i(y-y_1)\eta}\,dy_1dzd\eta\,,
\end{align*}
\begin{align*}
R^3_{k2}= -\frac{3}{2\pi} \int_{-\eta_0}^{\eta_0}\int_{\R^2}&
\{\pd_z(\tpsi_{c(t,y_1)}^2)+x_{yy}\tpsi_{c(t,y_1)}
+3c_y(t,y_1)^2\int_z^\infty\pd_c^2\tpsi_{c(t,y_1)}\}
\\ & \times 
\overline{g_k^*(z,\eta,c(t,y_1))}e^{i(y-y_1)\eta}\,dy_1dzd\eta\,,
\end{align*}
and
\begin{align*}
R^3_{k3}=& -\frac{3}{\pi} \int_{-\eta_0}^{\eta_0}\int_{\R^2}
\left(\sqrt{2c(t,y_1)}-2-\tc(t,y_1)\right)\psi(z+4t+L)\varphi_{c(t,y_1)}(z)
\\ & \qquad\qquad
 \times\overline{\pd_zg_k^*(z,\eta,c(t,y_1))}e^{i(y-y_1)\eta}dzdyd\eta
\\ &-\frac{3}{\pi}\int_{-\eta_0}^{\eta_0}\int_{\R^2} \tc(t,y_1)^2\psi(z+4t+L)
\overline{g_{k5}^*(z,\eta,c(t,y_1))}e^{i(y-y_1)\eta}dzdyd\eta\,,    
\end{align*}
where
$g_{k5}(z,\eta,c)^*=\{\varphi_c(z)\pd_zg_k^*(z,\eta,c)
-\varphi(z)\pd_zg_k^*(z,\eta)\}/\tc$.
Then $R^3_k=\sum_{i=1}^3R^3_{ki}$.
\par
Let us estimate $R^3_{k1}$ by using Claims~\ref{cl:S3} and \ref{cl:S4}.
Since $\tpsi_c(z)=(2\sqrt{2c}-2)\psi(z+4t+L)$, we have
\begin{align*}
R^3_{k1}=& S_k^3[\psi''']\left(2\sqrt{2c}-4-\tc\right)
+S_k^4[\psi''']\left(2\sqrt{2c}-4\right)
\\ & 
+3S^3_k[\pd_z^{-1}\psi]\left(((2/c)^{1/2}-1)c_{yy}\right)
+3S^4_k[\pd_z^{-1}\psi]\left((2/c)^{1/2}c_{yy}\right)\,,
\end{align*}
\begin{align*}
 R^3_{k2}=& -24(S_k^3+S_k^4)[(\psi^2)']((\sqrt{c}-\sqrt{2})^2)
-6\sqrt{2}(S^3_k+S^4_k)[\psi]((\sqrt{c}-\sqrt{2})x_{yy})
\\ & -\frac{3\sqrt{2}}{2}(S_k^3+S_k^4)[\pd_z^{-1}\psi](c^{-3/2}(c_y)^2)\,.
\end{align*}
Since  $2\sqrt{2c}-4=\tc+O(\tc^2)$ and $(2/c)^{1/2}-1=O(\tc)$
and $\wP_1L^1\subset Y_1$, it follows from Claim~\ref{cl:S3} that
$$\|S^3_k[\psi'''](2\sqrt{2c}-4-\tc)\|_{Y_1}
\lesssim e^{-a(4t+L)}\|\tc\|_Y^2\,,$$
$$\left\|\sqrt{2}S^3_k[\pd_z^{-1}\psi]\left((2/c)^{1/2}-1)c_{yy}\right)\right\|_{Y_1}
\lesssim e^{-a(4t+L)}\|\tc\|_Y\|c_{yy}\|_Y\,.$$
By Claim~\ref{cl:S4},
\begin{align*}
\|S^4_k[\psi'''](\sqrt{c}-\sqrt{2})\|_{Y_1}
+\|S^4_k[\pd_z^{-1}\psi](c^{-1/2}c_{yy})\|_{Y_1}
\lesssim & e^{-a(4t+L)} \|\tc\|_Y(\|\tc\|_Y+\|c_{yy}\|_Y)\,.
\end{align*}
Thus we prove
$\|R^3_{k1}\|_{Y_1}\lesssim e^{-a(4t+L)}  \|\tc\|_Y(\|\tc\|_Y+\|c_{yy}\|_Y)$.
Similarly, we have
$$\|R^3_{k2}\|_{Y_1} \lesssim e^{-a(4t+L)}
(\|\tc\|_Y^2+\|\tc\|_Y\|x_{yy}\|_Y+\|c_y\|_Y^2)\,,\quad
\|R^3_{k3}\|_{Y_1}\lesssim  e^{-a(4t+L)}\|\tc\|_Y^2\,.$$
Thus we complete the proof.
\end{proof}

\begin{claim}
\label{cl:wA1-bound}
There exists a positive constant $C$ such that
\begin{equation*}
\|\widetilde{\mathcal{A}}_1(t)\|_{B(Y)}+\|\widetilde{\mathcal{A}}_1(t)\|_{B(Y_1)}
\le Ce^{-a(4t+L)}
\quad\text{for every $t\ge0$ and $L\ge0$.}  
\end{equation*}
\end{claim}
\begin{proof}
In view of \eqref{eq:def-ak},
\begin{align*}
\ta_k(t,D_y)\tc=& S^3_k[\psi'''](\tc)+3S^3_k[\pd_z^{-1}\psi](c_{yy})
\\ &-6\mF_\eta^{-1}\left\{\int\varphi(z)\psi(z+4t+L)
\overline{\pd_zg_k^*(z,\eta)}dz (\mF_y\tc)(t,\eta)\right\}\,.
\end{align*}
Hence it follows from Claim~\ref{cl:S3} and \eqref{eq:s311a} that
$$\|\ta_k(t,D_y)\|_{B(Y)}+\|\ta_k(t,D_y)\|_{B(Y_1)}\lesssim e^{-a(4t+L)}\,.$$
Thus we complete the proof of Claim~\ref{cl:wA1-bound}.
\end{proof}

\begin{claim}
\label{cl:akbound}
  There exist  positive constants $C$ and $L_0$ such that
if $L\ge L_0$, then
$$\|A_1(t)\|_{B(Y)}\le Ce^{-a(4t+L)}
\quad\text{for every $t\ge0$.}$$
\end{claim}
\begin{proof}

Since $B_1$ is invertible and $\|\wS_3\|_{B(Y)}
\lesssim \sum_{k=1,\,2}\|S^3_k[\psi]\|_{B(Y)}\lesssim e^{-a(4t+L)}$,
we have Claim~\ref{cl:akbound}.
\end{proof}

\begin{claim}
  \label{cl:R4-R5}
Suppose $a\in(0,1)$ and $\bM_1(T)\le \delta$ If $\delta$ is sufficiently small,
then there exists a positive constant $C$ such that
\begin{align}
  \label{eq:R4k}
& \|R^4_k(t)\|_{Y_1}\le C(\bM_1(T)+\bM_2(T))\bM_2(T)\la t\ra^{-3/2}\,,
\\ \label{eq:R5k}
& \|R^5_k(t)\|_{Y_1}\le  CM_1(T)\bM_2(T)\la t\ra^{-1}\,,
\\ &  \label{eq:R6'}
\|R^6_k\|_{Y_1}\le Ce^{-a(4t+L)}\la t\ra^{-1}\bM_1(T)\bM_2(T)\,,
\\ & \label{eq:R5Y}
\|R^5_k(t)\|_Y\le  CM_1(T)\bM_2(T)\la t\ra^{-5/4}\,,
\,.
\end{align}
\end{claim}
\begin{proof}
By Lemma~\ref{cl:L*g} and \eqref{eq:orth}, we can rewrite
$II^1_k$ as $II^1_k= i\eta II^1_{k1} +II^1_{k2}+II^1_{k3}$, where
\begin{align*}
& II^1_{k1}(t,\eta)=-6\int_\R c_y(t,y)h_{1k}(t,y,\eta)e^{-iy\eta}dy\,,\\
& II^1_{k2}(t,\eta)=3\int_\R c_{yy}(t,y)h_{1k}(t,y,\eta)e^{-iy\eta}dy\,,\\
& II^1_{k3}(t,\eta)=3\int_\R (c_y(t,y))^2h_{2k}(t,y,\eta)e^{-iy\eta}dy\,,
\\ &
h_{jk}(t,y,\eta)=\int_\R v(t,z,y)
\left(\int_{-\infty}^z\overline{\pd_c^jg_k^*(z_1,\eta,c(t,y))}dz_1\right)dz
\quad\text{for $j=1$, $2$.}
\end{align*}
\par
First, we will estimate $II^1_k(t,\cdot)$. Since 
$$\sup_{-\eta_0\le \eta\le \eta_0\,,\, 2-\delta\le c\le 2+\delta}
\left\|e^{-az} \int_{-\infty}^z g_k^*(z_1,\eta,c)dz_1\right\|_{L^2_z}<\infty\,,$$
there exists a positive constant $C$ such that 
$$\sup_{\eta_0\le \eta\le\eta_0}|h_{jk}(t,y,\eta)|\le C\|e^{az}v(t,z,y)\|_{L^2_z}
\quad\text{for any $y\in\R$ and $t\ge0$.}$$
Thus by the Schwarz inequality,
\begin{equation}
  \label{eq:I1k-1}
\begin{split}
\|II^1_{k1}(t,\cdot)\|_{L_\eta^\infty(-\eta_0,\eta_0)} \le &
\|c_y(t)\|_Y\sup_{\eta\in[-\eta_0,\eta_0]}
\left(\int_\R |h_{1k}(t,y,\eta)|^2dy\right)^{1/2}
\\ \lesssim &  \|c_y(t)\|_Y\|v(t)\|_X\,.
\end{split}  
\end{equation}
We can prove
\begin{gather}
\label{eq:I1k-2}
\|II^1_{k2}(t,\eta)\|_{L_\eta^\infty[-\eta_0,\eta_0]}\lesssim
\|c_{yy}\|_Y\|v(t,\cdot)\|_X\,,
\\ \label{eq:I1k-3}
\|II^1_{k3}(t,\eta)\|_{L_\eta^\infty[-\eta_0,\eta_0]}\lesssim
\|c_y\|_{L^4(\R)}^2\|v(t,\cdot)\|_X\,,
\end{gather}
in exactly the same way.
\par

Next, we will estimate $II^2_k$ and $II^3_k$.
Since 
\begin{gather}
\sup _{c\in[2-\delta,2+\delta]\,,\eta\in[-\eta_0,\eta_0]}
\|e^{-2az}g_k^*(z,\eta,c)\|_{L^\infty_z}<\infty\,,\\
\label{eq:gk*bd}
\sup_{c\in[2-\delta,2+\delta]\,,\,\eta\in [-\eta_0,\eta_0]} \left(
\|e^{-az}g_k^*\|_{L^2}+\|e^{-az}\pd_zg_k^*\|_{L^2}+\|e^{-az}\pd_cg_k^*\|_{L^2}\right)
<\infty\,,    
\end{gather}
we have
\begin{align*}
\|II^2_k\|_{L^\infty[-\eta_0,\eta_0]} =&
3\sup_{\eta\in[-\eta_0,\eta_0]}\left|\int_{\R^2} v(t,z,y)^2 
\pd_z\overline{g_k^*(z,\eta,c(t,y))}e^{-iy\eta}dzdy\right|
\\ \lesssim & \|v\|_X^2\sup _{c\in[2-\delta,2+\delta]\,,\eta\in[-\eta_0,\eta_0]}
\|e^{-2az}g_k^*(z,\eta,c)\|_{L^\infty_z}
\\ \lesssim & \|v\|_X^2\,.
\end{align*}
and
$$\|II^3_{k1}\|_{L^\infty_\eta[-\eta_0,\eta_0]}\lesssim \|v\|_X\|x_{yy}\|_Y\,,
\quad
\|R^5_k\|_{Y_1}\lesssim \|II^3_{k2}\|_{L^\infty_\eta[-\eta_0,\eta_0]}
\lesssim \|v\|_X\|x_y\|_Y\,.$$
Combining the above, we have
\begin{align*}
 \|R^4_k(t)\|_{Y_1} \lesssim &
\sup_{-\eta_0\le \eta\le \eta_0}(|II^1_k(t,\eta)|+|II^2_k(t,\eta)|+|II^3_{k1}(t,\eta)|)
\\ \lesssim & \|v(t,\cdot)\|_X
(\|c_y(t)\|_Y+\|c_{yy}(t)\|_Y+\|c_y(t)\|_{L^4}^2
\\ & \qquad +\|x_{yy}\|_Y)
+\|v(t,\cdot)\|_X^2\,,
\end{align*}    
which implies \eqref{eq:R4k}.
\par
By the Schwarz inequality and \eqref{eq:psinorm},
\begin{align*}
\|R^6_k\|_{Y_1}\lesssim & \sup_{|\eta|\le \eta_0}
\left|\int_{\R^2}v(t,x,y)\tpsi_{c(t,y)}
\overline{\pd_zg_k^*(z,\eta,c(t,y))}e^{-iy\eta}dzdy\right|
\\ \lesssim & \|v(t)\|_X\|\tpsi_{c(t,y)}\|_X
\sup _{c\in[2-\delta,2+\delta]\,,\eta\in[-\eta_0,\eta_0]}
\|e^{-2az}g_k^*(z,\eta,c)\|_{L^\infty_z}
\\ \lesssim & e^{-a(4t+L)}\|\tc(t)\|_{L^2(\R)}\|v(t)\|_X\,.
\end{align*}
\par
Finally, we will estimate  $\|R^5_k\|_Y$.
Let
\begin{align*}
II^3_{k21}= & 6\int_{\R^2}v(t,z,y)x_y(t,y)
\overline{g_k^*(z,\eta)}e^{-iy\eta}\,dzdy
\\ =& 6\sqrt{2\pi}\int_\R 
\mF_y\left(x_y(t,\cdot)v(t,z,\cdot)\right)(\eta)\overline{g_k^*(z,\eta)}\,dz\,,  
\end{align*}
\begin{equation*}
II^3_{k22}= 6\int_{\R^2}v(t,z,y)x_y(t,y)\tc(t,y)
\overline{g_{k3}^*(z,\eta,c(t,y))}e^{-iy\eta}\,dzdy\,.
\end{equation*}
Then $II^3_{k2}=II^3_{k21}+II^3_{k22}$.
By the Schwarz inequality,
\begin{align*}
\|II^3_{k21}| \lesssim &
\|e^{-az}g_k^*(z,\eta)\|_{L^2_z}\left(\int_\R
\left|\mF_y\left(x_y(t,\cdot)v(t,z,\cdot)\right)(\eta)\right|^2\,dz\right)^{1/2}\,.
\end{align*}
By \eqref{eq:gk*bd} and Plancherel's theorem,
\begin{align*}
\|II^3_{k21}\|_{L^2[-\eta_0,\eta_0]} \lesssim &
\left(\int_{-\eta_0}^{\eta_0}\int_\R
\left|\mF_y\left(x_y(t,\cdot)e^{az}v(t,z,\cdot)\right)(\eta)\right|^2
\,dzd\eta\right)^{1/2}
\\ \lesssim & \|x_y(t)v(t)\|_X\,.
\end{align*}
Since $\|x_y\|_{L^\infty}\lesssim \|x_y\|_Y^{1/2}\|x_{yy}\|_Y^{1/2}
\lesssim \bM_1(T)\la t\ra^{-1/2}$, we have
$$\|II^3_{k21}\|_{L^2[-\eta_0,\eta_0]}\lesssim
\bM_1(T)\bM_2(T)\la t\ra^{-5/4}\,,$$
By the Schwarz inequality,
\begin{align*}
\|II^3_{k22}\|_{L^\infty[-\eta_0,\eta_0]}
\lesssim & \|v(t)\|_X\|x_y(t)\tc(t)\|_{L^2_y}
\sup_{\eta\in[-\eta_0,\eta_0]\,,\, c\in[2-\delta,2+\delta]}\|e^{-az}g_{k3}^*(z,\eta,c)\|_{L^2_z}
\\ \lesssim & \|v(t)\|_X\|\tc(t)\|_{L^\infty}\|x_y(t)\|_Y
\lesssim \bM_1(T)^2\bM_2(T)\la t\ra^{-3/2}\,.
\end{align*}
Combining the above, we have for $t\in[0,T]$,
$$\|R^5_k\|_Y\lesssim \|II^3_{k21}\|_{L^2[-\eta_0,\eta_0]}
+\|II^3_{k22}\|_{L^\infty[-\eta_0,\eta_0]}
\lesssim \bM_1(T)\bM_2(T)\la t\ra^{-5/4}\,.$$
Thus we complete the proof.
\end{proof}

To estimate $R^7_k$, we need the following.
\begin{claim}
  \label{cl:b-capprox}
There exist positive constants $\delta$ and $C$ such that
if $\sup_{t\in[0,T]}\|\tc(t)\|_Y \le \delta$, then for $t\in[0,T]$,
\begin{align}
\label{b-cap1}
& \|b-\tilde{c}\|_Y \le C\|\tilde{c}\|_{L^\infty}\|\tilde{c}\|_Y\,,\quad
\|b-\tilde{c}\|_{Y_1}\le C\|\tilde{c}\|_Y^2\,,\\
& \label{b-cap2} 
\|b_y-c_y\|_Y\le C \|\tilde{c}\|_{L^\infty}\|c_y\|_Y\,,\quad
\|b_y-c_y\|_{Y_1}\le C \|\tilde{c}\|_Y\|c_y\|_Y\,,\\
& \label{b-cap2a} 
\|b_t-c_t\|_Y\le C \|\tilde{c}\|_{L^\infty}\|c_t\|_Y\,,\\
\label{b-cap2b}
& \|b_{yy}-c_{yy}\|_Y\le C(\|\tilde{c}\|_{L^\infty}\|c_{yy}\|_Y
+\|c_y\|_{L^\infty}\|c_y\|_Y)\,,\\
\label{b-cap2c}
& \|b_{yy}-c_{yy}\|_{Y_1}\le C(\|\tilde{c}\|_Y\|c_{yy}\|_Y+\|c_y\|_Y^2)\,,  
\\
\label{b-cap3}
& \left\|\left(\frac{c}{2}\right)^{3/2}-1
-\frac{3}{4}b\right\|_{L^2}
\le C\|\tilde{c}\|_{L^\infty}\|\tilde{c}\|_Y\,,
\\
\label{b-cap4}
& \left\|b-\tc-\frac{1}{8}\wP_1\tc^2\right\|_Y
\le C\|\tilde{c}\|_{L^\infty}^2\|\tilde{c}\|_Y\,.
\end{align}
\end{claim}
\begin{proof}
By \eqref{eq:bdef},
$$b-\tc= \frac{4}{3}\wP_1
\left\{\left(\frac{c}{2}\right)^{3/2}-1-\frac{3}{4}\tc\right\}\,,$$
\begin{align*}
& b_y-c_y=\wP_1\{(c/2)^{1/2}-1\}c_y\,,\quad b_t-c_t=\wP_1\{(c/2)^{1/2}-1\}c_t\,,\\
& b_{yy}-c_{yy}=\wP_1\{(c/2)^{1/2}-1\}c_{yy}+\frac14\wP_1(c/2)^{-1/2}(c_y)^2\,.
\end{align*}
Using the fact that $(c/2)^{3/2}-1-3\tc/4-3\tc^2/32=O(\tc^3)$,
we can prove \eqref{b-cap1}--\eqref{b-cap2c} and \eqref{b-cap4}
in the same way as the proof of Claim~\ref{cl:T1-T2}. 
\par
Finally, we will show \eqref{b-cap3}.
Let $\wP_2=I-\wP_1$. Since $\wP_2\tilde{c}=0$ and
\begin{align*}
\left(\frac{c}{2}\right)^{3/2}-1-\frac{3b}{4}=& \wP_2
\left\{\left(\frac{c}{2}\right)^{3/2}-1\right\}\,, 
\end{align*}
we have
\begin{align*}
\left\|\left(\frac{c}{2}\right)^{3/2}-1-\frac{3b}{4}\right\|_{L^2}
=& \left\|\wP_2\left\{
\left(\frac{c}{2}\right)^{3/2}-1-\frac{3}{4}\tilde{c}\right\}\right\|_{L^2}
\lesssim  \|\tc\|_{L^4}^2\,.
\end{align*}
Thus we complete the proof.
\end{proof}

\begin{claim}
\label{cl:R6}
There exist positive constants $\delta$ and  $C$ such that
if $\sup_{t\in[0,T]}\|\tc(t)\|_Y \le \delta$, then for $t\in[0,T]$,
\begin{align}
\label{eq:R6-est1}
& \|\wP_1R^7_1(s)\|_{Y_1}\le C\bM_1(T)^2\la s \ra^{-5/4}\,,\quad
\|\wP_1R^7_2(s)\|_{Y_1}\le C\bM_1(T)^2\la s \ra^{-1}\,,\\
\label{eq:R6-est2}
& \|\wP_1R^7_1(s)\|_Y\le C\bM_1(T)^2\la s \ra^{-3/2}\,,
\quad \|\wP_1R^7_2(s)\|_Y\le C\bM_1(T)^2\la s \ra^{-5/4}\,,\\
\label{eq:R6-est3}
& \|\wP_1\pd_yR^7_2(s)\|_{Y_1}\le C\bM_1(T)^2\la s \ra^{-5/4}\,,
\quad \|\wP_1\pd_yR^7_2(s)\|_Y\le C\bM_1(T)^2\la s \ra^{-3/2}\,.
\end{align}
\end{claim}
\begin{proof}
In view of \eqref{eq:R6-def} and \eqref{eq:Y1-L1},
\begin{align*}
\|\wP_1R^7_1\|_{Y_1}\lesssim &
\|x_{yy}\|_Y\left\|\left(\frac{c}{2}\right)^{3/2}-1-\frac{3b}4\right\|_{L^2}
\\ & + \|b_{yy}-c_{yy}\|_{Y_1}+(\|b_y-c_y\|_Y+\|\tc\|_{L^\infty}\|c_y\|_Y)\|x_y\|_Y
+\|c_y\|_Y^2\,,
\end{align*}
$$\|\wP_1R^7_2\|_{Y_1}\lesssim 
\|\tc\|_Y\|x_{yy}\|_Y+\|c_y\|_Y\|x_y\|_Y+\|b_{yy}-c_{yy}\|_{Y_1}+\|c_y\|_Y^2\,.
$$
Combining the above with Claim~\ref{cl:b-capprox}, we have
\eqref{eq:R6-est1}.
We can obtain \eqref{eq:R6-est2} and \eqref{eq:R6-est3} in the same way.
Thus we complete the proof.
\end{proof}
\section{Local well-posedness in exponentially weighted space}
\label{ap:LWP}
The $L^2$ well-posedness of the KP-II equation around line solitons
has been proved by Molinet, Saut and Tzvetkov (\cite{MST})
by using Bourgain's norm. 
In this section, we will explain well-posedness for exponentially
localized initial data around a line soliton.

\par
Let $u(t,x,y)=\varphi(x-4t)+\tv(t,x-4t,y)$
be a solution to \eqref{KPII_integrated}.
Then
\begin{equation}
  \label{eq:v-fix}
  \pd_t\tv=\mL \tv-3\pd_x(\tv^2)\,.
\end{equation}
\begin{proposition}
\label{prop:LWPX}
Suppose $a>0$ and $v_0\in X\cap L^2(\R^2)$. If $\tv(0)=v_0$,
then there exists a unique solution 
of \eqref{eq:v-fix} such that for any $T>0$, 
\begin{equation}
\label{eq:sup-X}
\tv\in L^\infty(0,T;X)\cap X_T\,,
\end{equation}
where $X_T$ is the auxiliary Banach space used in Theorem 1.1 of \cite{MST}.
If $v_0\in H^1(\R^2)$ in addition, then $\tv(t)\in C([0,\infty);X)$.
\end{proposition}
\begin{remark}
\label{rem:LWP}
The Banach space $X_T$ is continuously imbedded into $C([0,T];L^2(\R^2))$.
Moreover \cite[Lemma~4.1]{MST} tells us that
\begin{equation}
\label{eq:MST1} \tv(t)\in C([0,T);L^2(\R^2))\cap L^4(0,T;L^4(\R^2))\,,
\end{equation}
and that $\tv(t)\in C([0,T);H^s(\R^2))$ if $\tv(0)\in H^s(\R^2)$ for an $s\ge0$.
\end{remark}

\begin{proof}[Proof of Proposition~\ref{prop:LWPX}]
To prove $\tv(t)\in L^\infty(0,T;X)$ for any $T>0$, we will use
the virial identity for the KP-II equation (\cite{dBM}).
Let $a>0$, $\bar{p}(x)=1+\tanh ax$  and $p_n(x)=
e^{2an}\bar{p}(x-n)$ for $n\in\N$. Suppose $v_0\in H^3(\R^2)
\cap X$ and $\pd_x^{-1}v_0\in H^2(\R^2)$. Then by \cite{MST}, we have
$\tv(t)\in C([0,\infty);H^3)$ and $\pd_x^{-1}v_0\in C([0,\infty);H^2)$.
Multiplying \eqref{eq:v-fix} by $2p_n(x)v(t,x,y)$
and integrating the resulting equation over $\R^2$, we have
  \begin{align*}
& \frac{d}{dt}\int_{\R^2} p_n(x)\tv(t,x,y)^2\,dxdy+
\int_{\R^2}p_n'(x)\left\{3(\pd_x\tv)^2+3(\pd_x^{-1}\pd_y\tv)^2
-4\tv^3\right\}\,dxdy
\\ =& 3\int_{\R^2}\left\{p_n'(x)\varphi(x)
-p_n(x)\varphi'(x)\right\}\tv(t,x,y)^2\,dxdy\,.
  \end{align*}
By  Claim 5.1 in \cite{MT},
$$\left|\int_{\R^2} \bar{p}'_n(x)\tv^3\,dxdy\right|
\le C_1\left(\int_{\R^2}\tv^2\,dxdy\right)^{1/2}
\left(\int_{\R^2}p_n'(x)\mathcal{E}(\tv)\,dxdy\right)^{1/2}\,,$$
where $C_1$ is a constant independent of $n\in \N$.
Hence there exists a positive constant $C$ such that for every $n\in\N$,
\begin{align*}
& \int_{\R^2} p_n(x)\tv(t,x,y)^2\,dxdy
+2\int_0^t\int_{\R^2}\bar{p}'_n(x)\left\{(\pd_x\tv)^2+(\pd_x^{-1}\pd_y\tv)^2\right\}
(s,x,y)\,dxdyds
\\ \le & \int_{\R^2} p_n(x)v_0(x,y)^2\,dxdy+
C\int_0^t\|\tv(s)\|_{L^2}^2\,ds\,.
\end{align*}
By approximating a solution $\tv(t)$ of \eqref{eq:v-fix} with
$\tv(0)=v_0\in X\cap L^2(\R^2)$ by a sequence solutions $\{\tv_k(t)\}$
of \eqref{eq:v-fix} satisfying 
$$\tv_k(0)\in H^3(\R^2)\,,\quad \pd_x^{-1}\tv_k(0)\in H^2(\R^2)\,,\quad 
\lim_{k\to\infty}\|\tv_k(0)-v_0\|_{L^2(\R^2)}=0\,,$$
we have for any $v_0\in L^2(\R^2)$,
\begin{equation*}
\int_{\R^2} p_n(x)\tv(t,x,y)^2\,dxdy
 \le  \int_{\R^2} p_n(x)v_0(x,y)^2\,dxdy+C\int_0^t\|\tv(s)\|_{L^2}^2\,ds\,.
\end{equation*}
Passing to the limit $n\to\infty$, we obtain
\begin{equation}
  \label{eq:virial-fx}
\|\tv(t)\|_X^2 \le \|v_0\|_X^2+C\int_0^t\|\tv(s)\|_{L^2}^2\,ds\,.
\end{equation}
Since $\sup_{t\in[0,T]}\|\tv(t)\|_{L^2(\R^2)}<\infty$ for any $T>0$,
we have \eqref{eq:sup-X}.
\par
Suppose $v_0\in H^1(\R^2)\cap X$. Then we have \eqref{eq:sup-X} and
$\tv(t)\in C(\R;H^1(\R^2))$.
By the variation of constants formula,
$$\tv(t)=e^{t\mL_0}v_0-6\int_0^t e^{(t-s)\mL_0}\pd_x(\varphi\tv(s))\,ds
-6\int_0^t e^{(t-s)\mL_0}\tv(s)\pd_x\tv(s)\,ds\,.$$
Since $\|e^{ax}\tv(s)\pd_x\tv(s)\|_{L^1}\le \|\tv(s)\|_X\|\tv(s)\|_{H^1(\R^2)}$,
we have $\tv(t)\in C([0,\infty);X)$ by using 
\eqref{eq:lemfs1} and \eqref{eq:lemfs4} in Lemma~\ref{lem:free-semigroup}.
Thus we complete the proof.
\end{proof}


\begin{thebibliography}{10}
\bibitem{APS}
J.~C.~Alexander, R.~L.~Pego and R.~L.~Sachs,
{\it On the transverse instability of solitary waves in the
Kadomtsev-Petviashvili equation},
Phys. Lett. A \textbf{226} (1997),  187--192.
%
\bibitem{Benjamin} T.~Benjamin,
 {\it The stability of solitary waves},
 Proc. R. Soc. Lond. A \textbf{328} (1972), 153-183.
%
\bibitem{BL} D.~J. Benney and J.~C. Luke, 
{\it Interactions of permanent waves of  finite amplitude},
J. Math. Phys., \textbf{43} (1964), 309--313.
%
\bibitem{Bona} J.~L.~Bona,
{\it The stability of solitary waves},
Proc. R. Soc. Lond. A \textbf{344} (1975), 363-374.
%
\bibitem{Ca-Li} T.~Cazenave and P.~L.~Lions,
{\it Orbital stability of standing waves for some nonlinear Schr\"odinger equations},
Comm. Math. Phys. \textbf{85} (1982),  549--561. 
%
\bibitem{dBM} A.~de Bouard and Y.~Martel,
{\it Non existence of $L^2$-compact solutions of the Kadomtsev-Petviashvili
II equation}, Math. Ann. \textbf{328} (2004) 525-544.
%
\bibitem{dBS} A.~Bouard and J.~C.~Saut,
\textit{Remarks on the stability of generalized KP solitary waves},
Mathematical problems in the theory of water waves, 75--84, 
Contemp. Math. \textbf{200}, Amer. Math. Soc., Providence, RI, 1996.
%
\bibitem {Bourgain} J.~Bourgain,
{\it On the Cauchy problem for the Kadomtsev-Petviashvili equation},
 GAFA \textbf{3} (1993), 315-341.
%
\bibitem{Cu} S.~Cuccagna,
{\it On asymptotic stability in 3D of kinks for the $\phi^4$ model},
Trans. Amer. Math. Soc. \textbf{360} (2008), 2581-2614. 
%
\bibitem{Fujita} H.~Fujita,
{\it On the blowing up of solutions of the Cauchy problem for} 
$u_t=\Delta u+u^{1+\alpha}$,
J. Fac. Sci. Univ. Tokyo Sect. I \textbf{13} (1966), 109--124. 
%
\bibitem{GSS} M.\ Grillakis, J.\ Shatah and W.~A.\ Strauss,
{\it Stability Theory of solitary waves in the presence of symmetry, I.},
 J. Funct. Anal. \textbf{74} (1987), 160--197.
%
\bibitem{GPS} A.~Gr\"unrock, M.~Panthee and J.~Drumond Silva,
{\it On KP-II equations on cylinders}, Ann. IHP Analyse non lin\'eaire \textbf{26} (2009), 2335-2358.
%
\bibitem{Hadac} M.~Hadac,
{\it Well-posedness of the KP-II equation and generalizations}, Trans. Amer. Math. Soc.
\textbf{360} (2008), 6555-6572.
%
\bibitem{Haragus}
M.~Haragus, {\it Transverse spectral stability of small periodic traveling waves for the KP equation},
Stud. Appl. Math. \textbf{126} (2011),  157--185.
%
\bibitem{HHK} 
M.~Hadac, S.~Herr and H.~Koch, 
{\it Well-posedness and scattering for the KP-II equation in a critical space},
Ann. IHP Analyse non lin\'eaire \textbf{26} (2009), 917-941.
%
\bibitem{IM} P.~Isaza and J.~Mejia,
{\it Local and global Cauchy problems for the Kadomtsev-Petviashvili (KP-II)
equation in Sobolev spaces of negative indices},
Comm. Partial Differential Equations \textbf{26} (2001), 1027-1057.
%
\bibitem{KP} B.~B.~Kadomtsev and V.~I.~Petviashvili, 
{\it On the stability of solitary waves in weakly dispersive media}, 
Sov. Phys. Dokl. \textbf{15} (1970), 539-541.
%
\bibitem{Kapitula}
T.~Kapitula, 
{\it Multidimensional stability of planar traveling waves},
Trans. Amer. Math. Soc. \textbf{349} (1997),  257--269. 
%
\bibitem{KaKo} C.~Y.~Kao and Y.~Kodama,
\textit{Numerical study of the KP equation for non-periodic waves},
Math. Comput. Simulation \textbf{82} (2012), 1185--1218. 
%
\bibitem{Karch} G.~Karch, 
Self-similar large time behavior of solutions to Korteweg-de Vries-Burgers
equation. Nonlinear Anal. Ser. A: Theory Methods \textbf{35} (1999), 199--219.
%
\bibitem{Le-Xin}
C.~D.~Levermore and J.~X.~Xin,
{\it Multidimensional stability of traveling waves in a bistable reaction-diffusion equation, II.}
Comm. Partial Differential Equations \textbf{17} (1992), 1901--1924.
%
\bibitem{LiuW}
Y.~Liu and X.~P.~Wang,
\textit{Nonlinear stability of solitary waves of a generalized
 Kadomtsev-Petviashvili equation}, Comm. Math. Phys. \textbf{183} (1997),
253--266. 
%
\bibitem{MM} Y.~Martel and F.~Merle,
{\it A Liouville theorem for the critical generalized Korteweg-de Vries
equation}, J. Math. Pures Appl. \textbf{79} (2000),  339-425.
%
\bibitem{MV} F.~Merle and L.~Vega,
\textit{$L^2$ stability of solitons for KdV equation},  
Int. Math. Res. Not. (2003), 735-753. 
%
\bibitem{MP} T.~Mizumachi and R.~L.~Pego,
{\it Asymptotic stability of Toda lattice solitons}, Nonlinearity  \textbf{21} (2008),
2099--2111.
%
\bibitem{MT}
T.~Mizumachi and N.~Tzvetkov,
{\it Stability of the line soliton of the KP-II equation under periodic
transverse perturbations},
Mathematische Annalen \textbf{352} (2012), 659--690.
%
\bibitem{MST_contr} L. Molinet, J.-C. Saut and N. Tzvetkov, 
{\it Remarks on the mass constraint for KP-type equations},
SIAM J. Math. Anal. \textbf{39} (2007), 627-641.
%
\bibitem{MST_KPI} L. Molinet, J.-C. Saut and N. Tzvetkov,
\textit{Global well-posedness for the KP-I equation on the background of
a non-localized solution}, Comm. Math. Phys. \textbf{272} (2007), 775--810.
%
\bibitem{MST} L.~Molinet, J.~C.~Saut and N.~Tzvetkov, 
{\it Global well-posedness for the KP-II equation on the background of a
non-localized solution}, Ann. Inst. H. Poincar\'e Anal. Non Lin\'eaire
\textbf{28} (2011), 653--676.
%
\bibitem{Ped} G.~ Pedersen,
{\it Nonlinear modulations of solitary waves}, J. Fluid Mech. \textbf{267} (1994),
83--108. 
%
\bibitem{PQ}
R.~L. Pego and J.~R. Quintero, {\it Two-dimensional solitary waves for a
Benney-Luke equation}, Physica D, \textbf{132} (1999), 476--496.
%
\bibitem{PW} R.~L.~Pego and M.~I.~Weinstein,
{\it Asymptotic stability of solitary waves},
Comm. Math. Phys. \textbf{164} (1994),  305--349.
%
\bibitem{RT1} F.~Rousset and N.~Tzvetkov,
{\it Transverse nonlinear instability for two-dimensional dispersive models },
Ann. IHP, Analyse Non Lin\'eaire \textbf{26} (2009), 477-496.
%
\bibitem{RT2} F.~Rousset and N.~Tzvetkov, {\it Transverse nonlinear instability for some Hamiltonian PDE's},
J. Math. Pures Appl. \textbf{90} (2008), 550-590.
%
\bibitem {Tak} H.~Takaoka,
{\it Global well-posedness for the Kadomtsev-Petviashvili II equation},
Discrete Contin. Dynam. Systems \textbf{6} (2000), 483-499.
%
\bibitem{TT} H.~Takaoka and N.~Tzvetkov,
{\it On the local regularity of Kadomtsev-Petviashvili-II equation},
IMRN \textbf{8} (2001), 77-114.
%
\bibitem{Tz} N.~Tzvetkov,
{\it Global low regularity solutions for Kadomtsev-Petviashvili equation},
Diff. Int. Eq. \textbf{13} (2000), 1289-1320.
%
\bibitem{Ukai} S.~Ukai,
{\it Local solutions of the Kadomtsev-Petviashvili equation},
J. Fac. Sc. Univ. Tokyo Sect. IA Math. \textbf{36} (1989), 193--209.
%
\bibitem{VA} J.~Villarroel and M.~Ablowitz, 
{\it On the initial value problem for the KPII equation with data that do not decay along a line},  
Nonlinearity  \textbf{17}  (2004), 1843-1866.
%
\bibitem{We} M.~Weinstein,
{\it Lyapunov stability of ground states of nonlinear dispersive evolution equations},
Comm.\ Pure Appl.\ Math.\ \textbf{39} (1986),  51--68.
%
\bibitem{Xin} J.~X.~Xin,
{\it Multidimensional stability of traveling waves in a bistable reaction-diffusion equation, I.}
Comm. Partial Differential Equations \textbf{17} (1992), 1889--1899.
%
\bibitem{Z} V.~Zakharov,
{\it Instability and nonlinear oscillations of solitons},
JEPT Lett. \textbf{22}(1975), 172-173.
\end{thebibliography}
\end{document}